\documentclass[english,10pt]{article}
\usepackage{fullpage}
\usepackage[T1]{fontenc}
\usepackage[latin9]{inputenc}
\usepackage{geometry}
\usepackage{array}
\usepackage{verbatim}
\usepackage{float}
\usepackage{amsmath}
\usepackage{amssymb}
\usepackage{graphicx}
\usepackage{esint}
\usepackage{enumitem}

\makeatletter

\providecommand{\tabularnewline}{\\}
\floatstyle{ruled}
\newfloat{algorithm}{tbp}{loa}
\providecommand{\algorithmname}{Algorithm}
\floatname{algorithm}{\protect\algorithmname}

 \usepackage{color}
\usepackage{gnuplottex}

\usepackage{tikz, pgf, pgfplots}
\usetikzlibrary{automata,positioning}
\usetikzlibrary{arrows,backgrounds,plotmarks}

\usepackage{cite}\usepackage{amsthm}\usepackage{dsfont}\usepackage{array}\usepackage{mathrsfs}

\usepackage{hyperref}

\definecolor{ejc}{RGB}{255,0,0}
\definecolor{yxc}{RGB}{0,0,200}

\newcommand{\yxc}[1]{\textcolor{yxc}{[YXC: #1]}}

\makeatother

\usepackage{babel}
\begin{document}
\theoremstyle{plain}\newtheorem{lemma}{\textbf{Lemma}}\newtheorem{theorem}{\textbf{Theorem}}\newtheorem{fact}{\textbf{Fact}}\newtheorem{corollary}{\textbf{Corollary}}\newtheorem{assumption}{\textbf{Assumption}}\newtheorem{example}{\textbf{Example}}\newtheorem{definition}{\textbf{Definition}}\newtheorem{proposition}{\textbf{Proposition}}

\theoremstyle{definition}

\theoremstyle{remark}\newtheorem{remark}{\textbf{Remark}}

\date{May 2015; ~Revised February 2016}

\title{Solving Random Quadratic Systems of Equations\\
  Is Nearly as Easy as Solving Linear Systems}

\author{Yuxin Chen \thanks{Department of Statistics, Stanford
    University, Stanford, CA 94305, U.S.A.}  \and Emmanuel J. Cand\`es
  \footnotemark[1] \thanks{Department of Mathematics, Stanford
    University, Stanford, CA 94305, U.S.A.} }

\maketitle

%{\centering% adjust the vertical skips according to your needs
%  \vspace*{-0.3cm} \textsuperscript{$\ast$} Department of
%  Statistics, Stanford
%  University, Stanford, CA 94305, USA\\
%  \textsuperscript{$\dagger$} Department of Mathematics, Stanford
%  University, Stanford, CA 94305, USA
%  \par\bigskip 
 % \date{May 2015}
%}

\begin{abstract}
  We consider the fundamental problem of solving quadratic systems of
  equations in $n$ variables, where
  $y_i = |\langle \boldsymbol{a}_i, \boldsymbol{x} \rangle|^2$,
  $i = 1, \ldots, m$ and $\boldsymbol{x} \in \mathbb{R}^n$ is
  unknown. We propose a novel method, which starting with an initial
  guess computed by means of a spectral method, proceeds by minimizing
  a nonconvex functional as in the Wirtinger flow approach
  \cite{candes2014wirtinger}. There are several key distinguishing
  features, most notably, a distinct objective functional and novel
  update rules, which operate in an adaptive fashion and drop terms
  bearing too much influence on the search direction.  These careful
  selection rules provide a tighter initial guess, better descent
  directions, and thus enhanced practical performance. On the
  theoretical side, we prove that for certain unstructured models of
  quadratic systems, our algorithms return the correct solution in
  linear time, i.e.~in time proportional to reading the data
  $\{\boldsymbol{a}_i\}$ and $\{y_i\}$ as soon as the ratio $m/n$
  between the number of equations and unknowns exceeds a fixed
  numerical constant.  We extend the theory to deal with noisy systems
  in which we only have
  $y_i \approx |\langle \boldsymbol{a}_i, \boldsymbol{x} \rangle|^2$
  and prove that our algorithms achieve a statistical accuracy, which
  is nearly un-improvable. We complement our theoretical study with
  numerical examples showing that solving random quadratic systems is
  both computationally and statistically not much harder than solving
  linear systems of the same size---hence the title of this paper. For
  instance, we demonstrate empirically that the computational cost of
  our algorithm is about four times that of solving a least-squares
  problem of the same size.

  %\yxc{Do you think we should move the proof of backtracking to
  %  supplemental materials?} \ejc{No, why?}
\end{abstract}

\section{Introduction}
 
\subsection{Problem formulation}

Imagine we are given a set of $m$ quadratic equations taking the form
\begin{equation}
	y_{i} = \left| \langle \boldsymbol{a}_i,\boldsymbol{x} \rangle \right|^2 , \quad\quad i = 1, \cdots, m,
	\label{eq:Noiseless}
\end{equation}
where the data $\boldsymbol{y}= [y_{i}]_{1\leq i\leq m}$ and design
vectors $\boldsymbol{a}_{i} \in\mathbb{R}^n/\mathbb{C}^n$ are known
whereas $\boldsymbol{x}\in\mathbb{R}^n/\mathbb{C}^{n}$ is unknown.  Having
information about
$\left| \langle \boldsymbol{a}_i,\boldsymbol{x} \rangle
\right|^2$---or,
equivalently,
$\left| \langle \boldsymbol{a}_i,\boldsymbol{x} \rangle
\right|$---means
that we a priori know nothing about the phases or signs of the linear
products $\langle \boldsymbol{a}_i,\boldsymbol{x} \rangle$.  The
problem is this: can we hope to identify a solution, if any,
compatible with this nonlinear system of equations?

This problem is combinatorial in nature as one can alternatively pose
it as recovering the missing signs of $\langle
\boldsymbol{a}_i,\boldsymbol{x} \rangle$ from magnitude-only
observations.  As is well known, many classical combinatorial problems
with Boolean variables may be cast as special instances of
(\ref{eq:Noiseless}). 
% 
% Solving this problem in its most general form is, unfortunately, at
% least as hard as any combinatorial problem can be.
As an example, consider the NP-hard {\em stone problem} \cite{ben2001lectures} in which
we have $n$ stones each of weight $w_i >0$ ($1\leq i\leq n$), which we
would like to divide into two groups of equal sum weight.  Letting
$x_i \in \{-1, 1\}$ indicate which of the two groups the $i$th stone
belongs to, one can formulate this problem as solving the following quadratic
system
\begin{align}
	\begin{cases}
	x_i^2 = 1,  \quad\quad i = 1, \cdots, n,  \\
        (w_1 x_1 + \cdots + w_n x_n)^2 = 0.	
	\end{cases}
	\label{eq:stones}
\end{align}
However simple this formulation may seem, even checking whether a
solution to (\ref{eq:stones}) exists or not is known to be NP-hard.

Moving from combinatorial optimization to the physical sciences, one
application of paramount importance is the \emph{phase retrieval}
\cite{gerchberg1972practical,fienup1982phase} problem, which permeates
through a wide spectrum of techniques including X-ray crystallography,
diffraction imaging, microscopy, and even quantum mechanics. In a nutshell, the problem of
phase retrieval arises due to the physical limitation of optical
sensors, which are often only able to record the intensities of the
diffracted waves scattered by an object under study.  Notably, upon
illuminating an object $\boldsymbol{x}$, the diffraction pattern is of
the form of $\boldsymbol{A}\boldsymbol{x}$; however, it is only
possible to obtain intensity measurements $\boldsymbol{y} =
|\boldsymbol{A}\boldsymbol{x}|^2$ leading to the quadratic system
\eqref{eq:Noiseless}.\footnote{Here and below, for $\boldsymbol{z} \in
  \mathbb{C}^n$, $|\boldsymbol{z}|$ (resp.~$|\boldsymbol{z}|^2$) represents
  the vector of magnitudes $(|z_1|, \cdots, |z_n|)^{\top}$ (resp.~squared magnitudes $(|z_1|^2, \cdots, |z_n|^2)^{\top}$).} In the
Fraunhofer regime where data is collected in the far-field zone,
$\boldsymbol{A}$ is given by the spatial Fourier transform.
% of the object transmission function, and hence one demands
% reconstruction of an object from the squared modulus of its linear
% transform. 
We refer to \cite{shechtman2014phase} for in-depth reviews of this
subject.

Continuing this motivating line of thought, in any real-world
application recorded intensities are always corrupted by at least a
small amount of noise so that observed data are only about
$|\langle \boldsymbol{a}_i, \boldsymbol{x} \rangle|^2$; i.e.
\begin{equation}
	y_{i} \approx \left| \langle \boldsymbol{a}_i,\boldsymbol{x} \rangle \right|^2 , \quad\quad i = 1, \cdots, m.
	\label{eq:approx}
\end{equation}
Although we present results for arbitrary noise distributions---even
for non-stochastic noise---we shall pay particular attention to the
Poisson data model, which assumes
\begin{equation}
	y_{i}\stackrel{\text{ind.}}{\sim} \mathsf{Poisson} \big(\left|\left\langle \boldsymbol{a}_i, \boldsymbol{x}\right\rangle \right|^2 \big),
	\quad\quad i = 1, \cdots, m. 
  	\label{eq:Poisson}
\end{equation}
The reason why this statistical model is of special interest is that it
naturally describes the variation in the number of photons detected by
an optical sensor in various imaging applications.  

\subsection{Nonconvex optimization}

Under a stochastic noise model with independent samples, a first
impulse for solving (\ref{eq:approx}) is to seek the maximum
likelihood estimate (MLE), namely,
\begin{eqnarray}
	\text{minimize}_{\boldsymbol{z}} \quad - \sum\nolimits_{i=1}^{m}\ell \left(\boldsymbol{z}; y_i\right),
\end{eqnarray}
where $\ell \left(\boldsymbol{z}; y_i\right)$ denotes the log-likelihood
of a candidate solution $\boldsymbol{z}$ given the outcome $y_{i}$.
For instance, under the Poisson data model (\ref{eq:Poisson})
one can write 
\begin{equation}
	 \ell  (\boldsymbol{z}; y_i) = y_i \log(|\boldsymbol{a}_i^{*}\boldsymbol{z}|^2)
				 - |\boldsymbol{a}_i^{*}\boldsymbol{z}|^2  
	\label{eq:PoissonLL}
\end{equation} 
modulo some constant offset.  Unfortunately, the log-likelihood is
usually not concave, thus making the problem of computing the MLE NP-hard in general.

To alleviate this computational intractability, several convex
surrogates have been proposed that work particularly well when the
design vectors $\{\boldsymbol{a}_{i}\}$ are chosen at random
\cite{candes2012phaselift,candes2012solving, waldspurger2012phase,chen2013exact,shechtman2011sparsity, ohlsson2011compressive, 
  demanet2012stable, li2013sparse,
  cai2015rop, jaganathan2012recovery, chen2014convex, kueng2014low, bahmani2015efficient, gross2014partial, gross2014improved}. The basic idea
is to introduce a rank-one matrix
$\boldsymbol{X}=\boldsymbol{x}\boldsymbol{x}^{*}$ to linearize the
quadratic constraints, and then relax the rank-one constraint. Suppose
we have Poisson data, then this strategy converts the problem into a
convenient convex program:
\[
\begin{array}{ll}
  \text{minimize}_{\boldsymbol{X}} &  \quad \sum_{i=1}^m  (\mu_i - y_i \log \mu_i)  + \lambda\mathsf{Tr}(\boldsymbol{X}) \nonumber\\
  \text{subject to} &  \quad \mu_{i}=\boldsymbol{a}_i^{\top} \boldsymbol{X}\boldsymbol{a}_i,  \quad1\leq i\leq m, \\
  &  \quad \boldsymbol{X}\succeq{\bf 0}.
\end{array}
\]
Note that the log-likelihood function is augmented by the trace functional $\mathsf{Tr}(\cdot)$
whose role is to promote low-rank solutions.  While such convex
relaxation schemes enjoy intriguing performance guarantees in many
aspects (e.g.~they achieve minimal sample complexity
and near-optimal error bounds for certain noise models), the
computational cost typically far exceeds the order of $n^3$. This
limits applicability to large-dimensional data.

This paper follows another route: rather than lifting the problem into
higher dimensions by introducing matrix variables, this paradigm
maintains its iterates within the vector domain and optimize the
nonconvex objective directly
(e.g. \cite{gerchberg1972practical,candes2014wirtinger,netrapalli2013phase,elser2003phase,Schniter2015,zheng2015convergent,marchesin2014alternating,repetti2014nonconvex,shechtman2013gespar,white2015local}).
One promising approach along this line is the recently proposed
two-stage algorithm called \emph{Wirtinger Flow} (WF)
\cite{candes2014wirtinger}.  Simply put, WF starts by computing a
suitable initial guess $\boldsymbol{z}^{(0)}$ using a spectral method,
and then successively refines the estimate via an update rule that
bears a strong resemblance to a gradient descent scheme, namely,
\[
	\boldsymbol{z}^{(t+1)}=\boldsymbol{z}^{(t)} + \frac{\mu_{t}}{m} \sum\nolimits_{i=1}^{m}\nabla\ell ( \boldsymbol{z}^{(t)}; y_i ),
	\label{eq:WF}
\]
where $\boldsymbol{z}^{(t)}$ denotes the $t$th iterate of the
algorithm, and $\mu_{t}$ is the step size (or learning rate). Here,
$\nabla\ell(\boldsymbol{z}; y_i)$ stands for the Wirtinger derivative
w.r.t. $\boldsymbol{z}$, which in the real-valued case reduces to the
ordinary gradient.  The main results of \cite{candes2014wirtinger}
demonstrate that WF is surprisingly accurate for independent Gaussian design. Specifically, when
$\boldsymbol{a}_i \sim \mathcal{N}( {\bf 0}, \boldsymbol{I})$ or
$\boldsymbol{a}_i \sim \mathcal{N}( {\bf 0}, \boldsymbol{I} ) + j
\mathcal{N}( {\bf 0}, \boldsymbol{I} )$:
\begin{enumerate}
\item WF achieves exact recovery from
  $m= O\left(n\log n\right)$ quadratic equations when there is no noise;\footnote{
    The standard notation $f(n)= O \left(g(n)\right)$ or
    $f(n)\lesssim g(n)$ (resp. $f(n)=\Omega\left(g(n)\right)$ or
    $f(n)\gtrsim g(n)$) means that there exists a constant $c>0$ such
    that $\left|f(n)\right|\leq c|g(n)|$ (resp.
    $|f(n)|\geq c\left|g(n)\right|$).
% $f(n)=\Theta\left(g(n)\right)$ or 
$f(n)\asymp g(n)$ means that there
 exist constants $c_{1},c_{2}>0$ such that $c_{1}|g(n)|\leq|f(n)|\leq c_{2}|g(n)|$.
} 
%\ejc{For which model?}
\item WF attains $\epsilon$-accuracy---in a relative sense---within
  $O (mn^2\log ({1}/{\epsilon}))$ time (or flops);
\item  In the presence of Gaussian noise, WF is stable
    and converges to the MLE as shown in
    \cite{soltanolkotabi2014algorithms}.
\end{enumerate}
While these results formalize the advantages of WF, the computational
complexity of WF is still much larger than the best one can hope for.  Moreover,
the statistical guarantee in terms of the sample complexity is weaker
than that achievable by convex relaxations.\footnote{
   % We note that a resampled version of WF is able to
    % achieve $\epsilon$-accuracy within
    % $\mathcal{O}(mn\log(\frac{1}{\epsilon}))$ time, although this
    % version is of mere theoretical interest.
    M.~Soltanolkotabi recently informed us that the sample complexity
    of WF may be improved if one employs a better initialization
    procedure.  }
%This gives rise to the natural question: is it possible to develop a nonconvex
%scheme that runs in linear time
%while enjoying matching statistical guarantees compared to  SDP? 

%This is the goal we
%aim to address in the present paper. 

\subsection{This paper: Truncated Wirtinger Flow}

This paper develops an efficient linear-time algorithm, which also enjoys
near-optimal statistical guarantees.  Following the spirit of WF, we
propose a novel procedure called \emph{Truncated Wirtinger Flow} (TWF)
adopting a more adaptive gradient flow. Informally, TWF proceeds in two
stages:
\begin{enumerate}
\item {\bf Initialization:} compute an initial guess
  $\boldsymbol{z}^{(0)}$ by means of a spectral method applied to a
  subset $\mathcal{T}_0$ of the observations $\{y_i\}$;
  \item {\bf Loop:} for $0\leq t<T$,
  \begin{equation}
    \boldsymbol{z}^{(t+1)} = \boldsymbol{z}^{(t)} + \frac{\mu_{t}}{m} \sum_{i\in\mathcal{T}_{t+1}}\nabla \ell (\boldsymbol{z}^{(t)}; y_i )
  \end{equation}
  for some index subset $\mathcal{T}_{t+1} \subseteq \{1,\cdots,m\} $ determined by $\boldsymbol{z}^{(t)}$. 
  %$\boldsymbol{z}^{(T)}$ is returned as the final estimate. 
\end{enumerate}

Three remarks are in order. 
\begin{itemize}
\item Firstly, we regularize both the initialization and the gradient flow in a data-dependent fashion by operating only upon some iteration-varying index subsets $\mathcal{T}_{t}$. This is a distinguishing feature of TWF in
  comparison to WF and other gradient descent variants.  In words, $\mathcal{T}_{t}$ corresponds to
  those data $\{y_i\}$ whose resulting spectral or gradient components are in some sense not
  excessively large; see Sections \ref{sec:Algorithm} and
  \ref{sec:Why-it-works} for details.  As we shall see later, the main point is
  that this careful data trimming procedure gives us a tighter initial guess and more stable search directions.
\item Secondly, we recommend that the step size $\mu_{t}$ is either
  taken as some appropriate constant or determined by a backtracking
  line search.  For instance, under appropriate conditions, we can
  take $\mu_t = 0.2$ for all $t$.
\item Finally, the most expensive part of the gradient stage consists
  in computing $\nabla\ell(\boldsymbol{z};y_{i})$, $1\leq i\leq
  m$, which can often be performed in an efficient manner.  More
  concretely, under the {\em real-valued} Poisson data model (\ref{eq:Poisson}) one has
\begin{equation*}
	\nabla\ell(\boldsymbol{z};y_{i})
	= 2\left\{ \frac{y_{i}}{|\boldsymbol{a}_{i}^{\top}\boldsymbol{z}|^{2}}\boldsymbol{a}_{i}\boldsymbol{a}_{i}^{\top}\boldsymbol{z}-\boldsymbol{a}_{i}\boldsymbol{a}_{i}^{\top}\boldsymbol{z}\right\} 
	= 2\left(\frac{y_{i}-|\boldsymbol{a}_i^{\top}\boldsymbol{z}|^{2}}{\boldsymbol{a}_i^{\top}\boldsymbol{z}}\right) \, \boldsymbol{a}_{i},
	%\quad 1\leq i\leq m.
\end{equation*}
Thus, calculating $\{\nabla\ell(\boldsymbol{z};y_{i})\}$ essentially
amounts to two matrix-vector products. Letting
$\boldsymbol{A} := [\boldsymbol{a}_1, \cdots,\boldsymbol{a}_m
]^{\top}$ as before, we have
\begin{equation*}
\sum_{i\in\mathcal{T}_{t+1}}\nabla \ell (\boldsymbol{z}^{(t)}; y_i ) = \boldsymbol{A}^{\top} \boldsymbol{v}, \qquad v_i = \begin{cases} 2\, \frac{y_{i}-|\boldsymbol{a}_i^{\top}\boldsymbol{z}|^{2}}{\boldsymbol{a}_i^{\top}\boldsymbol{z}},  & i \in \mathcal{T}_{t+1},\\
  0, & \text{otherwise}.
\end{cases}
\end{equation*}
Hence, $\boldsymbol{A}\boldsymbol{z}$ gives $\boldsymbol{v}$ and 
$\boldsymbol{A}^{\top}\boldsymbol{v}$ the desired regularized
gradient. 
\end{itemize}
A detailed specification of the algorithm is deferred to Section
\ref{sec:Algorithm}.

\subsection{Numerical surprises}

To give the readers a sense of the practical power of TWF, we present
here three illustrative numerical examples. Since it is
impossible to recover the global sign---i.e.~we cannot distinguish
$\boldsymbol{x}$ from $-\boldsymbol{x}$---we will evaluate our
solutions to the quadratic equations through the distance measure put
forth in \cite{candes2014wirtinger} representing the Euclidean
distance modulo a global sign: for complex-valued signals, 
\begin{eqnarray}
	\mathrm{dist}\left(\boldsymbol{z},\boldsymbol{x}\right) := \min\nolimits_{\varphi:\in[0,2\pi)}\|e^{-j\varphi}\boldsymbol{z} - \boldsymbol{x}\|,
	\label{eq:defn-dist}
\end{eqnarray}
while it is simply $\min\|\boldsymbol{z} \pm \boldsymbol{x}\|$ in the
real-valued case.  We shall use
$\mathrm{dist}(\hat{\boldsymbol{x}}, \boldsymbol{x} ) / \|
\boldsymbol{x}\|$
throughout to denote the relative error of an estimate
$\hat{\boldsymbol{x}}$.  In the sequel, TWF proceeds by attempting to maximize the
Poisson log-likelihood (\ref{eq:PoissonLL}).
%
% \[
%  \ell \left(\boldsymbol{z}; y_i \right)
%	=y_i \log\left(|\boldsymbol{a}_i ^{\top}\boldsymbol{z}|^2 \right)-|\boldsymbol{a}_i ^{\top}\boldsymbol{z}|^2.
% \]
% 
Standalone Matlab implementations of TWF are available at
\url{http://statweb.stanford.edu/~candes/publications.html} (see \cite{mahdiWFcode} for straight WF implementations).

The first numerical example concerns the following two problems under noiseless real-valued
data:
\begin{align*}
  (\text{a})\text{ }\quad\text{find }\boldsymbol{x} \in \mathbb{R}^{n} \quad \hspace{4em}
	\text{s.t. }  b_i &= \boldsymbol{a}_i^{\top}\boldsymbol{x}, &1\leq i\leq m;\\
  (\text{b})\quad\text{ }\text{find }\boldsymbol{x} \in \mathbb{R}^{n} \quad \hspace{4em}
	\text{s.t. }  b_i &= | \boldsymbol{a}_i^{\top}\boldsymbol{x} |, &1\leq i\leq m.
\end{align*}
Apparently, (a) involves solving a linear system of equations (or a linear least squares problem), while (b) is tantamount to solving a quadratic system. Arguably the most popular method for solving large-scale
least squares problems is the conjugate gradient (CG) method
\cite{nocedal2006numerical} applied to the normal equations.  We are
going to compare the computational efficiency between CG (for solving
least squares) and TWF with a step size $\mu_t\equiv 0.2$ (for solving
a quadratic system). Set $m=8n$ and generate
$\boldsymbol{x} \sim \mathcal{N} ({\bf 0},\boldsymbol{I} )$ and
$\boldsymbol{a}_{i} \sim \mathcal{N} ({\bf 0},\boldsymbol{I} )$,
$1\leq i\leq m$, independently.  This gives a matrix
$\boldsymbol{A}^{\top} \boldsymbol{A}$ with
a low condition number equal to about
$(1+\sqrt{1/8})^2/(1-\sqrt{1/8})^2 \approx 4.38$ by the
Marchenko-Pastur law. Therefore, this is an ideal setting for CG as it
converges extremely rapidly \cite[Theorem 38.5]{trefethen1997numerical}.
Fig.~\ref{Fig:CG-TWF} shows the relative estimation error of each
method as a function of the iteration count, where TWF is seeded
through 10 power iterations.  For ease of comparison, we illustrate
the iteration counts in different scales so that 4 TWF iterations are
equivalent to 1 CG iteration.

Recognizing that each iteration of CG and TWF involves two matrix
vector products $\boldsymbol{A}\boldsymbol{z}$ and
$\boldsymbol{A}^{\top}\boldsymbol{v}$, for such a design we reach a
suprising observation: 
{
\setlist{rightmargin=\leftmargin}
\begin{itemize}
\item[] {\em Even when all phase information is missing, TWF is
    capable of solving a quadratic system of equations only about 4
    times slower than solving a least squares problem of the same
    size!}
\end{itemize}
}

\begin{figure}[h]
  \centering
  \includegraphics[width=0.45\textwidth]{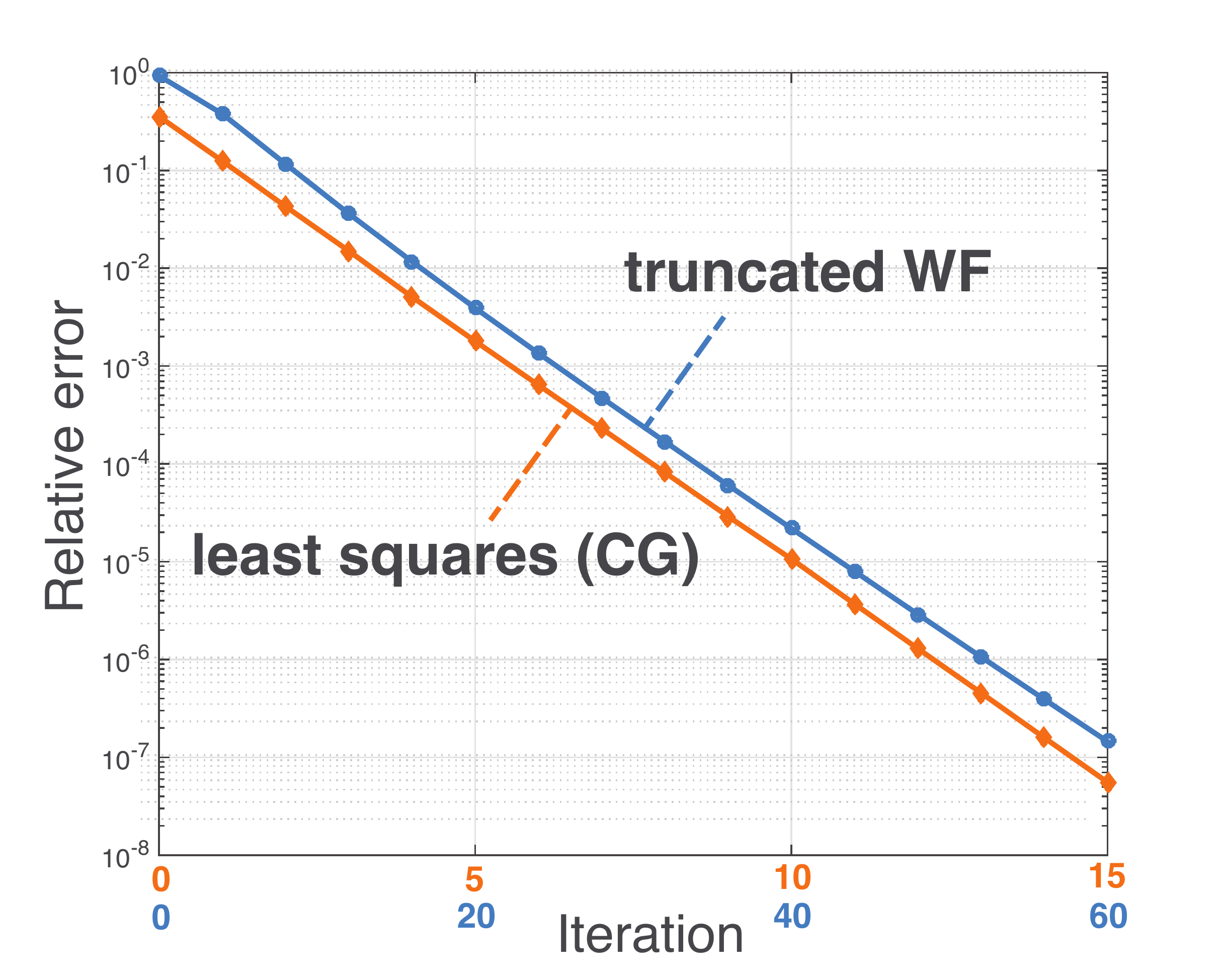}
  \caption{
    Relative errors of CG and TWF vs.~iteration count. Here, $n = 1000$,
    $m=8n$, and TWF is seeded using just 10 power iterations.}
  \label{Fig:CG-TWF}
\end{figure}

To illustrate the applicability of TWF on real images, we turn to
testing our algorithm on a  digital photograph of
Stanford main quad containing $320\times1280$ pixels. We consider a type of measurements that falls
under the category of coded diffraction patterns (CDP)
\cite{candes2014phase} and set
\begin{equation}
	\boldsymbol{y}^{(l)} = | \boldsymbol{F} \boldsymbol{D}^{(l)}  \boldsymbol{x} |^2,
	\quad\quad 1\leq l\leq L.
	\label{eq:CDP}
\end{equation}
Here, $\boldsymbol{F}$ stands for a discrete Fourier transform (DFT)
matrix, and $\boldsymbol{D}^{(l)}$ is a diagonal matrix whose diagonal
entries are independently and uniformly drawn from $\{ 1,-1,j,-j \}$
(phase delays).  In phase retrieval, each $\boldsymbol{D}^{(l)}$
represents a random mask placed after the object so as to modulate the
illumination patterns. When $L$ masks are employed, the total number
of quadratic measurements is $m=nL$.  In this example, $L=12$ random
coded patterns are generated to measure each color band (i.e. red,
green, or blue) separately. The experiment is carried out on a MacBook
Pro equipped with a 3 GHz Intel Core i7 and 16GB of memory.  We run 50
iterations of the truncated power method for initialization, and 50
regularized gradient iterations, which in total costs 43.9 seconds or
2400 FFTs for each color band. The relative error after regularized spectral
initialization and after 50 TWF iterations are 0.4773 and
$2.16\times10^{-5}$, respectively, with the recovered images displayed
in Fig. \ref{fig:real-image}.  
%In comparison, WF with 50 power iterations and 100 gradient iterations (which takes 54.5 seconds
%or 3600 FFTs) returns an image of relative error 1.309, still extremely far from the truth.
In comparison, the spectral initialization using 50 untruncated power iterations returns an image of relative error 1.409, which is almost like a random guess and extremely far from the truth.

\begin{figure}[h]
	\centering
	\begin{tabular}{c}
		\includegraphics[width=1\textwidth]{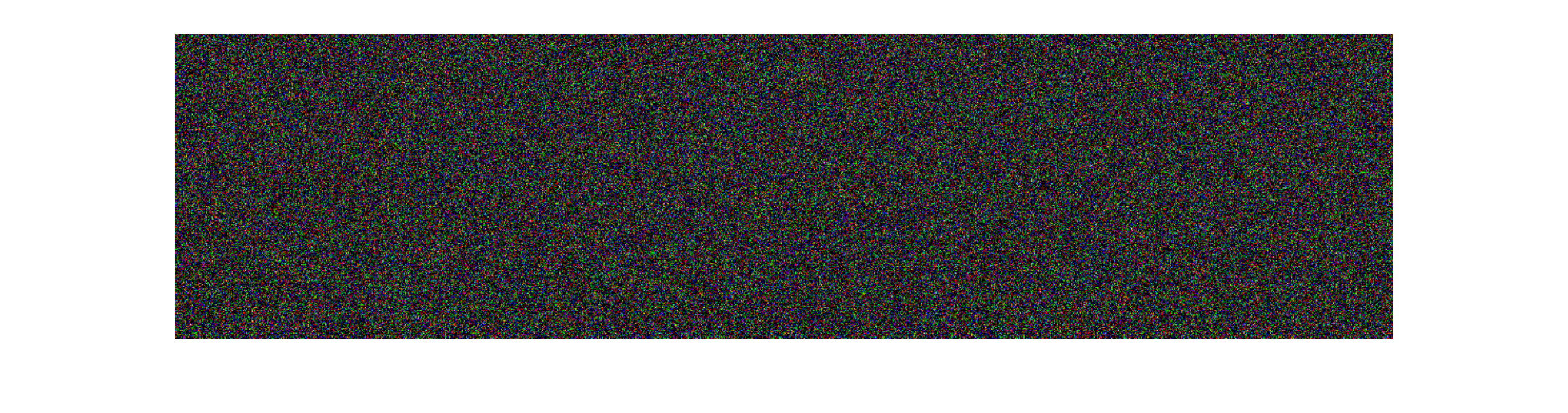} \tabularnewline
		(a) \tabularnewline
		\includegraphics[width=1\textwidth]{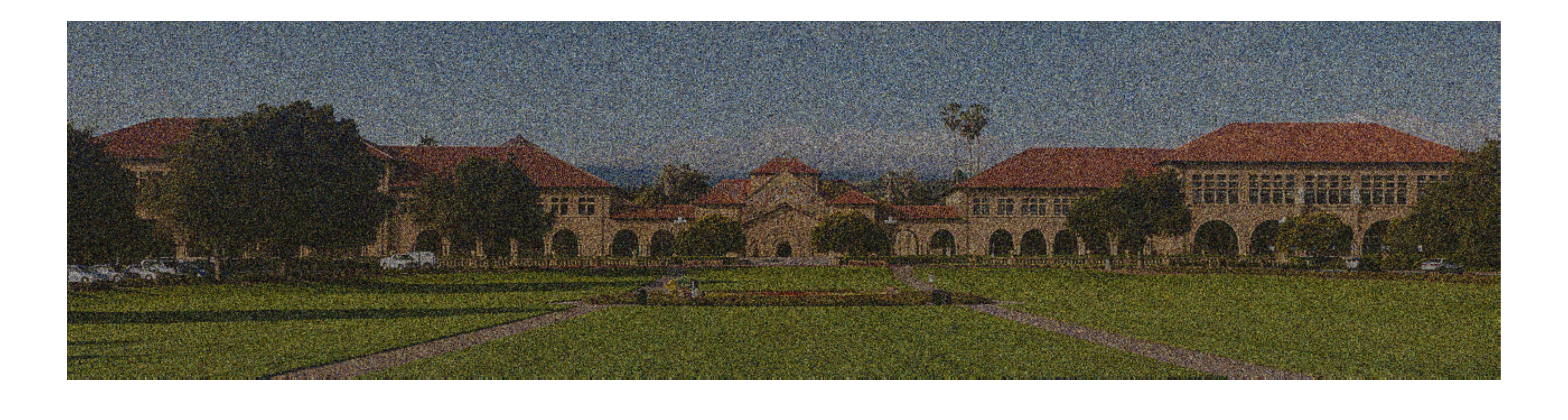} \tabularnewline
		(b) \tabularnewline
		\includegraphics[width=1\textwidth]{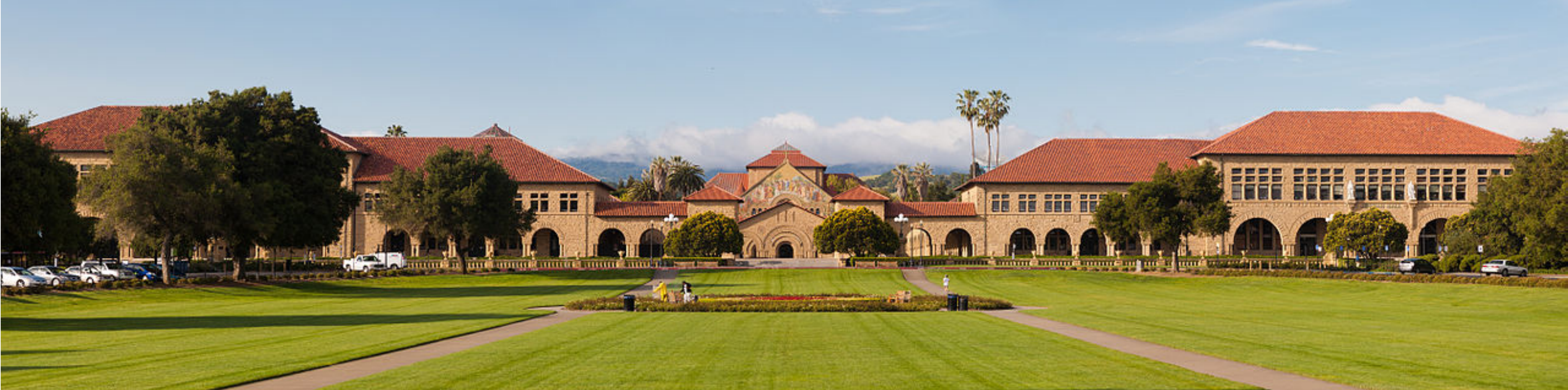} \tabularnewline
		(c) \tabularnewline
	\end{tabular}
	
	\caption{ The recovered images after (a) spectral initialization; (b) regularized spectral initialization; and (c) 50 TWF gradient iterations following the regularized initiliazation. 
          \label{fig:real-image}}
\end{figure}

\begin{figure}[h]
	\centering
	\includegraphics[width=0.48\textwidth]{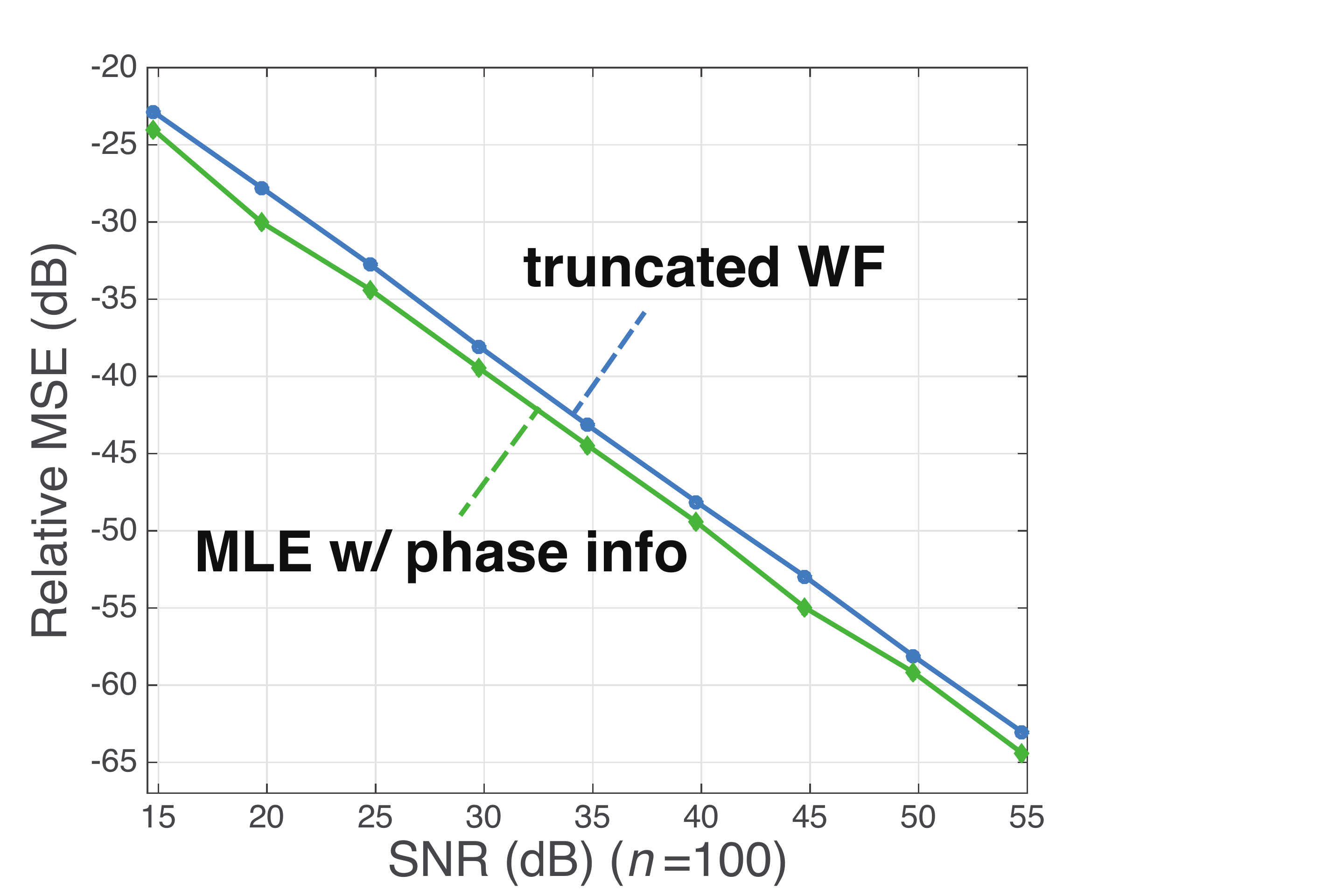}
	\caption{Relative MSE vs.~SNR in dB. The curves are shown for
          two settings: TWF for solving quadratic equations (blue),
          and MLE had we observed additional phase information
          (green). The results are shown for $n=100$, and each point
          is averaged over 50 Monte Carlo trials. }
	\label{Fig:MSE-MLE}
\end{figure}

While the above experiments concern noiseless data, the numerical
surprise extends to the noisy realm. Suppose the data are drawn
according to the Poisson noise model (\ref{eq:Poisson}), with
$\boldsymbol{a}_i \sim \mathcal{N}\left({\bf 0},\boldsymbol{I}\right)$
independently generated. Fig.~\ref{Fig:MSE-MLE} displays the empirical
relative mean-square error (MSE) of TWF as a function of the
signal-to-noise ratio (SNR), where the relative MSE for an estimate
$\hat{\boldsymbol{x}}$ and the SNR are defined as\footnote{To justify
  the definition of SNR, note that the signals and noise are captured
  by $\mu_i= (\boldsymbol{a}_i^{\top}\boldsymbol{x})^2$ and $y_i -
  \mu_i$, $1 \leq i \leq m$, respectively.  The ratio of the signal
  power to the noise power is therefore $\frac{ \sum_{i=1}^m
    \mu_i^2}{\sum_{i=1}^m {\bf Var} [y_i] } = \frac { \sum_{i=1}^m
    |\boldsymbol{a}_i ^{\top} \boldsymbol{x} |^4 } { \sum_{i=1}^m
    |\boldsymbol{a}_i ^{\top} \boldsymbol{x} |^2 } \approx \frac{3m
    \|\boldsymbol{x}\|^4 } { m \|\boldsymbol{x}\|^2 } = 3
  \|\boldsymbol{x} \|^2.$ }
\begin{equation}
  	\text{MSE}:= \frac{\text{dist}^2(\hat{\boldsymbol{x}}, \boldsymbol{x})} { \|\boldsymbol{x}\|^2}, \quad \text{and}\quad
	\text{SNR}:= 3\Vert \boldsymbol{x} \Vert ^2.
	\label{eq:SNR}
\end{equation}
Both SNR and MSE are displayed on a dB scale (i.e.~the values of
$10\log_{10}(\text{SNR})$ and $10\log_{10}(\text{rel. MSE})$ are
plotted).  To evaluate the accuracy of the TWF solutions, we consider
the performance achieved by MLE applied to an \emph{ideal} problem in
which the true phases are revealed.  In this ideal scenario, in
addition to the data $\{y_i\}$ we are further given exact phase
information $\{\varphi_{i}=\mathrm{sign}
(\boldsymbol{a}_{i}^{\top}\boldsymbol{x} )\}$. Such precious
information gives away the phase retrieval problem and makes the MLE
efficiently solvable since the MLE problem with side information
\[
\begin{array}{ll}
  \text{minimize}_{\boldsymbol{z}\in \mathbb{R}^n} & \quad  - \sum_{i=1}^{m}y_{i}\log\big( |\boldsymbol{a}_{i}^{\top}\boldsymbol{z} |^{2} \big)
  + ( \boldsymbol{a}_{i}^{\top}\boldsymbol{z} )^2 \\
  \text{subject to} &\quad   \varphi_{i}=\mathrm{sign} (\boldsymbol{a}_{i}^{\top}\boldsymbol{z})
\end{array}
\]
can be cast as a convex program
\begin{eqnarray*}
  \text{minimize}_{\boldsymbol{z}\in \mathbb{R}^n} \quad 
  - \sum\nolimits_{i=1}^{m} 2 y_{i} \log \left( \varphi_{i}\boldsymbol{a}_{i}^{\top}\boldsymbol{z} \right)
  +  (\boldsymbol{a}_{i}^{\top}\boldsymbol{z} )^2.
\end{eqnarray*}
Fig.~\ref{Fig:MSE-MLE} illustrates the empirical performance for this
ideal problem. The plots demonstrate that even when all phases are
erased, TWF yields a solution of nearly the best possible quality,
since it only incurs an extra $1.5$ dB loss compared to ideal MLE
computed with all true phases revealed. This phenomenon arises
regardless of the SNR!

\subsection{Main results}

The preceding numerical discoveries unveil promising features of TWF
in three aspects: (1) exponentially fast convergence; (2) exact
recovery from noiseless data with sample complexity
$O\left(n\right)$; (3) nearly minimal mean-square loss in
the presence of noise. This paper offers a formal substantiation of
all these findings. To this end, we assume a tractable model in which
the design vectors $\boldsymbol{a}_{i}$'s are independent Gaussian:
\begin{equation}
	\boldsymbol{a}_{i}\sim\mathcal{N}\left({\bf 0},\boldsymbol{I}_{n}\right).
	\label{eq:Gaussian}
\end{equation}
For concreteness, our results are concerned with TWF designed based on the
Poisson log-likelihood function
\begin{eqnarray}
	\ell_i (\boldsymbol{z}) := \ell\left(\boldsymbol{z}; y_i \right)  
		= y_{i}\log\left(| \boldsymbol{a}_{i}^{\top} \boldsymbol{z} |^2 \right)
				 - |\boldsymbol{a}_{i}^{\top}\boldsymbol{z}|^2,
\end{eqnarray}
where we shall use $\ell_i (\boldsymbol{z})$ as a shorthand for  $\ell (\boldsymbol{z}; y_i)$ from now on.
We begin with the performance guarantees of TWF in the absence of noise.

\begin{theorem}[\textbf{Exact recovery}]
  \label{theorem-Truncated-WF}
  Consider the noiseless case (\ref{eq:Noiseless}) with an arbitrary
  signal $\boldsymbol{x}\in\mathbb{R}^{n}$. Suppose that the step size
  $\mu_t$ is either taken to be a positive constant $\mu_t\equiv \mu$
  or chosen via a backtracking line search.  Then there exist some
  universal constants $0<\rho, \nu<1$ and $ \mu_0, c_{0}, c_1, c_{2}>0$
  such that with probability exceeding
  $1-c_{1}\exp\left(-c_{2}m\right)$, the truncated Wirtinger Flow
  estimates (Algorithm \ref{alg:TWF} with parameters specified in
  Table \ref{table:pars}) obey 
  \begin{eqnarray}
	\mathrm{dist} (\boldsymbol{z}^{(t)},\boldsymbol{x} ) & \leq & \nu (1-\rho)^{t} \Vert \boldsymbol{x} \Vert ,\quad \forall t\in\mathbb{N},
	\label{eq:Contraction-noiseless}
  \end{eqnarray}
  provided that 
  \[ m\geq c_{0}n \quad \text{and} \quad 0< \mu \leq \mu_0.  \] As
  explained below, we can often take $\mu_0 \approx 0.3$.
\end{theorem}

\begin{remark}As will be made precise in Section \ref{sec:Proof-of-Theorem-TruncatedWF} (and in particular Proposition \ref{prop-Contraction}), one can take
  \[
    \mu_0 = \frac{0.994-\zeta_{1}-\zeta_{2}-\sqrt{ 2/(9\pi)} \alpha_h^{-1} }{ 2 \left(1.02 + 0.665/ \alpha_h \right)}
  \]
  for some small quantities $\zeta_{1},\zeta_{2}$ and some predetermined threshold $\alpha_{h}$ that is usually taken to be $\alpha_{h}\geq 5$. 
  Under appropriate conditions, one can treat $\mu_0$ as $\mu_0 \approx 0.3$.
\end{remark}

% \begin{remark}
% 	The convergence rate  $1-\rho$ can also be provably controlled in a reasonably tight manner; see Section \ref{sec:Proof-of-Theorem-TruncatedWF} for details.
% \end{remark}

Theorem \ref{theorem-Truncated-WF} justifies at least two appealing
features of TWF: (i) {\em minimal sample complexity} and (ii) {\em
  linear-time computational cost}.  Specifically, TWF allows exact recovery
from $O(n)$ quadratic equations, which is optimal since one
needs at least $n$ measurements to have a well-posed problem.  Also,
because of the geometric convergence rate, TWF achieves
$\epsilon$-accuracy (i.e. $\mathrm{dist}
(\boldsymbol{z}^{(t)},\boldsymbol{x} ) \leq \epsilon\left\Vert
  \boldsymbol{x}\right\Vert $) within at most
$O\left(\log (1/\epsilon)\right)$ iterations.  The
total computational cost is therefore
$O\left(mn\log(1/\epsilon)\right)$, which is linear in
the problem size. These outperform the performance guarantees of WF
\cite{candes2014wirtinger}, which runs in
$O(mn^2\log (1/\epsilon))$ time and requires $O(n\log n)$
sample complexity.

We emphasize that enhanced performance vis-\`a-vis WF is not the
result of a sharper analysis, but rather, the result of key
algorithmic changes.  In both the initialization and iterative
refinement stages, TWF proceeds in a more prudent manner by means of
proper regularization, which effectively trims away those components that are too influential on either the initial
guess or search directions, thus reducing the volatility of each movement.  
With a tighter initialization and better-controlled search directions in place, we
take the step size in a far more liberal fashion---which is some constant bounded
away from 0---compared to a step size which is $O(1/n)$ as explained in \cite{candes2014wirtinger}. In fact, what
enables the movement to be more aggressive is exactly the cautious
choice of $\mathcal{T}_t$, which precludes adverse effects from
high-leverage samples.

To be broadly applicable, the proposed algorithm must guarantee
reasonably faithful estimates in the presence of noise. Suppose that
\begin{equation}
	y_{i} = \left| \left\langle \boldsymbol{a}_i,\boldsymbol{x} \right\rangle \right|^2 + \eta_i,\qquad 1\leq i\leq m,
	\label{eq:Noisy}
\end{equation}
where $\eta_{i}$ represents an error term. 
We claim that TWF is
stable against additive noise, as demonstrated in the theorem below.

\begin{theorem}[\textbf{Stability}]
  \label{theorem-Truncated-WF-noisy}
  Consider the noisy case (\ref{eq:Noisy}). Suppose that the step size
  $\mu_t$ is either taken to be a positive constant $\mu_t\equiv \mu$
  or chosen via a backtracking line search.  If 
	\begin{equation}
		m\geq c_{0}n, \quad \mu\leq\mu_{0},  
		\quad\text{and}\quad 
		\left\Vert \boldsymbol{\eta} \right\Vert_{\infty} \leq c_1\left\Vert \boldsymbol{x}\right\Vert^2, 
	\end{equation}
	then with probability at least $1-c_{2}\exp\left(-c_{3}m\right)$, the truncated Wirtinger Flow estimates 
	(Algorithm \ref{alg:TWF} with parameters specified in Table \ref{table:pars}) satisfy
	\begin{eqnarray}
		\mathrm{dist} ( \boldsymbol{z}^{(t)},\boldsymbol{x} ) 
		& \lesssim & \frac{\Vert \boldsymbol{\eta} \Vert }{\sqrt{m}\Vert \boldsymbol{x} \Vert }
		 + (1-\rho )^t \Vert \boldsymbol{x} \Vert ,
		\quad\forall t\in\mathbb{N}
		\label{eq:noisy-converge}
	\end{eqnarray}
	simultanesouly for all $\boldsymbol{x}\in \mathbb{R}^n$. Here, $0<\rho<1$ and $\mu_{0}, c_{0}, c_1, c_2, c_{3} > 0$ are some universal
  	constants.

        Under the Poisson noise model (\ref{eq:Poisson}), one has 
%there exists an an event of probability at least $1- c_2\exp(-c_3m)$ on which
        \begin{eqnarray}
          \label{eq:noisy-poisson}
          %\sup_{\boldsymbol{x}\in \mathbb{R}^n: }
	  %\mathbb{P} \Big\{ 
	  \mathrm{dist} ( \boldsymbol{z}^{(t)},\boldsymbol{x} ) ~\lesssim ~ 1 + (1-\rho )^t \Vert \boldsymbol{x} \Vert, ~~\forall t\in\mathbb{N} 
	 %~\Big|~ \{\boldsymbol{a}_i\}_{1\leq i\leq m}  %\Big\} \rightarrow 1.    
        \end{eqnarray}
        with probability approaching one, 
	%holds simultaneously for all $\boldsymbol{x}\in \mathbb{R}^n$ satisfying 
	provided that $\|\boldsymbol{x}\| \geq \log^{1.5}m$.
\end{theorem}

\begin{remark} In the main text, we will prove Theorem
  \ref{theorem-Truncated-WF-noisy} only for the case where
  $\boldsymbol{x}$ is fixed and independent of the design vectors
  $\{\boldsymbol{a}_i\}$. Interested readers are referred to the
  supplemental materials \cite{supplement2015} for the proof
  of the universal theory (i.e. the case simultaneously accommodating
  all $\boldsymbol{x}\in \mathbb{R}^n$). Note that when
    there is no noise ($\boldsymbol{\eta} = \boldsymbol{0}$), this
    stronger result guarantees the universality of the noiseless
    recovery.
\end{remark}
\begin{remark}
   \cite{soltanolkotabi2014algorithms} establishes
    stability estimates using the WF approach under Gaussian
    noise. There, the sample and computational
    complexities are still on the order of $n \log n$ and $m n^2$
    respectively whereas the computational complexity in Theorem
    \ref{theorem-Truncated-WF-noisy} is linear, i.e.~on the order of
    $m n$.
\end{remark}

Theorem \ref{theorem-Truncated-WF-noisy} essentially reveals that the
estimation error of TWF rapidly shrinks to
$O\left(\frac{\Vert \boldsymbol{\eta} \Vert
    /\sqrt{m}}{\Vert \boldsymbol{x} \Vert }\right)$
within logarithmic iterations.  Put another way, since the SNR for the model (\ref{eq:Noisy}) is captured by
\begin{equation}
	\text{SNR}:= \frac{ \sum_{i=1}^m |\langle\boldsymbol{a}_i, \boldsymbol{x}\rangle|^4 }{ \| \boldsymbol{\eta} \|^2 } 
		   \approx \frac{ 3m \| \boldsymbol{x} \|^4} { \| \boldsymbol{\eta} \|^2 },
\end{equation}
we immediately arrive at an alternative form of the performance guarantee:
\begin{eqnarray}
	\mathrm{dist} ( \boldsymbol{z}^{(t)},\boldsymbol{x} ) 
	& \lesssim & \frac{1 }{ \sqrt{\text{SNR}} } \Vert \boldsymbol{x} \Vert
	 + (1-\rho )^t \Vert \boldsymbol{x} \Vert,
	\quad\forall t\in\mathbb{N},
	\label{eq:noisy-converge-SNR}
\end{eqnarray}
revealing the stability of TWF as a function of SNR.   We
  emphasize that this estimate holds for any error term
  $\boldsymbol{\eta}$---i.e.~any noise structure, even deterministic.  This being
  said, specializing this estimate to the Poisson noise model
  (\ref{eq:Poisson}) with $\left\Vert \boldsymbol{x} \right\Vert
  \gtrsim \log^{1.5} m$ gives an estimation error that will eventually
  approach a numerical constant, independent of $n$ and $m$.
%
%\[
%	\frac{1}{\sqrt{m}} { \left\Vert \boldsymbol{\eta} \right\Vert }  \asymp \left\Vert \boldsymbol{x} \right\Vert 
%	\quad \text{and} \quad
%	\left\Vert \boldsymbol{\eta} \right\Vert_{\infty} \ll \left\Vert \boldsymbol{x} \right\Vert^2
%\]
%
%with high probability,  

Encouragingly, this is already the best statistical guarantee any
algorithm can achieve. We formalize this claim by deriving a
fundamental lower bound on the minimax estimation error.

\begin{theorem}[\textbf{Lower bound on the minimax risk}]
  \label{theorem-converse}
    Suppose that  $\boldsymbol{a}_i \sim
  \mathcal{N}({\bf 0}, \boldsymbol{I})$, $m=\kappa n$ for some fixed
  $\kappa$ independent of $n$, and $n$ is sufficiently large. 
  For  any $K \geq \log^{1.5} m$, define\footnote{Here, 0.1 can be replaced by any positive constant within (0, 1/2).}
  \[
	\Upsilon(K):=\{\boldsymbol{x}\in \mathbb{R}^n \mid
  	\|\boldsymbol{x}\| \in (1\pm 0.1)K\}.
  \] 
   With
  probability approaching one, the minimax risk under the Poisson model
  (\ref{eq:Poisson}) obeys
  \begin{equation}
    \inf_{\hat{\boldsymbol{x}}} \sup_{\boldsymbol{x}\in \Upsilon(K)}
    \mathbb{E}\big[\mathrm{dist}\left(\hat{\boldsymbol{x}},\boldsymbol{x}\right) ~\big|~ \{\boldsymbol{a}_i\}_{1 \le i \le m} \big]
    ~\geq~ \frac{\varepsilon_{1}}{\sqrt{\kappa}},
	\label{eq:MinimaxLoss}
  \end{equation}
	where the infimum is over all estimator $\hat{\boldsymbol{x}}$. Here, $\varepsilon_{1}>0$ is a numerical constant independent of $n$ and $m$. 
\end{theorem}

% \yxc{I just checked the proof -- the current proof (in particular the proof of Lemma 8(b)) doesn't extend to the general regime below.  While I still think $\sqrt{\frac{n}{m}}$ is the right scaling,  perhaps it would be safer to drop  this remark since we don't yet have a proof?}
% \begin{remark}
% 	In the more general regime where $m$ grows faster than $n$, the minimax risk obeys
% 	\[
% 	\inf_{\hat{\boldsymbol{x}}} \sup_{\boldsymbol{x}\in \Upsilon(K)}
% 	\mathbb{E}\left[\mathrm{dist}\left(\hat{\boldsymbol{x}},\boldsymbol{x}\right)\right]\geq \sqrt{\frac{n}{m}}\varepsilon_{1}.
% 	\]	
% \end{remark}

When the number $m$ of measurements is proportional to $n$ and the
energy of the planted solution exceeds $\log^3m$, Theorem
\ref{theorem-converse} asserts that there exists absolutely no
estimator that can achieve an estimation error that vanishes as $n$
increases.  This lower limit matches the estimation error of TWF,
which corroborates the optimality of TWF under noisy data.

Recall that in many optical imaging applications, the output data we collect are the intensities of the diffractive waves scattered by the sample or specimen under study.      
% Careful readers will naturally wonder whether the regime
% $\|\boldsymbol{x}\| \geq \log^{1.5}m$---or $\|\boldsymbol{x}\| \geq
% \sqrt{n}\log^{1.5}m$ if we normalize $\boldsymbol{a}_i$ so that
% $\|\boldsymbol{a}_i\|\approx 1$---is of practical importance. 
%Note that in the optical imaging applications, 
The Poisson noise model employs the input $\boldsymbol{x}$ and output $\boldsymbol{y}$ to describe the numbers
of photons diffracted by the specimen and detected by the optical
sensor, respectively.  Each specimen needs to be
sufficiently illuminated in order for the receiver to sense the
diffracted light. In such settings, the low-intensity regime 
$\|\boldsymbol{x}\| \leq \log^{1.5} m$ is of little practical interest as it 
corresponds to an illumination with just very few photons. 
We forego the details.
% \yxc{Carlos has a comment here on astronomy.}
%\ejc{I may want to expand a
%  bit here. Not sure. If we do expand, we should be careful that
%  $\|Ax\|^2 \approx m \|x\|^2$ the way things are normalized so that
%  we are really asking that $\|x\|^2 \gtrsim m \log^3 m$ in a
%  diffraction experiment, i.e.~about $\log^3m$ photons per sensor.}

%\subsection{Prior art}

It is worth noting that apart from WF, various other nonconvex
procedures have been proposed as well for phase retrieval, including
the error reduction schemes dating back to Gerchberg-Saxton and Fienup
\cite{gerchberg1972practical,fienup1982phase}, iterated projections
\cite{elser2003phase}, alternating
minimization \cite{netrapalli2013phase}, generalized approximate message passing
\cite{Schniter2015}, Kaczmarz method \cite{wei2015phase}, and  greedy methods that exploit additional sparsity constraint
\cite{shechtman2013gespar}, to name just a few. While these paradigms
enjoy favorable empirical behavior, most of them fall short of
theoretical support, except for a version of alternating minimization
(called {AltMinPhase}) \cite{netrapalli2013phase} that requires fresh
samples for each iteration. In comparison, AltMinPhase attains
$\epsilon$-accuracy when the sample complexity exceeds the order
of $n\log^3n + n\log^2n \log({1}/{\epsilon})$, which is at least a
factor of $\log^3n$ from optimal and is empirically
largely outperformed by the variant that reuses all samples.  In contrast, our algorithm uses the same
set of samples all the time and is therefore practically appealing. 
Furthermore, none of these
algorithms come with provable stability guarantees, which are
particularly important in most realistic scenarios.  Numerically, each iteration of Fienup's algorithm
(or alternating minimization) involves solving a least squares problem, and the algorithm converges
in tens or hundreds of iterations. This is computationally more expensive than TWF, whose computational
complexity is merely about 4 times that of solving a least squares problem.
Interesting
readers are referred to \cite{candes2014wirtinger}  for a comparison of
several non-convex schemes, and  \cite{candes2014phase} for a discussion of other alternative approaches (e.g. \cite{alexeev2014phase,balan2009painless}) and performance lower bounds (e.g. \cite{eldar2014phase,bandeira2014saving}).

\section{Algorithm: Truncated Wirtinger Flow}
\label{sec:Algorithm}

This section describes the two stages of TWF
in details, presented in a reverse order. For each stage, we start
with some algorithmic issues encountered by WF, which is then used
to motivate and explain the basic principles of TWF. 
%For notational simplicity, we shall let $\nabla \ell_i(\boldsymbol{z}) = \nabla\ell( \boldsymbol{z};y_i)$ from now on.
Here and throughout, we let $\mathcal{A}:\mathbb{R}^{n\times n}\mapsto\mathbb{R}^{m}$
be the linear map 
\[
	\boldsymbol{M}\in\mathbb{R}^{n\times n}\quad\mapsto\quad\mathcal{A}\left(\boldsymbol{M}\right):=
	\left\{ \boldsymbol{a}_{i}^{\top}\boldsymbol{M}\boldsymbol{a}_{i}\right\} _{1\leq i\leq m} 
\]
and $\boldsymbol{A}$ the design matrix 
\[
	\boldsymbol{A} := [\boldsymbol{a}_1, \cdots,\boldsymbol{a}_m ]^{\top}.
\]

\subsection{Regularized gradient stage}

%\ejc{Below, you sometimes write $a_i^T z$ and sometimes $z^T a_i$. I
%  know it is the same but it's probably better to stick to one,
%  i.e. $a_i^T z$ as not to distract.}
%\yxc{I have been debating about whether to use $a_i^T z$ or $z^T a_i$. For real-valued case, no difference;  for complex-valued case, there seems only one choice, i.e.  $z^*a_i$ in the denominator.}

For independent samples, the
gradient of the real-valued Poisson log-likelihood  obeys
\begin{equation}
	\sum_{i=1}^{m}\nabla\mathcal{\ell}_{i} ( \boldsymbol{z} )
	=\text{ } \sum_{i=1}^m   \underset{:= \nu_{i}}{2\underbrace{\frac{y_{i} - |\boldsymbol{a}_i^{\top}\boldsymbol{z} |^2 }{ \boldsymbol{a}_i^{\top}\boldsymbol{z}}}} \boldsymbol{a}_i,
	\label{eq:score}
\end{equation}
where $\nu_i$ represents the weight assigned to each $\boldsymbol{a}_i$.
This forms the descent direction of WF updates. 

Unfortunately, WF moving along the preceding direction 
might not come close to the truth unless $\boldsymbol{z}$ is already very close to $\boldsymbol{x}$. To see this, it is helpful to consider any fixed vector
$\boldsymbol{z}\in \mathbb{R}^n$ independent of the design vectors.
The typical size of
$\min_{1\leq i\leq m} |\boldsymbol{a}_i^{\top}\boldsymbol{z} |$ is
about on the order of $\frac{1}{m} \Vert \boldsymbol{z}\Vert $,
introducing some unreasonably large weights $\nu_{i}$, which can be as
large as
$m\Vert \boldsymbol{x} \Vert ^{2}/ \Vert \boldsymbol{z} \Vert$.
Consequently, the iterative updates based on (\ref{eq:WF}) often
overshoot, and this arises starting from the very initial stage\footnote{For
complex-valued data where
$\boldsymbol{a}_i \sim \mathcal{N} \left({\bf 0},
  \boldsymbol{I}\right)+j\mathcal{N}\left({\bf 0},
  \boldsymbol{I}\right)$,
WF converges empirically, as $\min_i |\boldsymbol{a}_i^* \boldsymbol{z}|$ is much larger than the real-valued case.}.

\begin{figure}
	\centering
	\includegraphics[width=0.48\textwidth]{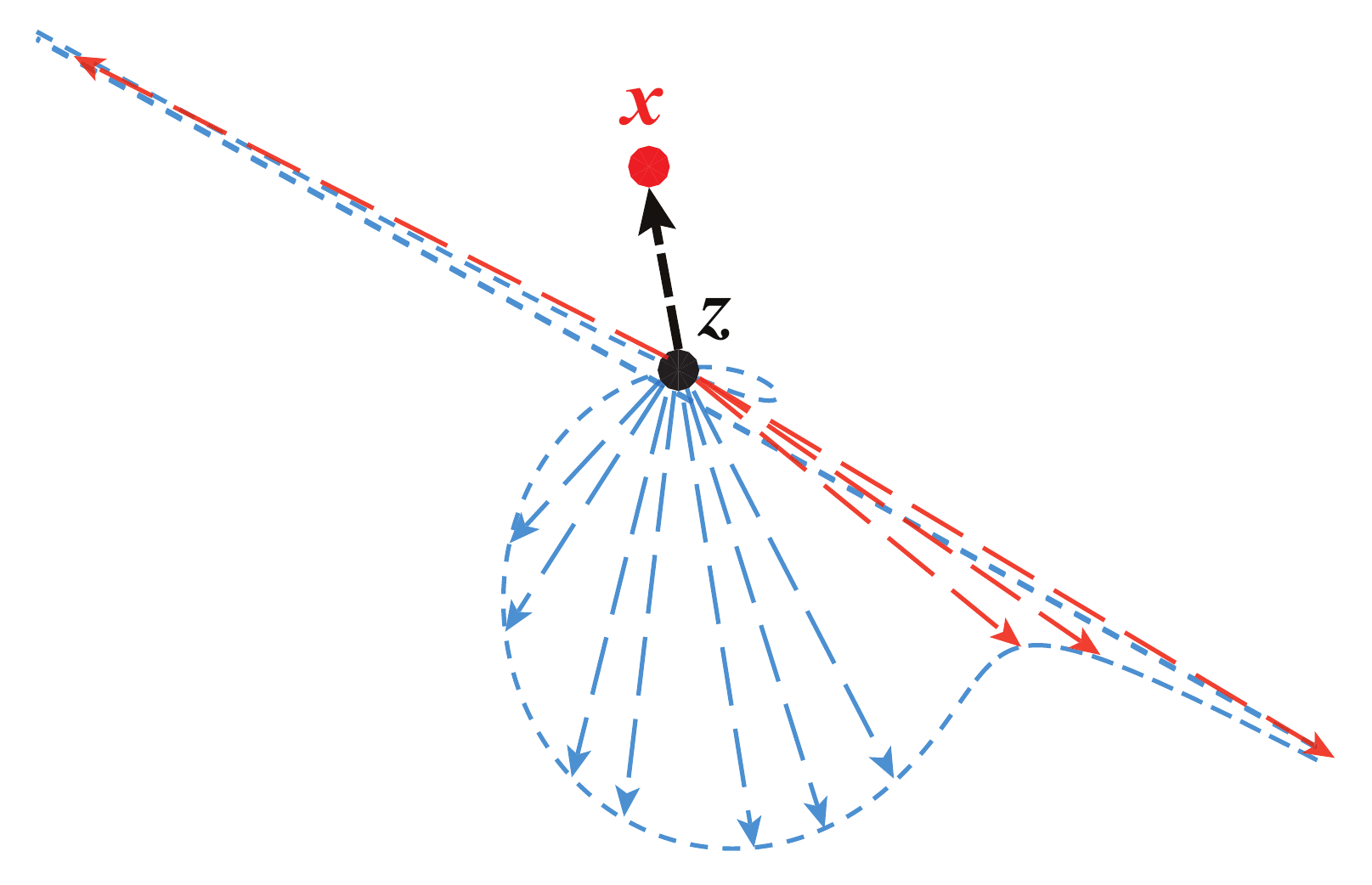}
	\caption{{The locus of $-\frac{1}{2}\nabla\ell_{i}\left(\boldsymbol{z}\right)
			       = \frac{\left|\boldsymbol{a}_{i}^{\top}\boldsymbol{z}\right|^{2}-\left|\boldsymbol{a}_{i}^{\top}\boldsymbol{x}\right|^{2}}{\boldsymbol{a}_{i}^{\top}\boldsymbol{z}}\boldsymbol{a}_{i}$
	when  $\boldsymbol{a}_{i}$ ranges over all unit vectors, 
 where $\boldsymbol{x}=\left(2.7,8\right)$
	and $\boldsymbol{z}=\left(3,6\right)$. 
For each direction $\boldsymbol{a}_i$, $-\frac{1}{2}\nabla\ell_{i}\left(\boldsymbol{z}\right)$ is aligned with $\boldsymbol{a}_i$, and its length represents the weight assigned to this component. 
        In particular, the red arrows depict a few directions that behave like outliers, whereas the blue arrows depict several directions whose resulting gradients take typical sizes.}}
	\label{fig:grad-distribution}
\end{figure}

\begin{figure}[h]
	\centering
	\includegraphics[width=0.4\textwidth]{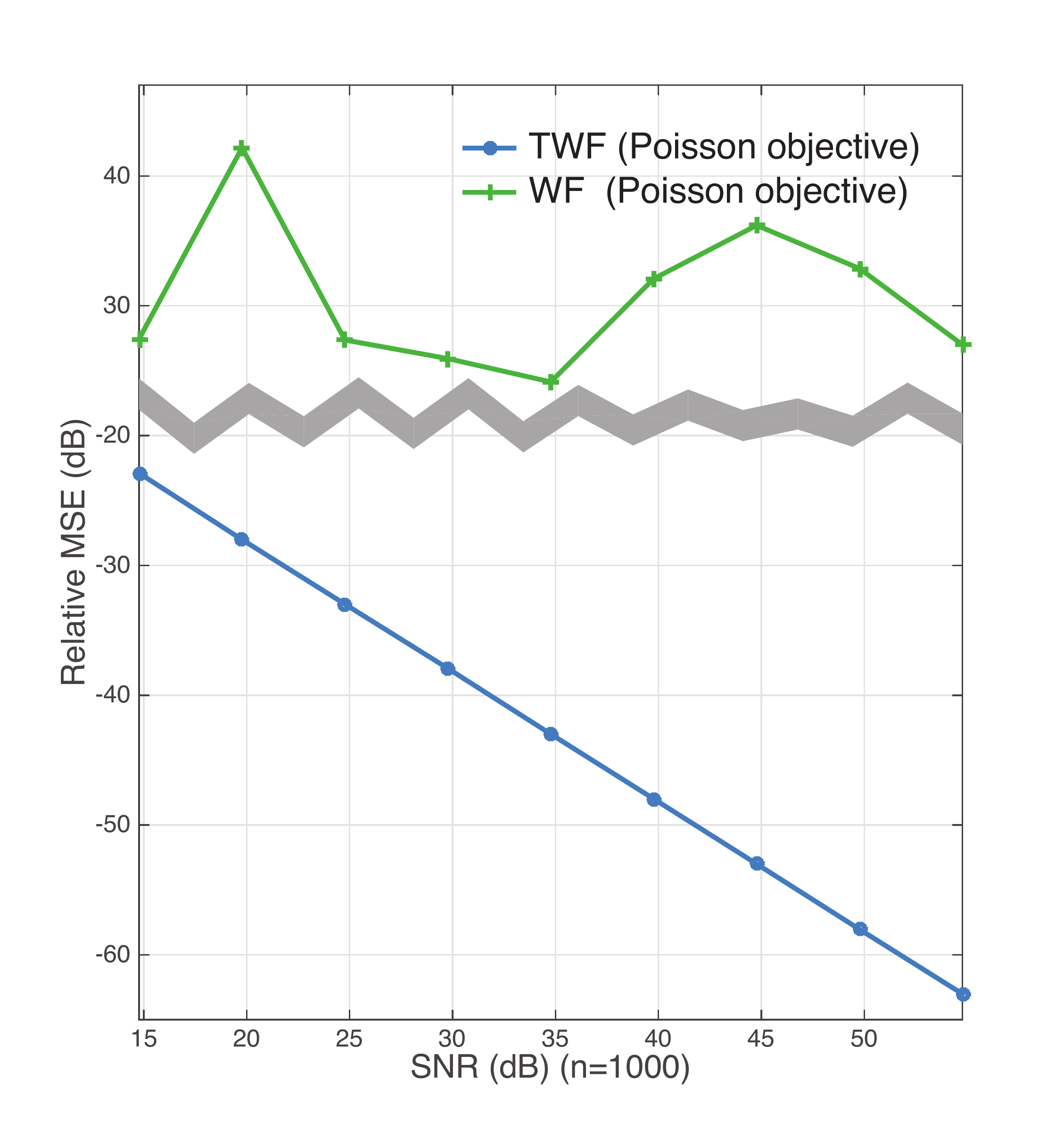} 
	\caption{ Relative MSE vs. SNR in dB. The curves
          are shown for WF and TWF, both employing the Poisson
          log-likelihood. Here, $\boldsymbol{a}_i \sim
          \mathcal{N}({\bf 0},\boldsymbol{I})$, $n=1000$, $m=8n$, and
          each point is averaged over 100 Monte Carlo trials.}
	\label{fig:RealPoisson}
\end{figure}

Fig.~\ref{fig:grad-distribution} illustrates this phenomenon by
showing the locus of
$-\nabla\mathcal{\ell}_{i}\left(\boldsymbol{z}\right)$ when
$\boldsymbol{a}_i$ has unit norm and ranges over all possible directions.  
 Examination of the figure seems to suggest
that most of the gradient components
$\nabla\ell_{i}\left(\boldsymbol{z}\right)$ are more or less pointing
towards the truth $\boldsymbol{x}$ and forming
reasonable search directions. But there exist a few outlier components that are excessively large, which lead to unstable search directions. 
Notably, an underlying
premise for a nonconvex procedure to succeed is to
ensure all iterates reside within a {\em basin of attraction}, that is, a neighborhood surrounding $\boldsymbol{x}$ within
which $\boldsymbol{x}$ is the unique stationary point of the objective. When a gradient is not well-controlled, the
iterative procedure might overshoot and end up leaving this basin of attraction. 
  This intuition is
  corroborated by numerical experiments under {\em real-valued}
  data. As illustrated in Fig.~\ref{fig:RealPoisson}, the solutions
  returned by the WF (designed for a real-valued Poisson
  log-likelihood and $m=8n$) are very far from the ground truth. 

Hence, to remedy the aforementioned
stability issue, it would be natural to separate the small fraction of
abnormal gradient components by regularizing the weights $\nu_{i}$,
possibly via data-dependent trimming rules. This gives rise to the update
rule of TWF:
\begin{eqnarray}
	\boldsymbol{z}^{(t+1)} & = & \boldsymbol{z}^{(t)}+ \frac{\mu_{t}}{m}\nabla\ell_{\mathrm{tr}}(\boldsymbol{z}^{(t)}), \quad\forall t\in\mathbb{N},		\label{eq:TWF-update}
\end{eqnarray}
where $\nabla\ell_{\mathrm{tr}}\left(\cdot\right)$ denotes the 
  regularized gradient given by\footnote{In the complex-valued case,
  the trimming rule is enforced upon the Wirtinger derivative, which
  reads $\nabla\ell_{\mathrm{tr}}\left(\boldsymbol{z}\right) :=
  \sum_{i=1}^{m} 2\frac{y_{i} - |\boldsymbol{z}^{*}\boldsymbol{a}_i
    |^{2}}{\boldsymbol{z}^{*}\boldsymbol{a}_i } \boldsymbol{a}_{i}{\bf
    1}_{\mathcal{E}_{1}^{i}\left(\boldsymbol{z}\right)\cap\mathcal{E}_{2}^{i}\left(\boldsymbol{z}\right)}.
  $ }
\begin{equation}
	\nabla\ell_{\mathrm{tr}}\left(\boldsymbol{z}\right)
	:= \sum_{i=1}^{m} 2\frac{y_{i} - |\boldsymbol{a}_i^{\top}\boldsymbol{z} |^{2}}{\boldsymbol{a}_i^{\top}\boldsymbol{z} }
	  \boldsymbol{a}_{i}{\bf 1}_{\mathcal{E}_{1}^{i}\left(\boldsymbol{z}\right)\cap\mathcal{E}_{2}^{i}\left(\boldsymbol{z}\right)}.
	\label{eq:score-truncate}
\end{equation}
 for some trimming criteria specified by $\mathcal{E}_{1}^{i}\left(\cdot\right)$
and $\mathcal{E}_{2}^{i}\left(\cdot\right)$. In our algorithm, we
take $\mathcal{E}_{1}^{i}\left(\boldsymbol{z}\right)$ and $\mathcal{E}_{2}^{i}\left(\boldsymbol{z}\right)$
to be two collections of events given by
\begin{eqnarray}
	\mathcal{E}_{1}^{i} (\boldsymbol{z}) & := & \bigg\{ \alpha_{z}^{\text{lb}}\leq\frac{\left|\boldsymbol{a}_i^{\top}\boldsymbol{z} \right|}{\Vert \boldsymbol{z} \Vert }\leq\alpha_{z}^{\text{ub}}\bigg\} ,
	\label{eq:defn-E1}\\
	\mathcal{E}_{2}^{i} (\boldsymbol{z}) & := & \bigg\{ |y_{i}-|\boldsymbol{a}_i^{\top}\boldsymbol{z} |^{2} |\leq\frac{\alpha_{h}}{m}\left\Vert \boldsymbol{y}-\mathcal{A}\left(\boldsymbol{z}\boldsymbol{z}^{\top}\right) \right\Vert _1 \frac{\left|\boldsymbol{a}_i^{\top}\boldsymbol{z}\right|} {\Vert \boldsymbol{z}\Vert } \bigg\} ,
	\label{eq:defn-E2}
\end{eqnarray}
%
% \ejc{I suppose we really want the red factor right?}\yxc{The red
%   factor gives us a better constant factor; but it doesn't affect
%   the order of the results; a remark is included at the end of
%   Section 2.1 for this matter}
where $\alpha_{z}^{\text{lb}}$, $\alpha_{z}^{\text{ub}}$, $\alpha_z$ are predetermined thresholds.
To keep notation light, we shall use $\mathcal{E}_{1}^{i}$ and $\mathcal{E}_{2}^{i}$
rather than $\mathcal{E}_{1}^{i}\left(\boldsymbol{z}\right)$ and
$\mathcal{E}_{2}^{i}\left(\boldsymbol{z}\right)$ whenever it is clear
from context. 

We emphasize that the above trimming procedure simply throws away those components whose
weights $\nu_{i}$'s fall outside some confidence range, so as to
remove the influence of outlier components. To achieve this, we
regularize both the numerator and denominator of $\nu_{i}$ by
enforcing separate trimming rules. Recognize that for any fixed
$\boldsymbol{z}$, the denominator obeys
\[	
	\mathbb{E}\left[\left|\boldsymbol{a}_i^{\top}\boldsymbol{z} \right|\right] = \sqrt{2/\pi} \Vert \boldsymbol{z} \Vert ,
\]
leading up to the rule (\ref{eq:defn-E1}). Regarding the numerator, by the law of large numbers
one would expect
\[
	\mathbb{E}\left[\left|y_{i}-|\boldsymbol{a}_i^{\top}\boldsymbol{z} |^{2}\right|\right]
	\approx \frac{1}{m}\left\Vert \boldsymbol{y}-\mathcal{A}\left(\boldsymbol{z}\boldsymbol{z}^{\top}\right)\right\Vert _{1},
\]
and hence it is natural to regularize the numerator by ensuring
\[
	\left|y_{i}-|\boldsymbol{a}_i^{\top}\boldsymbol{z} |^{2}\right|
	\lesssim \frac{1}{m}\left\Vert \boldsymbol{y}-\mathcal{A}\left(\boldsymbol{z}\boldsymbol{z}^{\top}\right)\right\Vert _{1}.
\]
As a remark, we include an extra term ${\left|\boldsymbol{a}_i^{\top}\boldsymbol{z} \right|}/{\left\Vert \boldsymbol{z}\right\Vert }$
in (\ref{eq:defn-E2}) to sharpen the theory, but all our results
continue to hold (up to some modification of constants) if we drop
this term in (\ref{eq:defn-E2}). Detailed procedures are summarized in Algorithm \ref{alg:TWF} \footnote{Careful readers might note that we include some extra factor $\frac{\sqrt{n}}{\|\boldsymbol{a}_i\|}$ (which  is approximately 1 in the Gaussian model) in Algorithm \ref{alg:TWF}. This occurs since we present Algorithm \ref{alg:TWF} in a more general fashion that applies beyond the model $\boldsymbol{a}_i \sim \mathcal{N}({\bf 0}, \boldsymbol{I})$, but all results / proofs continue to hold in the presence of this extra term. }.

The proposed paradigm could be counter-intuitive at first glance, since one might expect the larger terms
to be better aligned with the desired search direction. The issue, however, is that the large terms
are extremely volatile and could have too high of a leverage on the descent directions.  In contrast,
TWF discards these high-leverage data, which slightly increases the bias but
remarkably reduces the variance of the descent direction. We expect such gradient regularization
and variance reduction schemes to be beneficial for solving a broad family of nonconvex problems.

\newsavebox\Yx 
\begin{lrbox}{\Yx}   
\begin{minipage}{\textwidth}    
\begin{equation}
	\boldsymbol{Y} = \frac{1}{m}\sum_{i=1}^{m}y_{i}\boldsymbol{a}_{i}\boldsymbol{a}_{i}^{*} 
	{\bf 1}_{\{ |y_{i}|\leq\alpha_{y}^{2}\lambda_{0}^{2}\}}.
	\label{eq:TruncatedDual}
\end{equation}
\end{minipage} 
\end{lrbox}

\newsavebox\TWFupdate 
\begin{lrbox}{\TWFupdate}   
\begin{minipage}{\textwidth}
\begin{eqnarray}
	\boldsymbol{z}^{(t+1)} & = & \boldsymbol{z}^{(t)} +
		\frac{2\mu_{t}}{m} \sum_{i=1}^{m} \frac{y_{i}-\left| \boldsymbol{a}_{i}^* \boldsymbol{z}^{(t)}\right|^{2}}{\boldsymbol{z}^{(t)*}\boldsymbol{a}_{i}}\boldsymbol{a}_{i}
		{\bf 1}_{\mathcal{E}_{1}^{i}\cap\mathcal{E}_{2}^{i}},
\end{eqnarray}
\end{minipage} 
\end{lrbox}

\newsavebox\Truncate 
\begin{lrbox}{\Truncate}   
\begin{minipage}{\textwidth}
\begin{equation}
	\mathcal{E}_{1}^{i} := \left\{ \alpha_z^{\text{lb}} \leq 
		\frac{\sqrt{n}}{ \Vert \boldsymbol{a}_{i}\Vert }\frac{ |\boldsymbol{a}_{i}^* \boldsymbol{z}^{(t)} |}{\Vert \boldsymbol{z}^{(t)} \Vert } 
		\leq \alpha_{z}^{\text{ub}}\right\} ,
	\quad\quad 
	\mathcal{E}_{2}^{i} := \left\{ |y_{i}-|\boldsymbol{a}_{i}^* \boldsymbol{z}^{(t)}|^2 |
				\leq \alpha_h K_t \frac{\sqrt{n}}{ \Vert \boldsymbol{a}_{i} \Vert }\frac{ |\boldsymbol{a}_{i}^* \boldsymbol{z}^{(t)} |}{ \Vert \boldsymbol{z}^{(t)} \Vert }\right\} ,
\end{equation}
\end{minipage} 
\end{lrbox}

\begin{algorithm}[t]	
	\caption{Truncated Wirtinger Flow. \label{alg:TWF}}
	\begin{tabular}{>{\raggedright}p{1\textwidth}}
		\textbf{Input}: Measurements $\left\{ y_{i}\mid1\leq i\leq m\right\} $
			and sampling vectors $\left\{ \boldsymbol{a}_{i}\mid1\leq i\leq m\right\} $; 
			trimming thresholds $\alpha_z^{\mathrm{lb}}$,  $\alpha_z^{\mathrm{ub}}$, $\alpha_h$, and $\alpha_y$ 
			(see default values in Table \ref{table:pars}). 
		\vspace{0.7em}\tabularnewline
		\textbf{Initialize} $\boldsymbol{z}^{(0)}$ to be $\sqrt{ \frac{mn}{ \sum_{i=1}^{m}\left\Vert \boldsymbol{a}_{i}\right\Vert ^{2}} }  \lambda_{0}\tilde{\boldsymbol{z}}$,
			where $\lambda_{0}=\sqrt{\frac{1}{m}  \sum_{i=1}^{m}y_{i}}$
			and $\tilde{\boldsymbol{z}}$ is the leading eigenvector of \usebox{\Yx} \vspace{0.3em}\tabularnewline
		\textbf{Loop: for $t=0:T$ do} \usebox{\TWFupdate} \tabularnewline
			$\quad$ where 
			\usebox{\Truncate} \tabularnewline
			\[
			    \text{and}\quad K_{t}:= \frac{1}{m} \sum_{l=1}^{m}
					      \big| y_{l}-|\boldsymbol{a}_{l}^{*}\boldsymbol{z}^{(t)}|^{2} \big|.
			\]
		\textbf{Output }$\boldsymbol{z}_{T}$.\tabularnewline
	\end{tabular}
\end{algorithm}

\subsection{Truncated spectral initialization}

In order for the gradient stage to converge rapidly, we
need to seed it with a suitable initialization. One natural alternative
is the spectral method adopted in \cite{netrapalli2013phase,candes2014wirtinger},
which amounts to computing the leading eigenvector of 
$\widetilde{\boldsymbol{Y}}:=\frac{1}{m}\sum_{i=1}^{m}y_{i}\boldsymbol{a}_{i}\boldsymbol{a}_{i}^{\top}$.
This arises from the observation that when $\boldsymbol{a}_{i}\sim\mathcal{N}\left({\bf 0},\boldsymbol{I}\right)$
and $\left\Vert \boldsymbol{x}\right\Vert =1$, 
\[
	\mathbb{E}[\widetilde{\boldsymbol{Y}}] = \boldsymbol{I} + 2\boldsymbol{x}\boldsymbol{x}^{\top},
\]
whose leading eigenvector is exactly $\boldsymbol{x}$ with an eigenvalue of 3. 

Unfortunately, this spectral technique converges to a good initial
point only when $m\gtrsim n\log n$, due to the fact that $(
\boldsymbol{a}_i ^{\top}\boldsymbol{x} )^2 \boldsymbol{a}_i
\boldsymbol{a}_i^{\top} $ is heavy-tailed, a random quantity which
does not have a moment generating function.  To be more precise,
consider the noiseless case $y_i = |\boldsymbol{a}_i ^ {\top}
\boldsymbol{x}|^2$ and recall that $\max_{i} y_{i}\approx 2\log
m$. Letting $k=\arg\max_{i}y_{i}$, one can calculate
\[
	\left(  \frac{\boldsymbol{a}_k} {\Vert \boldsymbol{a}_k \Vert }\right)^{\top} \widetilde{\boldsymbol{Y}} 
		\frac{\boldsymbol{a}_k} {\Vert \boldsymbol{a}_k \Vert }
	\geq \left(\frac{\boldsymbol{a}_{k}}{\Vert \boldsymbol{a}_{k} \Vert }\right)^{\top} 
		\left( \frac{1}{m}\boldsymbol{a}_{k}\boldsymbol{a}_{k}^{\top}\right)
		\left(\boldsymbol{a}_k^{\top} \boldsymbol{x}\right)^{2}\left(\frac{\boldsymbol{a}_k}{\Vert \boldsymbol{a}_k \Vert }\right)
	\approx \frac{ 2n \log m}{m},
\]
which is much larger than
$\boldsymbol{x}^{\top}\widetilde{\boldsymbol{Y}}\boldsymbol{x}=3$
unless $m/n$ is very large. This tells us that in the regime where
$m\asymp n$, there exists some unit vector $\boldsymbol{a}_k /
\|\boldsymbol{a}_k\|$ that is closer to the leading eigenvector of
$\widetilde{\boldsymbol{Y}}$ than $\boldsymbol{x}$. This phenomenon happens because the summands of $\widetilde{\boldsymbol{Y}}$ have huge tails so that even one large term could end up dominating the empirical sum,  
thus preventing the spectral method from returning a meaningful initial
guess.

To address this issue, we propose a more robust version of spectral method, which discards those observations
$y_{i}$ that are several times larger than the mean during spectral
initialization.  Specifically, the initial estimate is obtained by
computing the leading eigenvector $\tilde{\boldsymbol{z}}$ of the truncated sum
\begin{equation}
	\boldsymbol{Y}:=\frac{1}{m}\sum_{i=1}^{m}y_{i}\boldsymbol{a}_{i}\boldsymbol{a}_{i}^{\top}
		{\bf 1}_{\left\{ |y_{i}| \leq \alpha_{y}^2 \left(\frac{1}{m}\sum_{l=1}^{m}y_{l}\right)\right\} }
\end{equation}
for some predetermined threshold $\alpha_y$, and then rescaling
$\tilde{\boldsymbol{z}}$ so as to have roughly the same norm as
$\boldsymbol{x}$ (which is estimated to be $\frac{1}{m}\sum_{l=1}^m
y_l$); see Algorithm \ref{alg:TWF} for the detailed procedure.

Notably, the aforementioned drawback of the spectral method is not merely a theoretical
concern but rather a substantial practical issue. We have seen this in
Fig.~\ref{fig:real-image} (main quad example) showing the enormous
advantage of truncated spectral initialization.  This is also further illustrated
in Fig.~\ref{fig:spectral-method}, which compares the empirical
efficiency of both methods with $\alpha_{y}=3$ set to be the
truncation threshold. For both Gaussian designs and CDP models, the
empirical loss incurred by the original spectral method increases as
$n$ grows, which is in stark constrast to the truncated spectral
method that achieves almost identical accuracy over the same range of
$n$.

\begin{figure}[h]
\centering
\begin{tabular}{cc}
	\includegraphics[width=0.47\textwidth]{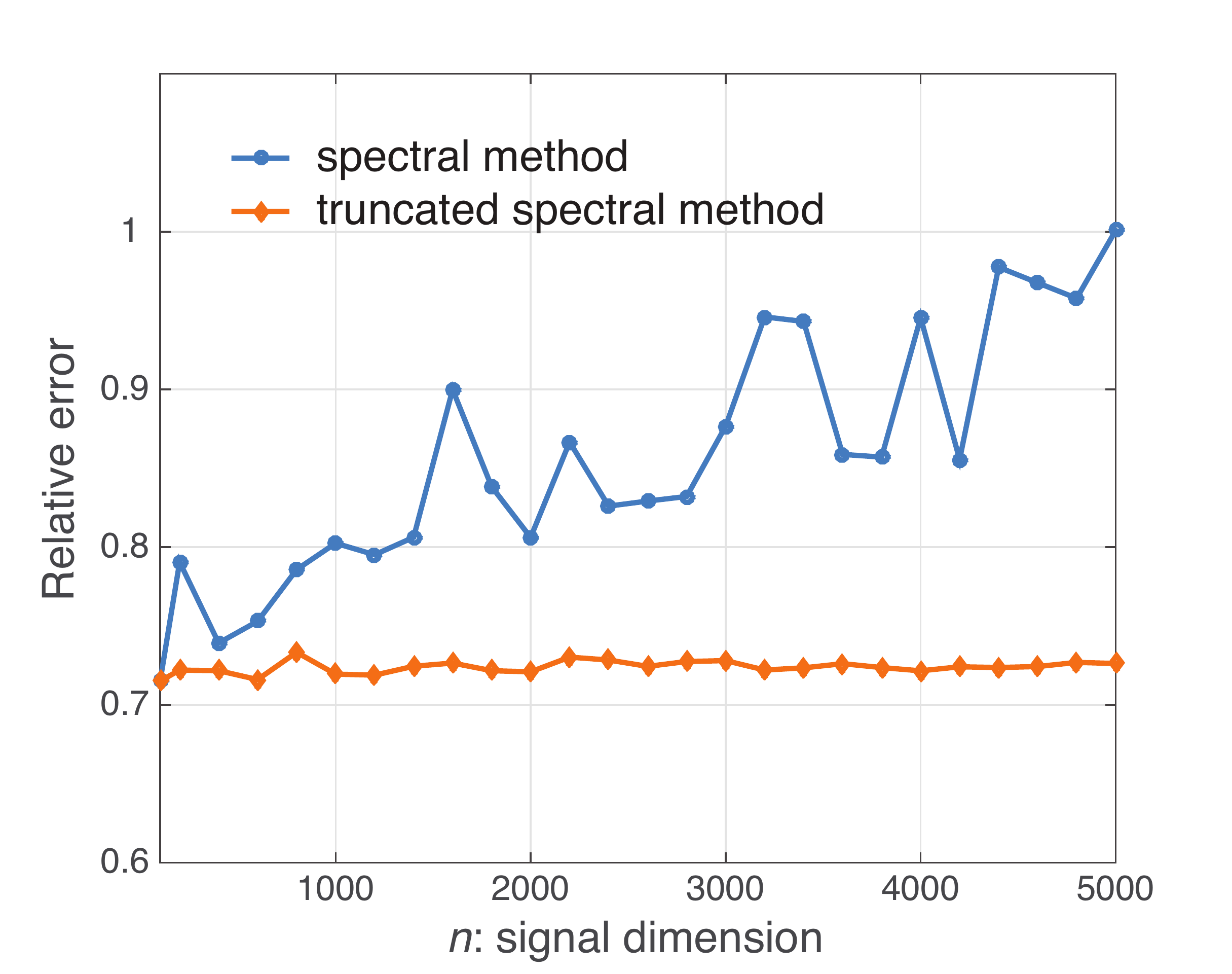} 
	& \includegraphics[width=0.47\textwidth]{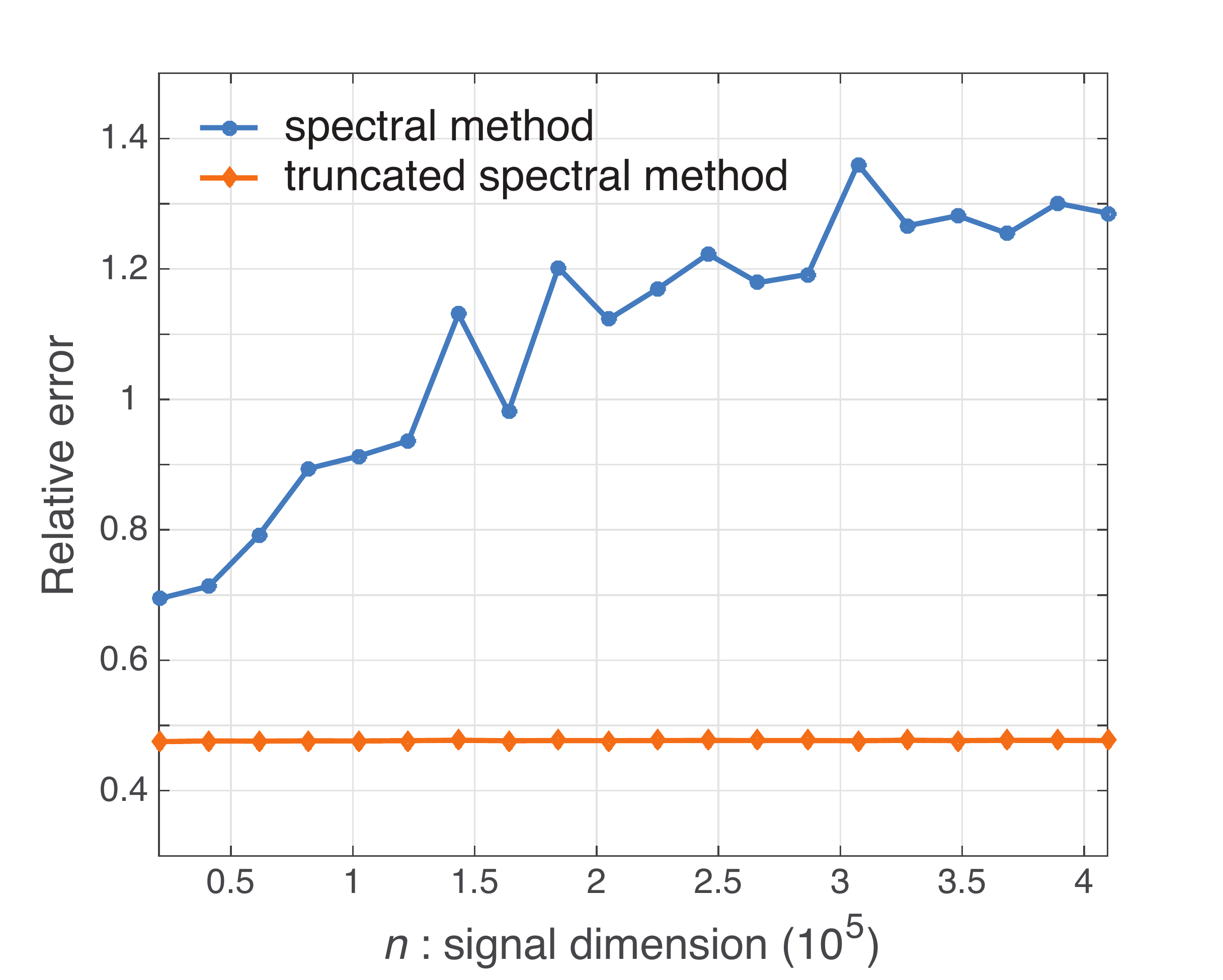}\tabularnewline
	(a) & (b)\tabularnewline
\end{tabular}
\caption{The empirical relative error for both the spectral and the
  truncated spectral methods. The results are averaged over 50 Monte
  Carlo runs, and are shown for: (a) 1-D Gaussian measurement where
  $\boldsymbol{a}_i \sim \mathcal{N}({\bf 0},\boldsymbol{I})$ and
  $m=6n$; (b) 2-D CDP model (\ref{eq:CDP}) where the diagonal entries of
  $\boldsymbol{D}^{(l)}$ are uniformly drawn from
  $\left\{ 1,-1,j,-j\right\} $, $n=n_1 \times n_2$ with $n_1=300$ and
  $n_2$ ranging from 64 to 1280, and $m=12n$.}
\label{fig:spectral-method}
\end{figure}

\subsection{Choice of algorithmic parameters} \label{sec:pars}

\begin{table}[t]
 %  \vspace{0.3em}
\begin{tabular}{>{\centering}p{1\textwidth}}
\tabularnewline
%\hline 
\end{tabular}

\vspace{-0.5em}\caption{Range of algorithmic parameters \label{table:pars}}

  \begin{tabular}{>{\raggedright}p{1\textwidth}}
    \hline 
    %\vspace{-1em}     
    %\vspace{-5em}
    %\tabularnewline 
    %\hline 
    \vspace{0.1em}
    (a) {\bf When a fixed step size $\mu_t\equiv\mu$ is employed: } $(\alpha_z^{\text{lb}}, \alpha_z^{\text{ub}}, \alpha_h, \alpha_y)$ obeys
    \tabularnewline
    \begin{equation}
	\begin{cases}
	\zeta_{1}\text{ } := \max \left\{ 
	  \mathbb{E}\Big[ \xi^{2}{\bf 1}_{\left\{ |\xi| \leq \sqrt{1.01}\alpha_{z}^{\mathrm{lb}} \text{ or } |\xi| \geq \sqrt{0.99}\alpha_{z}^{\mathrm{ub}}\right\} } \Big],  
	  \mathbb{P} \left( |\xi| \leq \sqrt{1.01}\alpha_{z}^{\mathrm{lb}} \text{ or } |\xi|\geq \sqrt{0.99}\alpha_{z}^{\mathrm{ub}} \right) 
	\right\} \vspace{0.3em}\\
	\zeta_{2}\text{ } := \text{ } \mathbb{E}\left[ \xi^{2}{\bf 1}_{ \{ |\xi| > 0.473 \alpha_{h} \} } \right],\vspace{0.3em} \\
	2(\zeta_{1} + \zeta_{2} ) + \sqrt{ 8/(9\pi)}\alpha_{h}^{-1} < 1.99, \vspace{0.3em}\\
	\alpha_y \geq 3,
	%, \quad \alpha_p \geq 3,
	\end{cases}
	\label{eq:Condition-tau_z}
    \end{equation}
    where $\xi\sim\mathcal{N}(0,1)$.  By default, $\alpha_z^{\mathrm{lb}}=0.3$, $\alpha_z^{\mathrm{ub}} = \alpha_h = 5$, and $\alpha_y=3$.
    \vspace{0.1em}
    \tabularnewline
    \hline\vspace{0.1em}
    (b) {\bf When $\mu_t$ is chosen by a backtracking line search: } $(\alpha_z^{\text{lb}}, \alpha_z^{\text{ub}}, \alpha_h, \alpha_y, \alpha_p)$ obeys
    \begin{equation}
	0<\alpha_{z}^{\mathrm{lb}} \leq 0.1, \quad \alpha_{z}^{\mathrm{ub}} \geq 5,\quad \alpha_{h}\geq 6, \quad \alpha_y \geq 3, \quad\text{and}\quad \alpha_p \geq 5.
	\label{eq:Cleaner-pars}
    \end{equation}
    By default, $\alpha_z^{\mathrm{lb}}=0.1$, $\alpha_z^{\mathrm{ub}} = 5$,  $\alpha_h = 6$, $\alpha_y=3$, and $\alpha_p = 5$.
    \vspace{0.2em}
    \tabularnewline
    \hline
  \end{tabular}
  
\end{table}

One implementation detail to specify is the step size $\mu_{t}$
at each iteration $t$. There are two alternatives that work well
in both theory and practice: 

\begin{enumerate}
\item \textbf{Fixed step size}. Take $\mu_{t}\equiv\mu$ ($\forall
  t\in\mathbb{N}$) for some constant $\mu>0$. As long as $\mu$ is not
  too large, our main results state that this strategy always
  works---although the convergence rate depends on $\mu$. Under appropriate conditions,
  our theorems hold for any constant $0 < \mu < 0.28$.
\item \textbf{Backtracking line search with truncated objective}. This strategy
performs a line search along the descent direction 
\[
	\boldsymbol{p}_{t}:=\frac{1}{m} \nabla\ell_{\mathrm{tr}} (\boldsymbol{z}_{t})
\]
and determines an appropriate step size that guarantees a sufficient
improvement. In contrast to the conventional search strategy that
determines the sufficient progress with respect to the true objective
function, we propose to evaluate instead a regularized version of the
objective function. Specifically, put
\begin{equation}
	\widehat{\ell} (\boldsymbol{z}) := \sum_{i\in\widehat{\mathcal{T}} (\boldsymbol{z})}
	\left\{ y_{i}\log (|\boldsymbol{a}_{i}^{\top}\boldsymbol{z}|^2 ) - |\boldsymbol{a}_{i}^{\top}\boldsymbol{z} |^2 \right\} ,
	\label{eq:ell_hat}
\end{equation}
where
\[
	\widehat{\mathcal{T}} (\boldsymbol{z}) := \left\{ i\mid\left|\boldsymbol{a}_{i}^{\top}\boldsymbol{z}\right|\geq\alpha_{z}^{\mathrm{lb}}\Vert \boldsymbol{z} \Vert \text{ and }\left|\boldsymbol{a}_{i}^{\top}\boldsymbol{p}\right|\leq\alpha_{p} \Vert \boldsymbol{p} \Vert \right\} .
\]
Then the backtracking line search proceeds as
\begin{enumerate}
	\item Start with $\tau=1$;
	\item Repeat $\tau\leftarrow\beta\tau$ until 
	\begin{equation}
		\frac{1}{m} \widehat{\ell}\big(\boldsymbol{z}^{(t)} + \tau\boldsymbol{p}^{(t)}\big)
		\geq \frac{1}{m} \widehat{\ell} \big( \boldsymbol{z}^{(t)} \big)
		+ \frac{1}{2}\tau \big\Vert \boldsymbol{p}^{(t)} \big\Vert ^2,
		\label{eq:backtrack}
	\end{equation}
	where $\beta\in(0,1)$ is some pre-determined constant;
	\item Set $\mu_{t}=\tau$.
\end{enumerate}
By definition (\ref{eq:ell_hat}), evaluating
$\widehat{\ell}(\boldsymbol{z}^{(t)}+\tau\boldsymbol{p}^{(t)})$ mainly
consists in calculating the matrix-vector product
$\boldsymbol{A}(\boldsymbol{z}^{(t)}+\tau\boldsymbol{p}^{(t)})$. In
total, we are going to evaluate
$\widehat{\ell}(\boldsymbol{z}^{(t)}+\tau\boldsymbol{p}^{(t)})$ for
$O \big(\log ({1}/{\beta}) \big)$ different $\tau$'s, and hence
the total cost amounts to computing
$\boldsymbol{A}\boldsymbol{z}^{(t)}$,
$\boldsymbol{A}\boldsymbol{p}^{(t)}$ as well as
$O (m\log({1}/{\beta}))$ additional flops.  Note that the
matrix-vector products $\boldsymbol{A}\boldsymbol{z}^{(t)}$ and
$\boldsymbol{A}\boldsymbol{p}^{(t)}$ need to be computed even when one
adopts a pre-determined step size. Hence, the extra cost incurred by a
backtracking line search, which is $O (m\log({1}/{\beta}))$
flops, is negligible compared to that of computing the gradient even
once.
\end{enumerate}

Another set of important algorithmic parameters to determine is the
trimming thresholds $\alpha_{h}$, $\alpha_{z}^{\mathrm{lb}}$, $\alpha_{z}^{\mathrm{ub}}$, $\alpha_y$, and $\alpha_p$ (for a backtracking line search only). The present paper isolates the set of
$(\alpha_{h},\alpha_{z}^{\mathrm{lb}},\alpha_{z}^{\mathrm{ub}}, \alpha_y)$
obeying (\ref{eq:Condition-tau_z}) as given in Table \ref{table:pars} when a fixed step size is employed.  More concretely, this range
subsumes as special cases all parameters obeying the following constraints:
    \begin{equation}
	0<\alpha_{z}^{\mathrm{lb}}\leq0.5,\quad\alpha_{z}^{\mathrm{ub}}\geq5,\quad\alpha_{h}\geq5, \quad\text{and}\quad \alpha_y \geq 3.
    \end{equation}
When a backtracking line search is adopted,  an extra parameter $\alpha_p$ is needed, which we take to be $\alpha_p \geq 5$.  In all theory presented herein, we assume that the parameters fall within the range singled out in  Table \ref{table:pars}.

\section{Why TWF works?}
\label{sec:Why-it-works}

Before proceeding, it is best to develop an intuitive understanding of
the TWF iterations. We start with a notation representing the
(unrecoverable) global phase \cite{candes2014wirtinger} for real-valued data
%
%\[
%	\phi\left(\boldsymbol{z}\right):=\arg\min\nolimits_{\varphi:\in[0,2\pi)}\|e^{-j\varphi}\boldsymbol{z} - \boldsymbol{x}\|.
%\]
% 
\begin{equation}
	\begin{array}{l}
	\phi\left(\boldsymbol{z}\right):=
	\begin{cases}
		0,\quad & \text{if }\left\Vert \boldsymbol{z}-\boldsymbol{x}\right\Vert \leq\left\Vert \boldsymbol{z}+\boldsymbol{x}\right\Vert ,\\
		\pi, & \text{else}.
	\end{cases}
	\end{array}
	\label{eq:phase-defn}
\end{equation}
It is self-evident that
% although we are unable to determine
% $\phi\left(\boldsymbol{z}\right)$, 
%
\[
	(-\boldsymbol{z}) + \frac{\mu}{m}\nabla_{\mathrm{tr}}\ell\big(  -\boldsymbol{z} \big)
	= - \left\{ \boldsymbol{z} + \frac{\mu}{m}\nabla_{\mathrm{tr}}\ell(\boldsymbol{z})\right\} ,
\]
and hence (cf. Definition (\ref{eq:defn-dist})) 
\[
	\mathrm{dist}\left( (-\boldsymbol{z}) + \frac{\mu}{m}\nabla_{\mathrm{tr}}\ell( -\boldsymbol{z} ),\text{}\boldsymbol{x}\right)
	= \mathrm{dist}\left(\boldsymbol{z}+\frac{\mu}{m}\nabla_{\mathrm{tr}}\ell\left(\boldsymbol{z}\right),\text{}\boldsymbol{x}\right)
\]
despite the global phase uncertainty. For simplicity of presentation, we shall drop the phase term by letting
$\boldsymbol{z}$ be $e^{-j\phi\left(\boldsymbol{z}\right)}\boldsymbol{z}$
and setting $\boldsymbol{h}=\boldsymbol{z}-\boldsymbol{x}$, whenever it
is clear from context.

%\subsection{Descent direction and regularity condition}

The first object to consider is the descent direction. To this end, we
find it convenient to work with a fixed $\boldsymbol{z}$ independent
of the design vectors $\boldsymbol{a}_{i}$, which is of course
heuristic but helpful in developing some intuition. Rewrite
\begin{eqnarray}
	\nabla\ell_i ( \boldsymbol{z} )  
	& = &  2 \frac{( \boldsymbol{a}_i ^{\top}\boldsymbol{x} )^2 - ( \boldsymbol{a}_i ^{\top}\boldsymbol{z} )^{2}}{\boldsymbol{a}_i^{\top} \boldsymbol{z}}  \boldsymbol{a}_i  
 	\text{ } \overset{(\text{i})}{=} \text{ } -2  \frac{ (\boldsymbol{a}_{i}^{\top}\boldsymbol{h} ) ( 2 \boldsymbol{a}_{i}^{\top} \boldsymbol{z} - \boldsymbol{a}_{i}^{\top}\boldsymbol{h} ) } {\boldsymbol{a}_i^{\top} \boldsymbol{z}}\boldsymbol{a}_i  
	\nonumber\\ 
&	= & - 4 (\boldsymbol{a}_{i}^{\top}\boldsymbol{h} )\boldsymbol{a}_{i}  +
		\underset{ := \boldsymbol{r}_i } {\underbrace{ 2 \frac{(\boldsymbol{a}_{i}^{\top}\boldsymbol{h} )^2 }{\boldsymbol{a}_i ^{\top}\boldsymbol{z}}\boldsymbol{a}_i  }}, 
 	%& = & -4\boldsymbol{h}+\boldsymbol{r},
	\label{eq:mean-ell}
\end{eqnarray}
where (i) follows from the identity $a^{2}-b^{2}=(a+b)(a-b)$. The
first component of (\ref{eq:mean-ell}), which on average gives
$-4\boldsymbol{h}$, makes a good search direction when averaged over
all the observations $i = 1, \ldots, m$.  The issue is that the other
term $\boldsymbol{r}_i$---which is in general non-integrable---could
be devastating. The reason is that
$\boldsymbol{a}_i^{\top}\boldsymbol{z}$ could be arbitrarily small, thus 
% \yxc{Carlos suggested using ``arbitrarily small'' here}, 
resulting in an unbounded $\boldsymbol{r}_i$. As a consequence, a
non-negligible portion of the $\boldsymbol{r}_i$'s may exert a very
strong influence on the descent direction in an undesired manner.

Such an issue can be prevented if one can detect and separate those
gradient components bearing abnormal $\boldsymbol{r}_i$'s. Since we
cannot observe the individual components of the decomposition
(\ref{eq:mean-ell}), we cannot reject indices with large values of
$\boldsymbol{r}_i$ directly. Instead, we examine each gradient
component as a whole and discard it if its size is not absolutely
controlled. Fortunately, such a strategy is sufficient to ensure that
most of the contribution from
the regularized gradient comes from the first
component of (\ref{eq:mean-ell}), namely,
$-4(\boldsymbol{a}_i^{\top}\boldsymbol{h}) \boldsymbol{a}_i$.  As will
be made precise in Proposition \ref{prop-regularity-noiseless} and
Lemma \ref{Lemma:norm-score}, the regularized gradient obeys
\begin{align}
	-\Big\langle \frac{1}{m}\nabla\ell_{\mathrm{tr}}(\boldsymbol{z}), \boldsymbol{h}  \Big\rangle  
	\text{ } \geq \text{ } (4-\epsilon) \| \boldsymbol{h} \|^2 &-  O\bigg(\frac{ \| \boldsymbol{h} \|^3}{ \| \boldsymbol{z} \|}\bigg)  \label{eq:mean-RC}\\
	\text{and}\quad  \Big\| \frac{1}{m} \nabla\ell_{\mathrm{tr}}(\boldsymbol{z}) \Big\| &\lesssim \| \boldsymbol{h} \|. \quad	
	\label{eq:norm-grad-tr}
\end{align}
Here, one has $(4-\epsilon)\|\boldsymbol{h}\|^2$ in (\ref{eq:mean-RC})
instead of $4\|\boldsymbol{h}\|^2$ to account for the bias introduced
by adaptive trimming, where $\epsilon$ is small as long as we only throw away
a small fraction of data. Looking at (\ref{eq:mean-RC}) and
(\ref{eq:norm-grad-tr}) we see that the search direction is
sufficiently aligned with the deviation
$-\boldsymbol{h}=\boldsymbol{x}-\boldsymbol{z}$ of the current
iterate; i.e. they form a reasonably good angle that is bounded away
from $90^{\circ}$. Consequently, $\boldsymbol{z}$ is expected to be
dragged towards $\boldsymbol{x}$ provided that the step size is
appropriately chosen.

The observations (\ref{eq:mean-RC}) and (\ref{eq:norm-grad-tr})  are reminiscent of a (local) regularity condition given in \cite{candes2014wirtinger}, 
which is a fundamental criterion that dictates
rapid convergence of iterative procedures (including WF and other
gradient descent schemes).  
When specialized to TWF,
we say that $-\frac{1}{m}\nabla\ell_{\mathrm{tr}}\left(\cdot\right)$
satisfies the \emph{regularity condition}, denoted by
$\mathsf{RC}\left(\mu,\lambda,\epsilon\right)$, if
\begin{eqnarray}
	\Big\langle \boldsymbol{h},-\frac{1}{m}\nabla\ell_{\mathrm{tr}} ( \boldsymbol{z} ) \Big\rangle  & \geq & \frac{\mu}{2}\Big\Vert \frac{1}{m}\nabla\ell_{\mathrm{tr}}\left(\boldsymbol{z}\right) \Big\Vert ^2 + \frac{\lambda}{2}\left\Vert \boldsymbol{h}\right\Vert ^{2}
	\label{eq:Regularity-Truncate-Original}
\end{eqnarray}
holds for all $\boldsymbol{z}$ obeying $\left\Vert \boldsymbol{z}-\boldsymbol{x}\right\Vert \leq\epsilon \Vert \boldsymbol{x} \Vert$,
where $0<\epsilon<1$ is some constant. Such an $\epsilon$-ball around
$\boldsymbol{x}$ forms a basin of attraction.
Formally, under $\mathsf{RC}\left(\mu,\lambda,\epsilon\right)$, a
little algebra gives
\begin{eqnarray}
	\mathrm{dist}^{2}\left(\boldsymbol{z}+\frac{\mu}{m}\nabla\ell_{\mathrm{tr}}\left(\boldsymbol{z}\right),\boldsymbol{x}\right) & \leq & \text{ }\left\Vert \boldsymbol{z}+\frac{\mu}{m}\nabla\ell_{\mathrm{tr}}\left(\boldsymbol{z}\right)-\boldsymbol{x}\right\Vert ^{2}\nonumber \\
 	& = & \left\Vert \boldsymbol{h}\right\Vert ^{2}+\left\Vert \frac{\mu}{m}\nabla\ell_{\mathrm{tr}}\left(\boldsymbol{z}\right)\right\Vert ^{2}+2\mu\left\langle \boldsymbol{h},\frac{1}{m}\nabla\ell_{\mathrm{tr}}\left(\boldsymbol{z}\right)\right\rangle \nonumber \\
 	& \leq & \left\Vert \boldsymbol{h}\right\Vert ^{2}+\left\Vert \frac{\mu}{m}\nabla\ell_{\mathrm{tr}}\left(\boldsymbol{z}\right)\right\Vert ^{2}-\mu^{2}\left\Vert \frac{1}{m}\nabla\ell_{\mathrm{tr}}\left(\boldsymbol{z}\right)\right\Vert ^{2}-\mu\lambda\left\Vert \boldsymbol{h}\right\Vert ^{2}\nonumber \\
 	& = & \left(1-\mu\lambda\right)\mathrm{dist}^{2}\left(\boldsymbol{z},\boldsymbol{x}\right)\label{eq:convergence-rate}
\end{eqnarray}
for any $\boldsymbol{z}$ with $\left\Vert
  \boldsymbol{z}-\boldsymbol{x}\right\Vert \leq\epsilon$.  In words,
the TWF update rule is locally contractive around the planted solution,
provided that $\mathsf{RC}\left(\mu,\lambda,\epsilon\right)$ holds for
some nonzero $\mu$ and $\lambda$.
Apparently, Conditions (\ref{eq:mean-RC}) and  (\ref{eq:norm-grad-tr}) already imply the validity of $\mathsf{RC}$ for some constants
$\mu, \lambda \asymp 1$ when $\|\boldsymbol{h}\| / \| \boldsymbol{z} \|$ is reasonably small,   which in turn allows us to take a constant step size $\mu$ and enables a constant contraction rate $1-\mu\lambda$.

\begin{figure}
	\centering
	\includegraphics[width=0.4\textwidth]{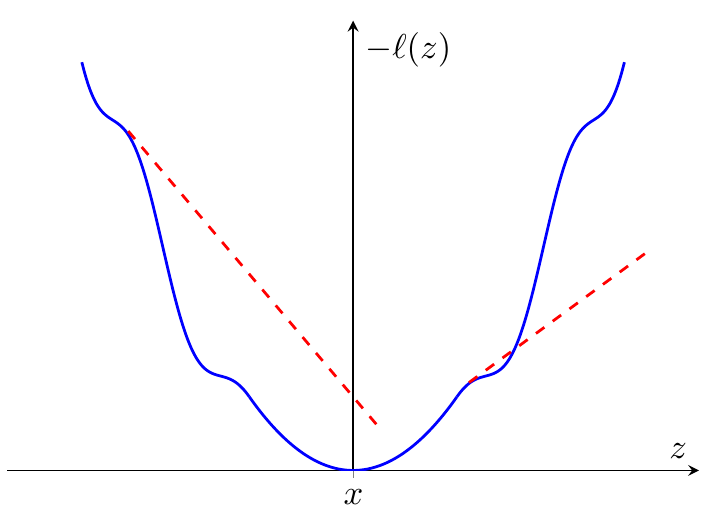}
	\caption{ A function
          $-\ell(z)$ satisfying $\mathsf{RC}$: $-\ell(z)=z^2$ for any
          $z \in [-6,6]$, and
          $-\ell(z) =
          z^2 + 1.5|z|\left(\cos( |z|-6 )-1\right)$
          otherwise. }
	\label{fig:regularity}
\end{figure}

Finally, caution must be exercised when connecting $\mathsf{RC}$ with strong
convexity, since the former does not necessarily guarantee the latter
within the basin of attraction.  As an illustration,
Fig.~\ref{fig:regularity} plots the graph of a non-convex function
obeying $\mathsf{RC}$.  The distinction stems from the fact that
$\mathsf{RC}$ is stated only for those pairs $\boldsymbol{z}$ and
$\boldsymbol{h}=\boldsymbol{z}-\boldsymbol{x}$ with $\boldsymbol{x}$
being a fixed component, rather than simultaneously accommodating all
possible $\boldsymbol{z}$ and
$\boldsymbol{h}=\boldsymbol{z}-\tilde{\boldsymbol{z}}$ with
$\tilde{\boldsymbol{z}}$ being an arbitrary vector.  In contrast,
$\mathsf{RC}$ says that the only stationary point of the truncated objective in a neighborhood of $\boldsymbol{x}$ is $\boldsymbol{x}$,
which often suffices
for a gradient-descent type scheme to succeed.  

\section{Numerical experiments\label{sec:Numerical-experiments}}

In this section, we report additional numerical results to verify the
practical applicability of TWF. In all numerical experiments conducted
in the current paper, we set
\begin{equation}
	\alpha_{z}^{\mathrm{lb}}=0.3,\quad\alpha_{z}^{\mathrm{ub}}=5, \quad\alpha_{h}=5, \quad\text{and}\quad \alpha_y = 3.
\end{equation}
This is a concrete combination of parameters
satisfying our condition (\ref{eq:Condition-tau_z}). Unless otherwise
noted, we employ 50 power iterations for initialization,
adopt a fixed step size $\mu_{t}\equiv0.2$ when updating TWF iterates,
and set the maximum number of iterations to be $T=1000$ for the iterative
refinement stage. 

The first series of experiments concerns exact recovery from
noise-free data. Set $n = 1000$ and  generate a real-valued
signal $\boldsymbol{x}$ at random. Then for $m$ varying between $2n$
and $6n$, generate $m$ design vectors $\boldsymbol{a}_{i}$
independently drawn from $\mathcal{N}\left({\bf
    0},\boldsymbol{I}\right)$.  An experiment is claimed to succeed if
the returned estimate $\hat{\boldsymbol{x}}$ satisfies
$\mathrm{dist}\left(\hat{\boldsymbol{x}},\boldsymbol{x}\right)/\left\Vert
  \boldsymbol{x}\right\Vert \leq10^{-5}$.
Fig.~\ref{Fig:PhaseTransition} illustrates the empirical success rate
of TWF (over 100 Monte Carlo trials for each $m$) revealing that exact
recovery is practially guaranteed from fewer than 1000 iterations when
the number of quadratic constraints is about 5 times the ambient
dimension.

\begin{figure}[t]
	\centering
	\includegraphics[width=0.5\textwidth]{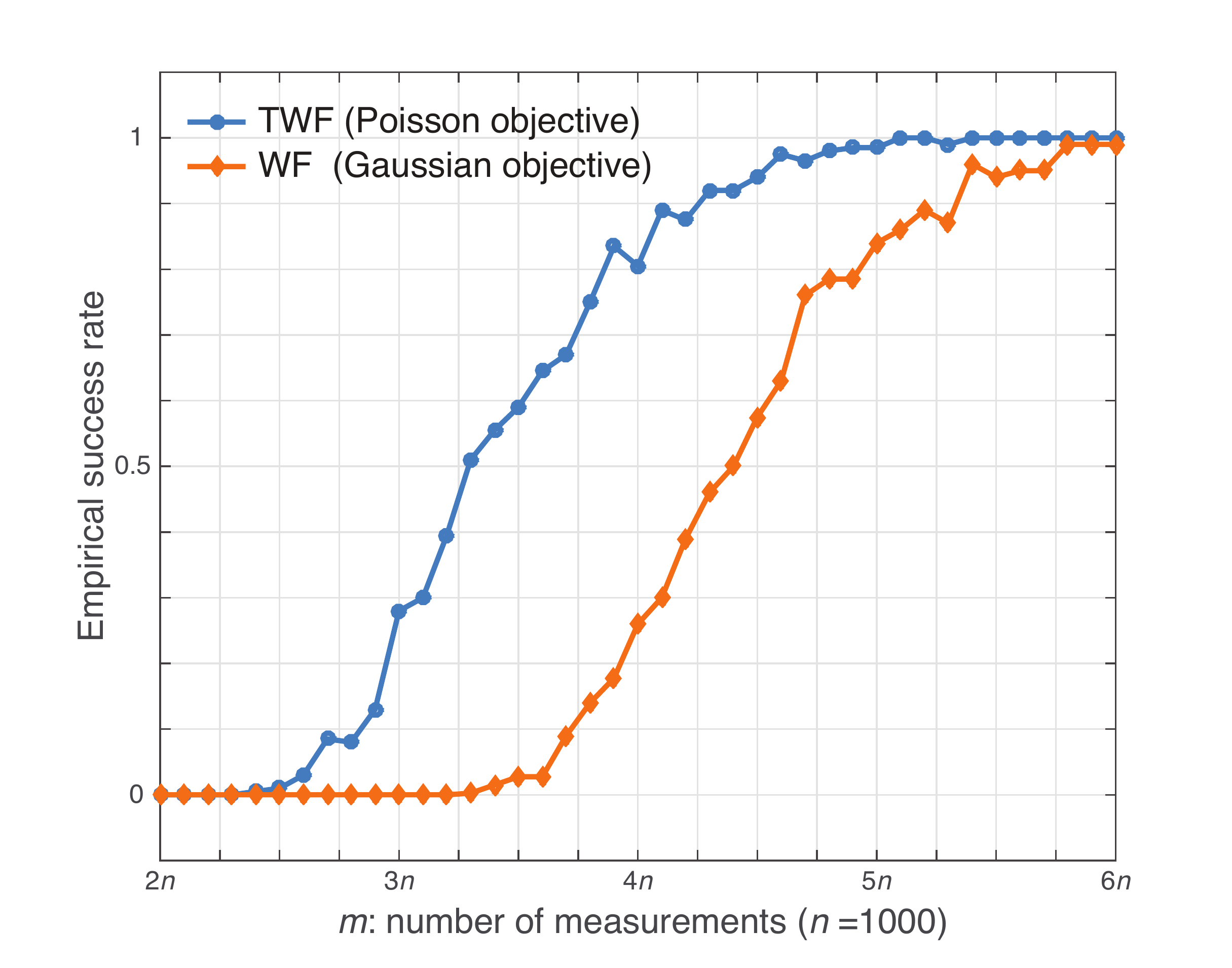}
	\caption{Empirical success rate under real-valued Gaussian sampling $\boldsymbol{a}_{i}\sim\mathcal{N}\left({\bf 0},\boldsymbol{I}_{n}\right)$.}
	\label{Fig:PhaseTransition}
\end{figure}

\begin{figure}[h]
	\centering
	\begin{tabular}{cc}
		  \includegraphics[width=0.5\textwidth]{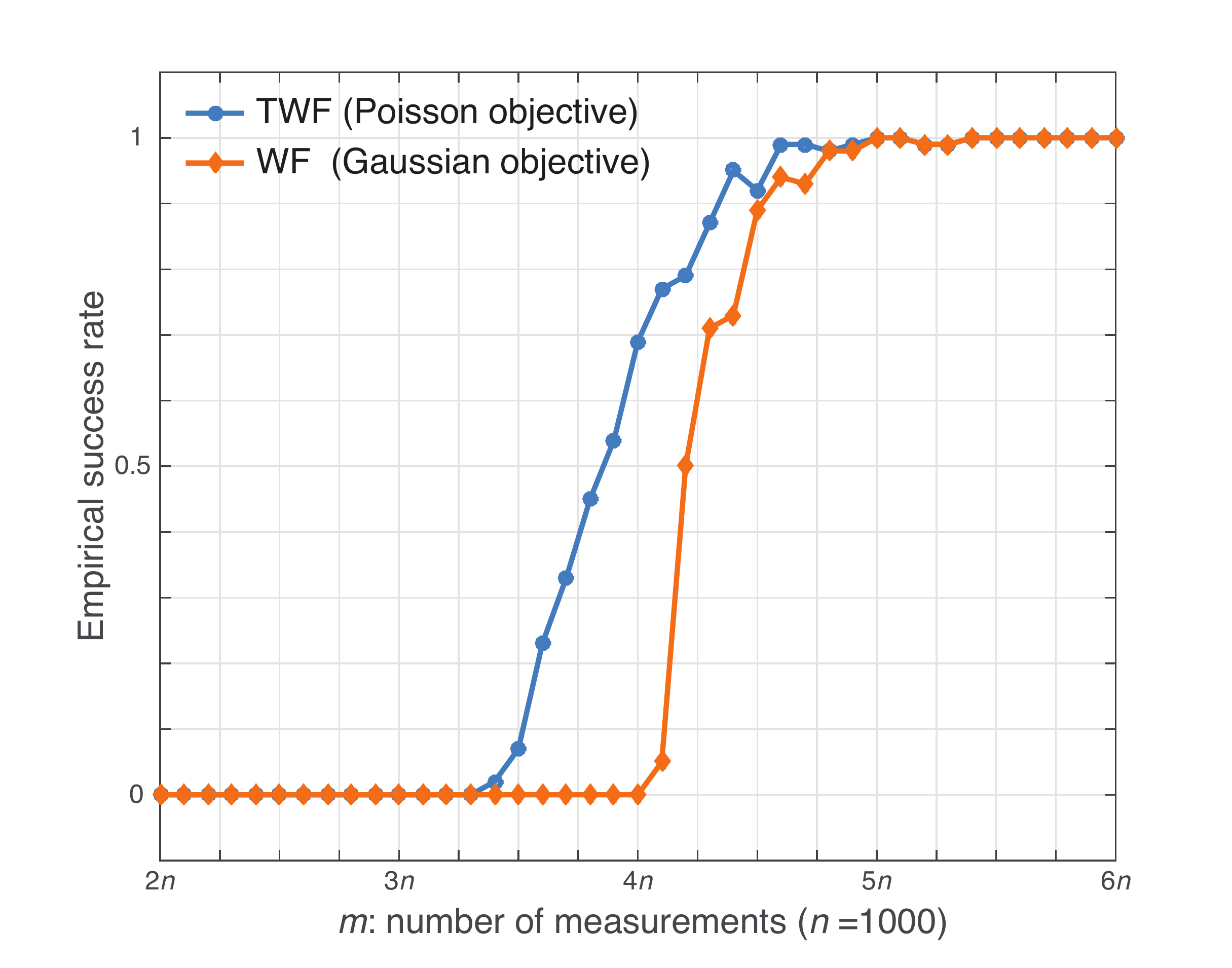} 
		& \includegraphics[width=0.5\textwidth]{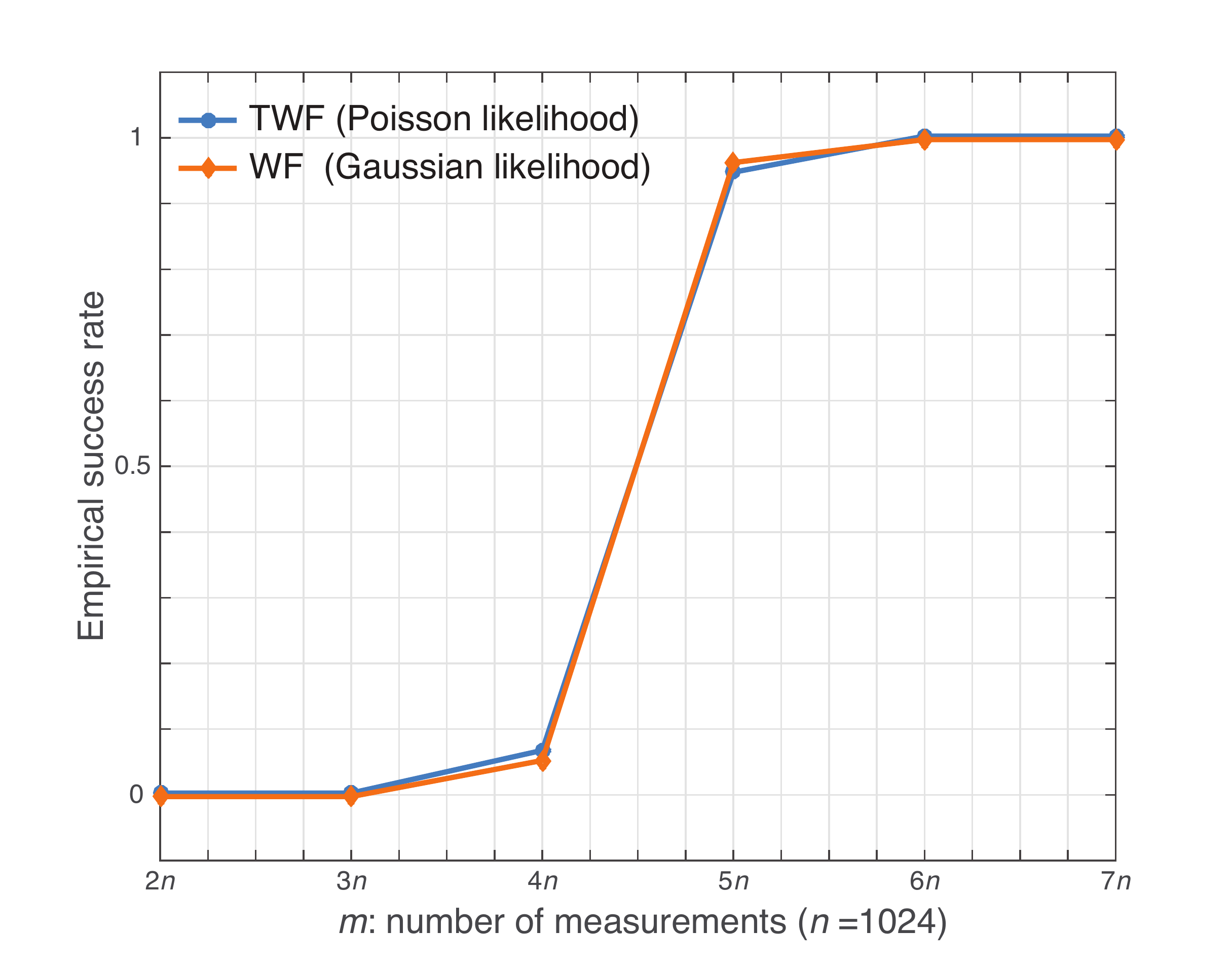}  \tabularnewline
		(a) & (b)  \tabularnewline
	\end{tabular}
	\caption{Empirical success rate for exact recovery using TWF. The results are shown for (a) complex-valued Gaussian sampling
	$\boldsymbol{a}_i \sim\mathcal{N}({\bf 0},\frac{1}{2}\boldsymbol{I}_n ) + j\mathcal{N} ({\bf 0},\frac{1}{2}\boldsymbol{I}_n )$, and (b) CDP with masks uniformly drawn from $\{ 1,-1,j,-j\} $.}
	\label{Fig:PhaseTransition2}
\end{figure}

To see how special the real-valued Gaussian designs are to our
theoretical finding, we perform experiments on two other types of
measurement models. In the first, TWF is applied to complex-valued
data by generating $\boldsymbol{a}_{i}\sim\mathcal{N}\left({\bf
    0},\frac{1}{2}\boldsymbol{I}\right)+j\mathcal{N}\left({\bf
    0},\frac{1}{2}\boldsymbol{I}\right)$.  The other is the model of coded
diffraction patterns described in (\ref{eq:CDP}).
%where 
%
%\[
%	y_{l,k}=\left|\boldsymbol{f}_{k}^{*}\mathrm{diag}\left(\boldsymbol{d}_{l}\right)\boldsymbol{x}\right|^{2},
%	\quad1\leq k\leq n,\text{}1\leq l\leq L.
%\]
%
%Here, the total number of measurements $m=nL$, the vectors $\boldsymbol{f}_{k}^{*}$'s
%stand for the rows of a DFT matrix, and the entries of $\boldsymbol{d}_{l}$
%are independently and uniformly drawn from $\left\{ 1,-1,j,-j\right\} $.
Fig.~\ref{Fig:PhaseTransition2} depicts the average success rate
for both types of measurements over 100 Monte Carlo trials, indicating
that $m>4.5n$ and $m\geq6n$ are often sufficient under complex-valued Gaussian
and CDP models, respectively.

For the sake of comparison, we also report the empirical performance of WF in all the above settings, where the step size is set to be the default choice of \cite{candes2014wirtinger}, that is, $\mu_t = \min \{ 1 - e^{-t/330}, 0.2\}$.  As can be seen, the empirical success rates of TWF outperform WF when $T=1000$ under Gaussian models,  suggesting that TWF either converges faster or exhibits better phase transition behavior.

Another series of experiments has been carried out to demonstrate
the stability of TWF when the number $m$ of quadratic equations varies.
We consider the case where $n=1000$, and vary the SNR (cf. (\ref{eq:SNR}))
from 15 dB to 55dB. The design vectors are real-valued independent
Gaussian $\boldsymbol{a}_{i}\sim\mathcal{N}\left({\bf 0},\boldsymbol{I}\right)$,
while the measurements $y_{i}$ are generated according to the Poisson
noise model (\ref{eq:Poisson}). Fig.~\ref{Fig:MSE} shows the relative
mean square error---in the dB scale---as a function of SNR, when
averaged over 100 independent runs. For all choices of $m$, the numerical
experiments demonstrate that the relative MSE scales inversely proportional
to
%\footnote{By convention, when referring to SNR (or other measurements concerning
%	power quantities), the SNR expressed in dBs is calculated by $\mathrm{SNR}_{\mathrm{dB}}=10\log_{10}\mathrm{SNR}$.
%	In contrast, when referring to relative MSE (or other measurements
%	of amplitude), we use $\mathrm{MSE}_{\mathrm{dB}}=20\log_{10}\mathrm{MSE}$.%
%}
% 
$\text{SNR}$, which matches our stability guarantees in Theorem
\ref{theorem-Truncated-WF-noisy} (since we observe that
  on the dB scale, the slope is about -1 as predicted by the
  theory (\ref{eq:noisy-converge-SNR})).

\begin{figure}[h]
	\centering
	\includegraphics[width=0.48\textwidth]{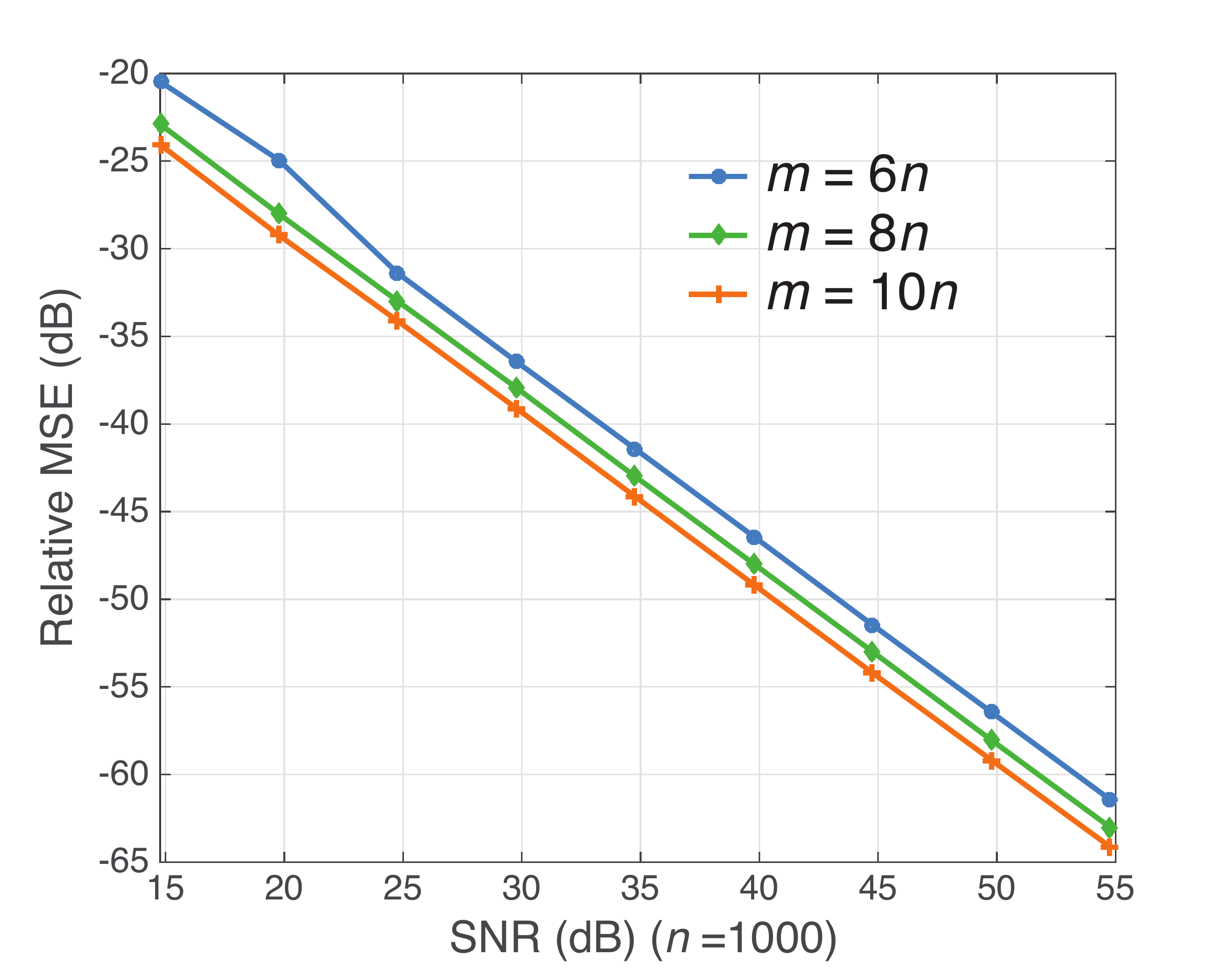}
	\caption{\label{Fig:MSE}Relative MSE vs.~SNR when the
          $y_{i}$'s follow the Poisson model.}
\end{figure}

%\yxc{Currently I remove all plots on WF and stay with only TWF.  But
%  do you think we should keep those WF plots for comparison purpose?}
%\ejc{I think we could show one or two WF plots.}

\section{Exact recovery from noiseless data
\label{sec:Proof-of-Theorem-TruncatedWF}}

This section proves the theoretical guarantees of TWF in the absence
of noise (i.e.~Theorem \ref{theorem-Truncated-WF}). We separate the
noiseless case mainly out of pedagogical reasons, as most of the steps
carry over to the noisy case with slight modification.

% For ease of presentation, we summarize our assumptions on algorithmic
% parameters as follows:

% \begin{itemize}
%	\item \textbf{P1} $\text{ }$The numerical constants $(\alpha_{h},\alpha_{z}^{\mathrm{lb}},\alpha_{z}^{\mathrm{ub}})$
% obey (\ref{eq:Condition-tau_z}).
	%
	%\item \textbf{P2} $\text{ }$The learning parameter is set to be a constant,
%i.e. $\mu_{t}\equiv\mu>0$ ($\forall t\in\mathbb{N}$).
% \end{itemize}

The analysis for the truncated spectral method relies on the celebrated Davis-Kahan $\sin\Theta$ theorem \cite{davis1970rotation},  which we defer  to Appendix \ref{proof-truncated-spectral}. 
%\yxc{The following paragraph is incorrect in terms of the sample complexity -- will need truncation in spectral initialization}
%We start with the spectral method, for which the quality of the initialization
%has been studied in \cite[Section 7.8]{candes2014wirtinger}. Repeating
%exactly the same proof arguments (which we omit in the current paper)
%suggests that: 
%
In short, for any fixed $\delta>0$ and $\boldsymbol{x}\in\mathbb{R}^{n}$,
the initial point $\boldsymbol{z}^{(0)}$ returned by the truncated spectral method
obeys
\[
	\mathrm{dist}(\boldsymbol{z}^{(0)},\boldsymbol{x}) \leq \delta\Vert \boldsymbol{x}\Vert 
\]
with high probability, provided that $m/n$ exceeds some 
numerical constant. With this in place, it suffices to demonstrate
that the TWF update rule is locally contractive, as stated
in the following proposition. 

\begin{proposition}[{\bf Local error contraction}]
\label{prop-Contraction}
Consider the noiseless case (\ref{eq:Noiseless}). Under the condition (\ref{eq:Condition-tau_z}),
there exist some universal constants
$0 < \rho_0 < 1$  and $c_{0}, c_1, c_{2}>0$
such that with probability exceeding $1 - c_{1}\exp\left(-c_{2}m\right)$,
\begin{eqnarray}
	\mathrm{dist}^2\left( \boldsymbol{z}+\frac{\mu}{m}\nabla\ell_{\mathrm{tr}}\left(\boldsymbol{z}\right), \boldsymbol{x} \right) 
	& \leq & \left(1-\rho_0\right)\mathrm{dist}^2 \left( \boldsymbol{z},\boldsymbol{x} \right)
\end{eqnarray}
holds simultaneously for all $\boldsymbol{x},\boldsymbol{z}\in\mathbb{R}^{n}$
obeying
\begin{eqnarray}
	\frac{\mathrm{dist}\left(\boldsymbol{z},\boldsymbol{x}\right)}{\left\Vert \boldsymbol{z}\right\Vert } 
	& \leq & \min\left\{ 	\frac{1}{11}, \text{ }
			     	\frac{\alpha_{z}^{\mathrm{lb}}}{3\alpha_{h}}, \text{ }
				\frac{\alpha_{z}^{\mathrm{lb}}}{6}, \text{ }
				\frac{5.7\left(\alpha_{z}^{\mathrm{lb}}\right)^{2}}{2\alpha_{z}^{\mathrm{ub}}+\alpha_{z}^{\mathrm{lb}}}\right\} ,	\label{eq:Condition-Contraction}
\end{eqnarray}
provided that $m \geq c_{0}n$ and that $\mu$ is some constant obeying 
\[ 
  0< \mu \leq \mu_0 := \frac{0.994-\zeta_{1}-\zeta_{2}-\sqrt{2/(9\pi)}{\alpha_{h}^{-1}}}{2\left(1.02 + 0.665/\alpha_{h} \right)}.
\]
% and
% $\rho_0:= 16(1.02+0.665/\alpha_h) \mu (\mu_0 - \mu)$.
\end{proposition}

Proposition \ref{prop-Contraction} reveals the monotonicity
of the estimation error: once entering a neighborhood around $\boldsymbol{x}$
of a reasonably small size, the iterative updates will remain within
this neighborhood all the time and be attracted towards $\boldsymbol{x}$
at a geometric rate. 

%To give the reader a better sense of the tightness of the derived convergence rate, we remark that when $\alpha_z^{\text{lb}}=0.3$, $\alpha_z^{\text{ub}}=\alpha_h=5$, and $\mu_t \equiv 0.2$ (i.e. the parameters used in Fig. \ref{Fig:CG-TWF}), one has $1-\rho_0 \approx 0.6221$, and thus 
%
%\[
%	\mathrm{dist}(\boldsymbol{z}^{(t+1)}, \boldsymbol{x}) 
%		\text{ } \leq \text{ } \sqrt{1-\rho_0} \text{ } \mathrm{dist}(\boldsymbol{z}^{(t)}, \boldsymbol{x})  
%		\text{ } \approx\text{ }   0.7887 \mathrm{dist}(\boldsymbol{z}^{(t)}, \boldsymbol{x}).
%\]
%
%Comparing it with  Fig. \ref{Fig:CG-TWF}, we see that the empirical convergence rate $\mathrm{dist}(\boldsymbol{z}^{(t+1)}, \boldsymbol{x}) / \mathrm{dist}(\boldsymbol{z}^{(t)}, \boldsymbol{x})$ is around 0.7556. This is quite close to our theoretical upper bound even though $m$ is merely $8n$.

As shown in Section \ref{sec:Why-it-works},
under the hypothesis $\mathsf{RC}\left(\mu,\lambda,\epsilon\right)$
one can conclude
\begin{align}
	\mathrm{dist}^{2}\left(\boldsymbol{z} + \frac{\mu}{m} \nabla\ell_{\text{tr}}(\boldsymbol{z}), \boldsymbol{x}\right) 
	 \leq  (1-\mu\lambda) \mathrm{dist}^{2} (\boldsymbol{z}, \boldsymbol{x}), 
	\quad \forall (\boldsymbol{z}, \boldsymbol{x}) \text{ with }\mathrm{dist}(\boldsymbol{z},\boldsymbol{x}) \leq \epsilon.
\end{align}
Thus, everything now boils down to showing
$\mathsf{RC}\left(\mu,\lambda,\epsilon\right)$ for some constants
$\mu,\lambda,\epsilon>0$. This occupies the rest of this section.

\subsection{Preliminary facts about $\left\{ \mathcal{E}_{1}^{i}\right\} $ and
$\left\{ \mathcal{E}_{2}^{i}\right\} $}
\label{fact:event}

Before proceeding, we gather a few properties of the events
$\mathcal{E}_{1}^{i}$ and $\mathcal{E}_{2}^{i}$, which will prove
crucial in establishing
$\mathsf{RC}\left(\mu,\lambda,\epsilon\right)$. To begin with, recall
that the truncation level given in $\mathcal{E}_{2}^{i}$ depends on
$\frac{1}{m}\left\Vert
  \mathcal{A}\left(\boldsymbol{x}\boldsymbol{x}^{\top}-\boldsymbol{z}\boldsymbol{z}^{\top}\right)\right\Vert
_{1}$.  Instead of working with this random variable directly, we use
deterministic quantities that are more amenable to analysis.
Specifically, we claim that $\frac{1}{m}\left\Vert
  \mathcal{A}\left(\boldsymbol{x}\boldsymbol{x}^{\top}-\boldsymbol{z}\boldsymbol{z}^{\top}\right)\right\Vert
_{1}$ offers a uniform and orderwise tight estimate on $\left\Vert
  \boldsymbol{h}\right\Vert \left\Vert \boldsymbol{z}\right\Vert $,
which can be seen from the following two facts.

\begin{lemma}
%[{\cite[Lemmas 3.1 and 3.2]{candes2012phaselift} }]
\label{lemma:RIP-Candes}
Fix $\zeta\in(0,1)$. If $m>c_{0}n\zeta^{-2}\log\frac{1}{\zeta}$, then
with probability at least $1-C\exp (-c_{1}\zeta^{2}m )$,
\begin{equation}
	0.9\left(1-\zeta\right)\left\Vert \boldsymbol{M}\right\Vert _{\mathrm{F}}  
	\leq \frac{1}{m}\left\Vert \mathcal{A}\left(\boldsymbol{M}\right)\right\Vert _{1}  
	\leq \left(1 + \zeta\right)\sqrt{2}\left\Vert \boldsymbol{M}\right\Vert _{\mathrm{F}}
	\label{eq:RIP-rank2}
\end{equation}
holds for all symmetric rank-2 matrices $\boldsymbol{M}\in\mathbb{R}^{n\times n}$.
Here, $c_{0},c_{1},C>0$ are some universal constants.\end{lemma} 

\begin{proof}
Since \cite[Lemma 3.1]{candes2012phaselift} already
establishes the upper bound, it suffices to prove the lower tail bound.
Consider all symmetric rank-2 matrices $\boldsymbol{M}$ with eigenvalues
$1$ and $-t$ for some $-1\leq t\leq 1$. When $t\in[0,1]$, it has been shown in the proof of \cite[Lemma 3.2]{candes2012phaselift}
that with high probability,
\begin{equation}
	\frac{1}{m}\left\Vert \mathcal{A}\left(\boldsymbol{M}\right)\right\Vert _{1}
	\geq \left(1-\zeta\right)f\left(t\right),
\end{equation}
for all such rank-2 matrices $\boldsymbol{M}$, where 
$f\left(t\right):=\frac{2}{\pi}\left\{ 2\sqrt{t}+\left(1-t\right)\left(\pi/2-2\text{arc}\tan(\sqrt{t})\right)\right\}$.
The lower bound in this case can then be justified by recognizing
that $f\left(t\right)/\sqrt{1+t^{2}}\geq0.9$ for all $t\in[0,1]$,
as illustrated in Fig.~\ref{fig:ft_plot}. The case where 
$t\in[-1,0]$ is an immediate consequence from \cite[Lemma 3.1]{candes2012phaselift}.
%that 
%
%\[
%	\frac{1}{m}\left\Vert \mathcal{A}\left(\boldsymbol{M}\right)\right\Vert _{1}
%	\geq \left(1-\zeta\right)\left\Vert \boldsymbol{M}\right\Vert _{*}
%	\geq \left(1-\zeta\right)\left\Vert \boldsymbol{M}\right\Vert _{\mathrm{F}},
%\]
%
%concluding the proof. 
\end{proof}

\begin{figure}
	\centering
	\includegraphics[scale=0.8]{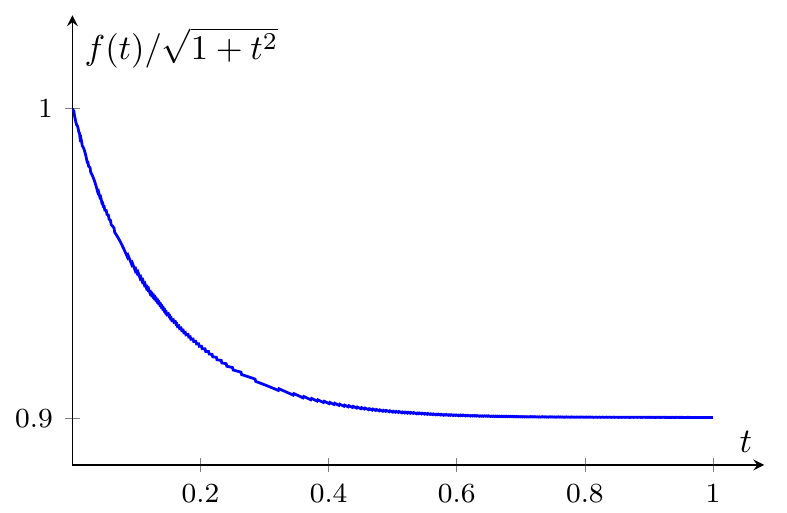}\caption{$\frac{f(t)}{\sqrt{1+t^{2}}}$ as a function of $t$.
	\label{fig:ft_plot}}
\end{figure}

\begin{lemma}\label{lemma: bound-h}
Consider any $\boldsymbol{x},\boldsymbol{z}\in\mathbb{R}^{n}$
obeying $\left\Vert \boldsymbol{z}-\boldsymbol{x}\right\Vert \leq\delta\left\Vert \boldsymbol{z}\right\Vert $
for some $\delta<\frac{1}{2}$. Then one has 
\begin{equation}
	\sqrt{2-4\delta}\left\Vert \boldsymbol{z}-\boldsymbol{x}\right\Vert \left\Vert \boldsymbol{z}\right\Vert
	\leq  \left\Vert \boldsymbol{x}\boldsymbol{x}^{\top} - \boldsymbol{z}\boldsymbol{z}^{\top}\right\Vert _{\mathrm{F}} 
	\leq  \left(2+\delta\right)\left\Vert \boldsymbol{z} - \boldsymbol{x}\right\Vert \left\Vert \boldsymbol{z}\right\Vert .
	\label{eq:bound-h}
\end{equation}
\end{lemma} 

\begin{proof}Take $\boldsymbol{h}=\boldsymbol{z}-\boldsymbol{x}$
and write 
%rewrite the term of interest as
%
\begin{eqnarray*}
	\left\Vert \boldsymbol{x}\boldsymbol{x}^{\top}-\boldsymbol{z}\boldsymbol{z}^{\top}\right\Vert _{\mathrm{F}}^{2} & = & \big\|-\boldsymbol{h}\boldsymbol{z}^{\top}-\boldsymbol{z}\boldsymbol{h}^{\top}+\boldsymbol{h}\boldsymbol{h}^{\top}\big\|_{\mathrm{F}}^{2}\\
 	& = & \big\|\boldsymbol{h}\boldsymbol{z}^{\top}+\boldsymbol{z}\boldsymbol{h}^{\top}\big\|_{\mathrm{F}}^{2}+\left\Vert \boldsymbol{h}\right\Vert ^{4}
	  - 2 \langle \boldsymbol{h}\boldsymbol{z}^{\top} + \boldsymbol{z}\boldsymbol{h}^{\top}, \boldsymbol{h}\boldsymbol{h}^{\top}\rangle \\
 	& = & 2\left\Vert \boldsymbol{z}\right\Vert ^{2}\left\Vert \boldsymbol{h}\right\Vert ^{2}+2|\boldsymbol{h}^{\top}\boldsymbol{z}|^{2} + \left\Vert \boldsymbol{h}\right\Vert ^{4} 
	  - 2 \Vert \boldsymbol{h}\Vert ^2 (\boldsymbol{h}^{\top}\boldsymbol{z} + \boldsymbol{z}^{\top}\boldsymbol{h} ).
\end{eqnarray*}
When $\Vert \boldsymbol{h} \Vert <\frac{1}{2} \Vert \boldsymbol{z}\Vert $,
the Cauchy-Schwarz inequality gives
\begin{equation}
	2\left\Vert \boldsymbol{z}\right\Vert ^{2}\left\Vert \boldsymbol{h}\right\Vert ^{2}-4\left\Vert \boldsymbol{z}\right\Vert \left\Vert \boldsymbol{h}\right\Vert ^{3}  
	\leq  \left\Vert \boldsymbol{x}\boldsymbol{x}^{\top}-\boldsymbol{z}\boldsymbol{z}^{\top} \right\Vert _{\mathrm{F}}^{2}  
	\leq  4 \left\Vert \boldsymbol{z}\right\Vert ^{2}\left\Vert \boldsymbol{h}\right\Vert ^{2} + 4\left\Vert \boldsymbol{h}\right\Vert ^{3}\left\Vert \boldsymbol{z} \right\Vert + \left\Vert \boldsymbol{h} \right\Vert ^{4},
\end{equation}
\begin{equation}
	\Rightarrow\quad
	\sqrt{\left(2\left\Vert \boldsymbol{z}\right\Vert - 4\left\Vert \boldsymbol{h}\right\Vert \right)\left\Vert \boldsymbol{z}\right\Vert }\cdot\left\Vert \boldsymbol{h}\right\Vert  
	\leq \left\Vert \boldsymbol{x}\boldsymbol{x}^{\top} - \boldsymbol{z}\boldsymbol{z}^{\top}\right\Vert _{\mathrm{F}} 
	\leq \left(2\left\Vert \boldsymbol{z}\right\Vert +\left\Vert \boldsymbol{h}\right\Vert \right)\cdot\left\Vert \boldsymbol{h}\right\Vert 
\end{equation}
as claimed.
\end{proof}

Taken together the above two facts demonstrate that with probability
$1-\exp\left(-\Omega\left(m\right)\right)$, 
\begin{equation}
	1.15\left\Vert \boldsymbol{z}-\boldsymbol{x}\right\Vert \left\Vert \boldsymbol{z}\right\Vert 
	\leq \frac{1}{m}\left\Vert \mathcal{A}\left(\boldsymbol{x}\boldsymbol{x}^{\top} - \boldsymbol{z}\boldsymbol{z}^{\top}\right)\right\Vert _{1}
	\leq 3\left\Vert \boldsymbol{z}-\boldsymbol{x}\right\Vert \left\Vert \boldsymbol{z}\right\Vert 
	\label{eq:A_1_bound}
\end{equation}
holds simultaneously for all $\boldsymbol{z}$ and $\boldsymbol{x}$
satisfying $\left\Vert \boldsymbol{h}\right\Vert \leq\frac{1}{11}\left\Vert \boldsymbol{z}\right\Vert $.
Conditional on (\ref{eq:A_1_bound}), the inclusion 
\begin{eqnarray}
	\mathcal{E}_{3}^{i} &  \subseteq  & \mathcal{E}_{2}^{i}\text{ } \subseteq \text{ }\mathcal{E}_{4}^{i} 
	\label{eq:Event2}
\end{eqnarray}
holds with respect to the following events
\begin{eqnarray}
	\mathcal{E}_{3}^{i}: & = & \left\{ \left||\boldsymbol{a}_{i}^{\top}\boldsymbol{x}|^{2}-|\boldsymbol{a}_{i}^{\top}\boldsymbol{z}|^{2}\right|
		\leq 1.15\alpha_{h}\left\Vert \boldsymbol{h}\right\Vert \cdot\left|\boldsymbol{a}_{i}^{\top}\boldsymbol{z}\right|\right\} ,\\
	\mathcal{E}_{4}^{i}: & = & \left\{ \left||\boldsymbol{a}_{i}^{\top}\boldsymbol{x}|^{2}-|\boldsymbol{a}_{i}^{\top}\boldsymbol{z}|^{2}\right|
		\leq 3 \alpha_{h}\left\Vert \boldsymbol{h}\right\Vert \cdot\left|\boldsymbol{a}_{i}^{\top}\boldsymbol{z}\right|\right\} .
\end{eqnarray}
The point of introducing these new events is that the
$\mathcal{E}_{3}^{i}$'s (resp.~$\mathcal{E}_{4}^{i}$'s) are
statistically independent for any fixed $\boldsymbol{x}$ and
$\boldsymbol{z}$ and are, therefore, easier to work with. % compared to
% $\mathcal{E}_{2}^{i}$'s.

Note that each $\mathcal{E}_{3}^{i}$ (resp.~$\mathcal{E}_{4}^{i}$) is
specified by a quadratic inequality. A closer inspection reveals that
in order to satisfy these quadratic inequalities, the quantity
$\boldsymbol{a}_{i}^{\top}\boldsymbol{h}$ must fall within two
intervals centered around $0$ and
$2\boldsymbol{a}_{i}^{\top}\boldsymbol{z}$, respectively. One can thus
facilitate analysis by decoupling each quadratic inequality of
interest into two simple linear inequalities, as stated in the following lemma.

\begin{lemma}\label{Lemma:superset-D}For any $\gamma>0$, define
\begin{eqnarray}
	\mathcal{D}_{\gamma}^{i} & := & \left\{ \left||\boldsymbol{a}_{i}^{\top}\boldsymbol{x}|^{2}-|\boldsymbol{a}_{i}^{\top}\boldsymbol{z}|^{2}\right|\leq\gamma\left\Vert \boldsymbol{h}\right\Vert \left|\boldsymbol{a}_{i}^{\top}\boldsymbol{z}\right|\right\} ,\\
	\mathcal{D}_{\gamma}^{i,1} & := & \left\{ \frac{|\boldsymbol{a}_{i}^{\top}\boldsymbol{h}|}{\|\boldsymbol{h}\|}\leq\gamma\right\} ,\label{eq:defn-D_gamma-1}\\
	\text{and} \quad \mathcal{D}_{\gamma}^{i,2} & := & \left\{ \left|\frac{\boldsymbol{a}_{i}^{\top}\boldsymbol{h}}{\left\Vert \boldsymbol{h}\right\Vert }-\frac{2\boldsymbol{a}_{i}^{\top}\boldsymbol{z}}{\left\Vert \boldsymbol{h}\right\Vert }\right|\leq\gamma\right\} .
	\label{eq:defn-D_gamma-2}
\end{eqnarray}
Thus, $\mathcal{D}_{\gamma}^{i,1}$ and $\mathcal{D}_{\gamma}^{i,2}$
represent the two intervals on $\boldsymbol{a}_i^{\top}\boldsymbol{h}$
centered around $0$ and $2\boldsymbol{a}_i^{\top}\boldsymbol{z}$. If
$\frac{\left\Vert \boldsymbol{h}\right\Vert }{\left\Vert
    \boldsymbol{z}\right\Vert
}\leq\frac{\alpha_z^{\mathrm{lb}}}{\gamma}$, then the following
inclusion holds
\begin{align}
	\left(\mathcal{D}_{\frac{\gamma}{1+\sqrt{2}}}^{i,1}\cap\mathcal{E}_{1}^{i}\right)\cup\left(\mathcal{D}_{\frac{\gamma}{1+\sqrt{2}}}^{i,2}\cap\mathcal{E}_{1}^{i}\right)
	\subseteq \text{ }\text{ }\mathcal{D}_{\gamma}^i \cap \mathcal{E}_{1}^{i}\text{ } 
	~ \subseteq ~ \left(\mathcal{D}_{\gamma}^{i,1}\cap\mathcal{E}_{1}^{i}\right)\cup\left(\mathcal{D}_{\gamma}^{i,2}\cap\mathcal{E}_{1}^{i}\right).
	\label{eq:superset-D}
\end{align}
\end{lemma}
%\begin{proof}See Appendix \ref{proof-of-Lemma:superset-D}.\end{proof}

%{\color{red} On the event
%  $\mathcal{D}_{\gamma}^{i,1}\cap\mathcal{E}_{1}^{i}$, all three
%  linear quanitites $ $$\boldsymbol{a}_{i}^{\top}\boldsymbol{h}$,
%  $\boldsymbol{a}_{i}^{\top}\boldsymbol{z}$ and
%  $\boldsymbol{a}_{i}^{\top}\boldsymbol{x}$ can be absolutely
%  controlled, which is the ideal situation for TWF to evolve in a
%  regular manner. On the other hand, $\mathcal{D}_{\gamma}^{i,2}$ does
%  not share such a desired feature, simply because
%  $2\boldsymbol{a}_{i}^{\top}\boldsymbol{z}/\left\Vert
%    \boldsymbol{h}\right\Vert $ becomes increasingly larger when
%  $\left\Vert \boldsymbol{h}\right\Vert /\left\Vert
%    \boldsymbol{z}\right\Vert $ decreases. Fortunately, the number of
%  measurements satisfying the hypothesis of
%  $\mathcal{D}_{\gamma}^{i,2}$ is vanishingly small, thus precluding a
%  noticable effect of $\mathcal{D}_{\gamma}^{i,2}$ upon the descent
%  direction.} \ejc{We sort of see where you're getting at but this is
%  not very clear. We need to discuss.}

\subsection{Proof of the regularity condition \label{sub:Expectation-and-Concentration}}

By definition, one step towards proving the regularity condition (\ref{eq:Regularity-Truncate-Original})
is to control the norm of the regularized gradient. In fact, a crude argument already reveals that 
$\| \frac{1}{m} \nabla\ell_{\mathrm{tr}}(\boldsymbol{z}) \| \lesssim \| \boldsymbol{h} \|$. 
To see this, introduce 
$\boldsymbol{v}=\left[v_{i}\right]_{1\leq i\leq m}$ with
$v_{i} := 2\frac{\left|\boldsymbol{a}_{i}^{\top}\boldsymbol{x}\right|^{2}-\left|\boldsymbol{a}_{i}^{\top}\boldsymbol{z}\right|^{2}}{\boldsymbol{a}_{i}^{\top}\boldsymbol{z}}
	{\bf 1}_{\mathcal{E}_{1}^i \cap \mathcal{E}_2^i }.$
It comes from the trimming rule  $\mathcal{E}_1^i$ as well as the inclusion property (\ref{eq:Event2}) that
\[
	\left|\boldsymbol{a}_{i}^{\top}\boldsymbol{z}\right|\gtrsim\left\Vert \boldsymbol{z}\right\Vert 
	\quad \text{and} \quad
	\left|y_{i}-\left|\boldsymbol{a}_{i}^{\top}\boldsymbol{z}\right|^{2}\right|
	\lesssim \frac{1}{m}\|\mathcal{A}(\boldsymbol{x}\boldsymbol{x}^{\top}-\boldsymbol{z}\boldsymbol{z}^{\top})\|_1
	\asymp \left\Vert \boldsymbol{h}\right\Vert \left\Vert \boldsymbol{z}\right\Vert ,
\]
implying $|v_i|\lesssim \| \boldsymbol{h} \|$ and hence $\Vert
\boldsymbol{v} \Vert \lesssim \sqrt{m} \Vert \boldsymbol{h} \Vert $.
The Marchenko\textendash{}Pastur law gives $\Vert \boldsymbol{A}\Vert
\lesssim \sqrt{m}$, whence
\begin{equation}
	\frac{1}{m} \Vert \nabla\ell_{\mathrm{tr}} (\boldsymbol{z}) \Vert 
	= \frac{1}{m} \Vert \boldsymbol{A}^{\top} \boldsymbol{v} \Vert
	\leq \frac{1}{m} \Vert \boldsymbol{A} \Vert \cdot \Vert \boldsymbol{v} \Vert 
	\lesssim  \Vert \boldsymbol{h}\Vert .
	\label{eq:crude-norm}
\end{equation}
A more refined estimate will be provided in Lemma \ref{Lemma:norm-score}. 

The above argument essentially tells us that to establish $\mathsf{RC}$,
it suffices to verify a uniform lower bound of the form
\begin{equation}
	-\Big\langle\boldsymbol{h},\text{ }\frac{1}{m}\nabla\ell_{\mathrm{tr}}\left(\boldsymbol{z}\right)\Big\rangle\text{ }
	\gtrsim\text{ }\left\Vert \boldsymbol{h}\right\Vert ^{2},
	\label{eq:Regularity-Simple}
\end{equation}
as formally derived in the following proposition.

\begin{proposition}
\label{prop-regularity-noiseless}
  Consider the noise-free measurements $y_{i}=|\boldsymbol{a}_{i}^{\top}\boldsymbol{x}|^{2}$
  and any fixed constant $\epsilon>0$. Under the condition (\ref{eq:Condition-tau_z}),
  if $m>c_{1}n$, then with probability exceeding $1-C\exp\left(-c_{0}m\right)$,
  \begin{equation}
	-\Big\langle\boldsymbol{h},\frac{1}{m}\nabla\ell_{\mathrm{tr}}\left(\boldsymbol{z}\right)\Big\rangle\hspace{0.3em}
	\geq \hspace{0.3em} 2\left\{ 1.99-2\left(\zeta_{1}+\zeta_{2}\right)-\sqrt{8/(9\pi)}\alpha_{h}^{-1} - \epsilon\right\} \left\Vert \boldsymbol{h}\right\Vert ^{2}
  \end{equation}
  holds uniformly over all $\boldsymbol{x}$, $\boldsymbol{z}\in \mathbb{R}^n$ obeying
  \begin{equation}
	\frac{\left\Vert \boldsymbol{h}\right\Vert }{\left\Vert \boldsymbol{z}\right\Vert }\hspace{0.3em}\leq\hspace{0.3em}
	\min\left\{ \frac{1}{11},\text{ }   \frac{\alpha_{z}^{\mathrm{lb}}}{3\alpha_{h}},\text{ }
		    \frac{\alpha_{z}^{\mathrm{lb}}}{6},\text{ }   
		    \frac{5.7\left(\alpha_{z}^{\mathrm{lb}}\right)^{2}}{2\alpha_{z}^{\mathrm{ub}}+\alpha_{z}^{\mathrm{lb}}}\right\} .
  \label{eq:Condition}
  \end{equation}
  Here, $c_{0},c_{1},C>0$ are some universal constants, and $\zeta_1$ and $\zeta_2$ are defined in (\ref{eq:Condition-tau_z}).
\end{proposition}

The basic starting point is the observation that
$(\boldsymbol{a}_{i}^{\top}\boldsymbol{z})-(\boldsymbol{a}_{i}^{\top}\boldsymbol{x})^{2}
= (\boldsymbol{a}_{i}^{\top}\boldsymbol{h}) ( 2\boldsymbol{a}_{i}^{\top}\boldsymbol{z}-\boldsymbol{a}_{i}^{\top}\boldsymbol{h} )$ and hence
\begin{eqnarray}
	-\frac{1}{2m}\nabla\ell_{\mathrm{tr}}\left(\boldsymbol{z}\right) 
	& = & \frac{1}{m} \sum_{i=1}^{m} \frac{ ( \boldsymbol{a}_{i}^{\top}\boldsymbol{z} )^{2} - (\boldsymbol{a}_{i}^{\top}\boldsymbol{x})^{2}}{\boldsymbol{a}_{i}^{\top}\boldsymbol{z}}\boldsymbol{a}_{i}{\bf 1}_{\mathcal{E}_{1}^{i}\cap\mathcal{E}_{2}^{i}}  \nonumber\\
 	%& = & \frac{1}{m} \sum_{i=1}^{m} \frac{\left(\boldsymbol{a}_{i}^{\top}\boldsymbol{h}\right)\left(\boldsymbol{a}_{i}^{\top}\left(2\boldsymbol{z}-\boldsymbol{h}\right)\right)}{\boldsymbol{a}_{i}^{\top}\boldsymbol{z}}\boldsymbol{a}_{i}{\bf 1}_{\mathcal{E}_{1}^{i}\cap\mathcal{E}_{2}^{i}}\nonumber \\
 	& = & \frac{1}{m} \sum_{i=1}^{m} 2 (\boldsymbol{a}_{i}^{\top}\boldsymbol{h}) \boldsymbol{a}_{i}{\bf 1}_{\mathcal{E}_{1}^{i}\cap\mathcal{E}_{2}^{i}} 
		- \frac{1}{m} \sum_{i=1}^{m} \frac{(\boldsymbol{a}_{i}^{\top}\boldsymbol{h})^{2}}{\boldsymbol{a}_{i}^{\top}\boldsymbol{z}}\boldsymbol{a}_{i}{\bf 1}_{\mathcal{E}_{1}^{i}\cap\mathcal{E}_{2}^{i}}.
	\label{eq:simplify}
\end{eqnarray}
One would expect the contribution of the second term of (\ref{eq:simplify})
(which is a second-order quantity) to be small as $\left\Vert \boldsymbol{h}\right\Vert /\left\Vert \boldsymbol{z}\right\Vert $
decreases. 

To facilitate analysis, we rewrite (\ref{eq:simplify}) in terms of the
more convenient events $\mathcal{D}_{\gamma}^{i,1}$ and
$\mathcal{D}_{\gamma}^{i,2}$.  Specifically, the inclusion property
(\ref{eq:Event2}) together with Lemma \ref{Lemma:superset-D} reveals
that  
\begin{equation}
	\mathcal{D}_{\gamma_{3}}^{i,1}~ \cap \mathcal{E}_1^i~
	\subseteq \text{ }\mathcal{E}_{3}^{i}~ \cap \mathcal{E}_1^i \subseteq\text{ }\mathcal{E}_{2}^{i} ~ \cap \mathcal{E}_1^i~
	\subseteq \text{ }\mathcal{E}_{4}^{i}~ \cap \mathcal{E}_1^i~
	\subseteq\text{ } \left(\mathcal{D}_{\gamma_{4}}^{i,1}\cup\mathcal{D}_{\gamma_{4}}^{i,2}\right) ~ \cap \mathcal{E}_1^i,
	\label{eq:Event-Containment}
\end{equation}
where the parameters $\gamma_{3},\gamma_{4}$ are given by
\begin{equation}
	\gamma_{3}:=0.476\alpha_{h},\quad\text{and}\quad\gamma_{4}:=3\alpha_{h}.\label{eq:gamma-34}
\end{equation}
This taken collectively with the identity (\ref{eq:simplify}) leads to a lower
estimate
\begin{align}
  & -\Big\langle \frac{1}{2m}\nabla\ell_{\mathrm{tr}}
  (\boldsymbol{z}),\boldsymbol{h} \Big\rangle \geq \nonumber\\
  & \qquad \frac{2}{m}\sum_{i=1}^{m}\left(\boldsymbol{a}_{i}^{\top}\boldsymbol{h}\right)^{2}
  {\bf 1}_{\mathcal{E}_{1}^{i}\cap\mathcal{D}_{\gamma_3}^{i,1}} -
  \frac{1}{m}\sum_{i=1}^{m}\frac{\left|\boldsymbol{a}_{i}^{\top}\boldsymbol{h}\right|^3}{\left|\boldsymbol{a}_{i}^{\top}\boldsymbol{z}\right|}
  {\bf 1}_{\mathcal{E}_{1}^{i}\cap\mathcal{D}_{\gamma_4}^{i,1}} 
-  \frac{1}{m}\sum_{i=1}^{m}\frac{\left|\boldsymbol{a}_{i}^{\top}\boldsymbol{h}\right|^{3}}{\left|\boldsymbol{a}_{i}^{\top}\boldsymbol{z}\right|}
  {\bf 1}_{\mathcal{E}_{1}^{i}\cap\mathcal{D}_{\gamma_4}^{i,2}},
  \label{eq:I1_I2_I3}
\end{align}
leaving us with three quantities in the right-hand side to deal with.
We pause here to explain and compare the influences of these three
terms. 

To begin with, as long as the trimming step does not discard too
many data, the first term should be close to $\frac{2}{m}\sum_i
|\boldsymbol{a}_i^{\top} \boldsymbol{h}|^2$, which approximately gives
$2\|\boldsymbol{h}\|^2$ from the law of large numbers. This term turns
out to be dominant in the right-hand side of (\ref{eq:I1_I2_I3}) as
long as $\|\boldsymbol{h}\| / \|\boldsymbol{z}\|$ is reasonably
small. To see this, please recognize that the second term in the
right-hand side is $O(\|\boldsymbol{h}\|^3 / \|
\boldsymbol{z} \|)$, simply because both
$\boldsymbol{a}_{i}^{\top}\boldsymbol{h}$ and
$\boldsymbol{a}_{i}^{\top}\boldsymbol{z}$ are absolutely controlled on
$\mathcal{D}_{\gamma_4}^{i,1}\cap\mathcal{E}_{1}^{i}$.
However, $\mathcal{D}_{\gamma_4}^{i,2}$ does not share such a desired
feature.  By the very definition of $\mathcal{D}_{\gamma_4}^{i,2}$,
each nonzero summand of the last term of (\ref{eq:I1_I2_I3}) must obey
$\left|\boldsymbol{a}_{i}^{\top}\boldsymbol{h}\right|\approx2\left|\boldsymbol{a}_{i}^{\top}\boldsymbol{z}\right|$
and, therefore,
$\frac{\left|\boldsymbol{a}_{i}^{\top}\boldsymbol{h}\right|^{3}}{\left|\boldsymbol{a}_{i}^{\top}\boldsymbol{z}\right|}{\bf
  1}_{\mathcal{E}_{1}^{i}\cap\mathcal{D}_{\gamma_{4}}^{i,2}}$ is
roughly of the order of $\Vert \boldsymbol{z} \Vert ^2$; this could be
much larger than our target level $\left\Vert
  \boldsymbol{h}\right\Vert ^{2}$.  Fortunately,
$\mathcal{D}_{\gamma_4}^{i,2}$ is a rare event, thus precluding a
noticable influence upon the descent direction. All of this is made
rigorous in Lemma \ref{Lemma:I1_I2} (first term), Lemma \ref{lemma:I2}
(second term) and Lemma \ref{Lemma:I3} (third term) together with
subsequent analysis.

%We will see that the first term in the right-hand side of
%(\ref{eq:I1_I2_I3}) dominates the other two.  Intuitively, if the
%truncation step does not discard many samples, then this term should
%be close to the untruncated mean
%$2\mathbb{E}[(\boldsymbol{a}_{i}^{\top}\boldsymbol{h})^{2}]=2\left\Vert
%  \boldsymbol{h}\right\Vert ^{2}$.  

% is provided in the
% following lemma.

\begin{lemma}\label{Lemma:I1_I2}
Fix $\gamma > 0$, and let $\mathcal{E}_{1}^{i}$
and $\mathcal{D}_{\gamma}^{i,1}$ be defined in (\ref{eq:defn-E1})
and (\ref{eq:defn-D_gamma-1}), respectively. Set
\begin{align}
	\zeta_{1} &:= 1 - \min \left\{ \mathbb{E}\left[\xi^{2}{\bf 1}_{\left\{ \sqrt{1.01}\alpha_{z}^{\mathrm{lb}}\leq\left|\xi\right|\leq\sqrt{0.99}\alpha_{z}^{\mathrm{ub}}\right\} }\right], 
		\mathbb{E}\left[{\bf 1}_{\left\{ \sqrt{1.01}\alpha_{z}^{\mathrm{lb}}\leq\left|\xi\right|\leq\sqrt{0.99}\alpha_{z}^{\mathrm{ub}}\right\} }\right]
	\right\}  \label{eq:Assumption-zeta1}\\
	\quad\text{and}\quad
	\zeta_{2} &:= \mathbb{E}\left[\xi^{2}{\bf 1}_{\left\{ \left|\xi\right|>\sqrt{0.99}\gamma\right\} }\right],
	\label{eq:Assumption-zeta2}
\end{align}
where $\xi\sim\mathcal{N}\left(0,1\right)$. For any $\epsilon>0$, if
$m>c_{1}n\epsilon^{-2}\log \epsilon^{-1}$, then with probability at
least $1-C\exp (-c_{0}\epsilon^{2}m )$,
\begin{eqnarray}
	\frac{1}{m}\sum_{i=1}^{m}\left|\boldsymbol{a}_{i}^{\top}\boldsymbol{h}\right|^{2}
		{\bf 1}_{\mathcal{E}_{1}^{i}\cap\mathcal{D}_{\gamma}^{i,1}} 
	& \geq & \left(1-\zeta_{1}-\zeta_{2}-\epsilon\right)\left\Vert \boldsymbol{h}\right\Vert ^{2}
	\label{eq:I1}
\end{eqnarray}
holds for all non-zero vectors $\boldsymbol{h}, \boldsymbol{z}\in \mathbb{R}^n$.
Here, $c_{0},c_{1},C>0$ are some universal constants.
\end{lemma}

%\begin{proof}See Appendix \ref{sub:Proof-of-Lemma-I1_I2}.\end{proof}

We now move on to the second term in the right-hand side of (\ref{eq:I1_I2_I3}). For any fixed $\gamma>0$,
the definition of $\mathcal{E}_{1}^{i}$ gives rise to an upper estimate
\begin{align}
  \frac{1}{m}\sum_{i=1}^{m}\frac{\left|\boldsymbol{a}_{i}^{\top}\boldsymbol{h}\right|^{3}}{\left|\boldsymbol{a}_{i}^{\top}\boldsymbol{z}\right|}
  {\bf 1}_{\mathcal{E}_{1}^{i}\cap\mathcal{D}_{\gamma}^{i,1}} 
  \leq  \frac{1}{\alpha_{z}^{\mathrm{lb}}\left\Vert \boldsymbol{z}\right\Vert }\frac{1}{m}\sum_{i=1}^{m}\left|\boldsymbol{a}_{i}^{\top}\boldsymbol{h}\right|^{3}
  {\bf 1}_{\mathcal{D}_{\gamma}^{i,1}}\text{ }
  {\leq}\text{ }\frac{\left(1+\epsilon\right)\sqrt{8/\pi}\left\Vert \boldsymbol{h}\right\Vert ^{3}}{\alpha_{z}^{\mathrm{lb}}\left\Vert \boldsymbol{z}\right\Vert },
	\label{eq:I2}
\end{align}
where $\sqrt{8/\pi}\left\Vert \boldsymbol{h}\right\Vert ^{3}$ is
exactly the untruncated moment
$\mathbb{E}[|\boldsymbol{a}_{i}^{\top}\boldsymbol{h}|^{3}]$.  The
second inequality is a consequence of the lemma below, which arises by
observing that the summands
$|\boldsymbol{a}_{i}^{\top}\boldsymbol{h}|^3 {\bf
  1}_{\mathcal{D}_{\gamma}^{i,1}}$ are independent sub-Gaussian random
variables.

\begin{lemma}
\label{lemma:I2}
For any constant $\gamma>0$, if $m/n \geq c_{0} \cdot \epsilon^{-2}
\log \epsilon^{-1}$, then
\begin{equation}
	\frac{1}{m}\sum_{i=1}^{m}\left|\boldsymbol{a}_{i}^{\top}\boldsymbol{h}\right|^{3}{\bf 1}_{\mathcal{D}_{\gamma}^{i,1}}
	\leq \left(1+\epsilon\right)\sqrt{8/\pi}\left\Vert \boldsymbol{h}\right\Vert ^{3},\quad\forall\boldsymbol{h}\in\mathbb{R}^{n}
\end{equation}
with probability at least $1-C\exp (-c_{1}\epsilon^{2}m )$
for some universal constants $c_{0},c_{1},C>0$.
\end{lemma}

%\begin{proof}See Appendix \ref{sub:Proof-of-Lemma-I2}.\end{proof}

It remains to control the last term of (\ref{eq:I1_I2_I3}). As mentioned above,  the influence of this term is small since the set of $\boldsymbol{a}_{i}$'s satisfying
$\mathcal{D}_{\gamma}^{i,2}$ accounts for a small fraction of
measurements. Put formally, the number of equations satisfying
$\left|\boldsymbol{a}_{i}^{\top}\boldsymbol{h}\right|\geq\gamma\left\Vert
  \boldsymbol{h}\right\Vert $ decays rapidly for large $\gamma$ (at
least at a quadratic rate), as stated below.

\begin{lemma}\label{Lemma:I3}
For any $0 < \epsilon < 1$,
there exist some universal constants $c_{0},c_{1},C>0$ such that
\begin{equation}
	\frac{1}{m}\sum_{i=1}^{m}{\bf 1}_{\left\{ \left|\boldsymbol{a}_{i}^{\top}\boldsymbol{h}\right|
	\geq \text{ \ensuremath{\gamma}}\Vert \boldsymbol{h} \Vert \right\} }\leq\frac{1}{0.49\gamma}\exp\left(-0.485\gamma^{2}\right)
		+ \frac{\epsilon}{\gamma^{2}},
	\quad\forall\boldsymbol{h}\in\mathbb{R}^{n} \backslash \{{\bf 0}\}
	\text{ and }\gamma \geq 2
	\label{eq:I3-indicator}
\end{equation}
with probability at least
$1-C\exp\left(-c_{0}\epsilon^{2}m\right)$. This holds with the proviso
$m/n \geq c_{1} \cdot \epsilon^{-2} \log \epsilon^{-1}$.
\end{lemma}

%\begin{proof}See Appendix \ref{sub:Proof-of-Lemma-I3}.\end{proof}

To connect this lemma with the last term of (\ref{eq:I1_I2_I3}),
we recognize that when $\gamma\leq\frac{\alpha_{z}^{\mathrm{lb}} \Vert \boldsymbol{z} \Vert }{\Vert \boldsymbol{h} \Vert }$,
one has
\begin{eqnarray}
	%\sum\nolimits_{i=1}^{m}
	{\bf 1}_{\mathcal{E}_{1}^{i}\cap\mathcal{D}_{\gamma}^{i,2}} 
	& \leq & 
	%\sum\nolimits_{i=1}^{m}
	{\bf 1}_{\left\{ \left|\boldsymbol{a}_{i}^{\top}\boldsymbol{h}\right|\geq\alpha_{z}^{\mathrm{lb}} \Vert \boldsymbol{z} \Vert \right\} }.
	\label{eq:I3-card}
\end{eqnarray}
The constraint
$\left|\frac{\boldsymbol{a}_{i}^{\top}\boldsymbol{h}}{\left\Vert
      \boldsymbol{h}\right\Vert
  }-\frac{2\boldsymbol{a}_{i}^{\top}\boldsymbol{z}}{\left\Vert
      \boldsymbol{h}\right\Vert }\right|\leq\gamma$ of
$\mathcal{D}_{\gamma}^{i,2}$ necessarily requires
\begin{equation}
	\frac{\left|\boldsymbol{a}_{i}^{\top}\boldsymbol{h}\right|}{\Vert \boldsymbol{h} \Vert }
	\geq \frac{2\left|\boldsymbol{a}_{i}^{\top}\boldsymbol{z}\right|} { \Vert \boldsymbol{h} \Vert }-\gamma
	\geq \frac{2\alpha_{z}^{\mathrm{lb}} \Vert \boldsymbol{z} \Vert } { \Vert \boldsymbol{h} \Vert }-\gamma
	\geq \frac{\alpha_{z}^{\mathrm{lb}} \Vert \boldsymbol{z} \Vert } { \Vert \boldsymbol{h} \Vert },
\end{equation}
where the last inequality comes from our assumption on $\gamma$.
With Lemma \ref{Lemma:I3} in place, (\ref{eq:I3-card}) immediately gives
%suggests
%
\begin{eqnarray}
	\sum_{i=1}^{m}{\bf 1}_{\mathcal{E}_{1}^{i}\cap\mathcal{D}_{\gamma}^{i,2}} 
	& \leq & \frac{\left\Vert \boldsymbol{h}\right\Vert }{0.49\alpha_{z}^{\mathrm{lb}}\left\Vert \boldsymbol{z}\right\Vert }\exp\left(-0.485\left(\frac{\alpha_{z}^{\mathrm{lb}}\left\Vert \boldsymbol{z}\right\Vert }{\left\Vert \boldsymbol{h}\right\Vert }\right)^{2}\right)
		+ \frac{\epsilon\left\Vert \boldsymbol{h}\right\Vert ^{2}}{\left(\alpha_{z}^{\mathrm{lb}}\right)^{2}\left\Vert \boldsymbol{z}\right\Vert ^{2}}\nonumber \\
 	& \leq & \frac{1}{9800}\left(\frac{\left\Vert \boldsymbol{h}\right\Vert }{\alpha_{z}^{\mathrm{lb}}\left\Vert \boldsymbol{z}\right\Vert }\right)^{4}
		+ \frac{\epsilon}{\left(\alpha_{z}^{\mathrm{lb}}\right)^{2}}\left(\frac{\left\Vert \boldsymbol{h}\right\Vert }{\left\Vert \boldsymbol{z}\right\Vert }\right)^{2}
	\label{eq:card-I3}
\end{eqnarray}
as long as $\frac{\left\Vert \boldsymbol{h}\right\Vert }{\left\Vert
    \boldsymbol{z}\right\Vert
}\leq\frac{\alpha_{z}^{\mathrm{lb}}}{6}$, where the last inequality
uses the majorization
$\frac{1}{20000x^{4}}\geq\frac{1}{x}\exp\left(-0.485x^{2}\right)$
holding for any $x\geq6$.

In addition, on $\mathcal{E}_{1}^{i}\cap\mathcal{D}_{\gamma}^{i,2}$,
the amplitude of each summand can be bounded in such a way that
\begin{eqnarray}
	\frac{\left|\boldsymbol{a}_{i}^{\top}\boldsymbol{h}\right|^{3}}{\left|\boldsymbol{a}_{i}^{\top}\boldsymbol{z}\right|} 
	& \leq & \frac{\left|2\boldsymbol{a}_{i}^{\top}\boldsymbol{z}\right|+\gamma\left\Vert \boldsymbol{h}\right\Vert }{\left|\boldsymbol{a}_{i}^{\top}\boldsymbol{z}\right|}\left(2\alpha_{z}^{\mathrm{ub}}\left\Vert \boldsymbol{z}\right\Vert +\gamma\left\Vert \boldsymbol{h}\right\Vert \right)^{2}
	\label{eq:ub1}\\
 	& \leq & \left(2+\frac{\gamma}{\alpha_{z}^{\mathrm{lb}}}\frac{\left\Vert \boldsymbol{h}\right\Vert }{\left\Vert \boldsymbol{z}\right\Vert }\right)\left(2\alpha_{z}^{\mathrm{ub}} + \gamma\frac{\left\Vert \boldsymbol{h}\right\Vert }{\left\Vert \boldsymbol{z}\right\Vert }\right)^{2}\left\Vert \boldsymbol{z}\right\Vert ^{2},
	\label{eq:ub2}
\end{eqnarray}
where both inequalities are immediate consequences from the definitions
of $\mathcal{D}_{\gamma}^{i,2}$ and $\mathcal{E}_{1}^i$ (see (\ref{eq:defn-D_gamma-2})
and (\ref{eq:defn-E1})). Taking this together with the cardinality
bound (\ref{eq:card-I3}) and picking $\epsilon$ appropriately, we
get
\begin{align}
	\frac{1}{m}\sum_{i=1}^{m}\frac{\left|\boldsymbol{a}_{i}^{\top}\boldsymbol{h}\right|^{3}}{\left|\boldsymbol{a}_{i}^{\top}\boldsymbol{z}\right|}{\bf 1}_{\mathcal{E}_{1}^{i}\cap\mathcal{D}_{\gamma}^{i,2}} 
	 \leq  
	\left\{ \underset{\vartheta_{1}}{\underbrace{\frac{\left(2+\frac{\gamma}{\alpha_{z}^{\mathrm{lb}}}\frac{\left\Vert \boldsymbol{h}\right\Vert }{\left\Vert \boldsymbol{z}\right\Vert }\right)\left(2\alpha_{z}^{\mathrm{ub}}+\gamma\frac{\left\Vert \boldsymbol{h}\right\Vert }{\left\Vert \boldsymbol{z}\right\Vert }\right)^{2}}{9800\left(\alpha_{z}^{\mathrm{lb}}\right)^{4}}}}\frac{\left\Vert \boldsymbol{h}\right\Vert ^{2}}{\left\Vert \boldsymbol{z}\right\Vert ^{2}}
		+ \epsilon\right\} \left\Vert \boldsymbol{h}\right\Vert ^{2}.
	\label{eq:ub3}
\end{align}
Furthermore, under the condition that
\[
	\gamma\leq\alpha_{z}^{\mathrm{lb}}\frac{\left\Vert \boldsymbol{z}\right\Vert }{\left\Vert \boldsymbol{h}\right\Vert }
	\quad\text{and}\quad
	\frac{\left\Vert \boldsymbol{h}\right\Vert }{\left\Vert \boldsymbol{z}\right\Vert }\leq\frac{\sqrt{98}\left(\alpha_{z}^{\mathrm{lb}}\right)^{2}}{\sqrt{3}\left(2\alpha_{z}^{\mathrm{ub}}+\alpha_{z}^{\mathrm{lb}}\right)},
\]
one can simplify (\ref{eq:ub3}) by observing that $\vartheta_{1}\leq\frac{1}{100}$, which results in
%To summarize, for all $\left(\boldsymbol{h},\boldsymbol{z}\right)$
%obeying
%
%\begin{equation}
%	\frac{\left\Vert \boldsymbol{h}\right\Vert }{\left\Vert \boldsymbol{z}\right\Vert }
%	\leq\min\left\{ \frac{\alpha_{z}^{\mathrm{lb}}}{\gamma},\text{ }\frac{\alpha_{z}^{\mathrm{lb}}}{6},\text{ }\frac{\sqrt{98/3}\left(\alpha_{z}^{\mathrm{lb}}\right)^{2}}{2\alpha_{z}^{\mathrm{ub}}+\alpha_{z}^{\mathrm{lb}}}\right\} ,
%	\label{eq:ConditionGamma-2}
%\end{equation}
%
%one has, with exponentially high probability, that
%
\begin{eqnarray}
	\frac{1}{m}\sum_{i=1}^{m}\frac{\left|\boldsymbol{a}_{i}^{\top}\boldsymbol{h}\right|^{3}}{\left|\boldsymbol{a}_{i}^{\top}\boldsymbol{z}\right|}{\bf 1}_{\mathcal{E}_{1}^{i}\cap\mathcal{D}_{\gamma}^{i,2}} 
	& \leq & \left(\frac{1}{100}+\epsilon\right)\left\Vert \boldsymbol{h}\right\Vert ^{2}.
	\label{eq:I3}
\end{eqnarray}

Putting all preceding results in this subsection together reveals
that with probability exceeding $1-\exp\left(-\Omega\left(m\right)\right)$,
\begin{eqnarray}
	-\left\langle \boldsymbol{h}, \frac{1}{2m}\nabla\ell_{\mathrm{tr}}\left(\boldsymbol{z}\right)\right\rangle  
	& \geq & \left\{ 1.99-2\left(\zeta_{1}+\zeta_{2}\right)-\sqrt{8/\pi}\frac{\left\Vert \boldsymbol{h}\right\Vert }{\alpha_{z}^{\mathrm{lb}}\left\Vert \boldsymbol{z}\right\Vert }-3\epsilon\right\} \left\Vert \boldsymbol{h}\right\Vert ^{2}\nonumber \\
 	& \geq & \left\{ 1.99-2\left(\zeta_{1}+\zeta_{2}\right)-\sqrt{8/\pi}(3\alpha_{h})^{-1}-3\epsilon\right\} \Vert \boldsymbol{h} \Vert ^{2}
	\label{eq:regularity}
\end{eqnarray}
holds simultaneously over all $\boldsymbol{x}$ and $\boldsymbol{z}$
satisfying
\begin{equation}
	\frac{\left\Vert \boldsymbol{h}\right\Vert }{\left\Vert \boldsymbol{z}\right\Vert }
	\leq \min\left\{ \frac{\alpha_{z}^{\mathrm{lb}}}{3\alpha_{h}},\text{ }\frac{\alpha_{z}^{\mathrm{lb}}}{6},\text{ }\frac{\sqrt{98/3}\left(\alpha_{z}^{\mathrm{lb}}\right)^{2}}{2\alpha_{z}^{\mathrm{ub}}+\alpha_{z}^{\mathrm{lb}}},\frac{1}{11}\right\}
\end{equation}
as claimed in Proposition \ref{prop-regularity-noiseless}.
%and in turn establishes Proposition \ref{prop-Contraction}.

To conclude this section, we provide a tighter estimate about the norm
of the regularized gradient.

\begin{lemma}\label{Lemma:norm-score} Fix $\delta>0$, and assume that $y_i = (\boldsymbol{a}_i^{\top}\boldsymbol{x})^2 $. 
Suppose that $m\geq c_{0}n$ for some large constant $c_{0}>0$. There exist some
universal constants $c,C>0$ such that with probability at least $1-C\exp\left(-cm\right)$,
\begin{equation}
	\frac{1}{m}\big\|\nabla\ell_{\mathrm{tr}}\left(\boldsymbol{z}\right)\big\|\hspace{0.3em}
	\leq\hspace{0.3em}\left(1+\delta\right)\cdot4\sqrt{1.02+ 0.665 / \alpha_{h}}\left\Vert \boldsymbol{h}\right\Vert 
	\label{eq:UB-norm-grad}
\end{equation}
holds simultaneously for all $\boldsymbol{x}$, $\boldsymbol{z}\in \mathbb{R}^n$ satisfying
%
%\begin{equation}
$	\frac{\left\Vert \boldsymbol{h}\right\Vert }{\left\Vert \boldsymbol{z}\right\Vert }
	\leq \min\Big\{ \frac{\alpha_{z}^{\mathrm{lb}}}{3\alpha_{h}},
		\frac{\alpha_{z}^{\mathrm{lb}}}{6},
		\frac{\sqrt{98/3}\left(\alpha_{z}^{\mathrm{lb}}\right)^{2}}{2\alpha_{z}^{\mathrm{ub}}+\alpha_{z}^{\mathrm{lb}}},
		\frac{1}{11}\Big\} .
$
%\end{equation}
%
\end{lemma}

%\begin{proof}See Appendix \ref{sec:Proof-of-Lemma-norm-score}.\end{proof}

Lemma \ref{Lemma:norm-score} complements the preceding arguments by allowing us to identify
a concrete plausible range for the step size. Specifically, putting Lemma \ref{Lemma:norm-score} and Proposition
\ref{prop-regularity-noiseless} together suggests that 
\begin{equation}
	-\Big\langle\boldsymbol{h},\frac{1}{m}\nabla\ell_{\mathrm{tr}}\left(\boldsymbol{z}\right)\Big\rangle\hspace{0.3em}
	\geq \hspace{0.3em}\frac{2\left\{ 1.99 - 2\left(\zeta_{1} + \zeta_{2}\right) - \sqrt{8/(9\pi)}\alpha_{h}^{-1} - \epsilon\right\} }{\left( 1 + \delta\right)^{2} \cdot 16\left( 1.02 + 0.665 / \alpha_{h} \right)}\left\Vert \frac{1}{m}\nabla\ell_{\mathrm{tr}}\left(\boldsymbol{z}\right)\right\Vert ^{2}.
\end{equation}
Taking $\epsilon$ and $\delta$ to be sufficiently small we arrive
at a feasible range (cf. Definition (\ref{eq:Regularity-Truncate-Original}))
\begin{equation}
	\mu \leq \frac{ 0.994 - \zeta_{1} - \zeta_{2} - \sqrt{ 2 / (9\pi) }\alpha_{h}^{-1}} {2 \left(1.02 + 0.665/\alpha_{h} \right)} := \mu_0.
	\label{eq:feasible-step-size}
\end{equation}
%
% In addition, for any $\mu < \mu_0$, setting $\delta, \epsilon$ to be small in Proposition \ref{prop-regularity-noiseless} and Lemma \ref{Lemma:norm-score} indicates that $\mathsf{RC}(\mu, \lambda, \epsilon)$ holds for
%
% \begin{eqnarray*}
%	\lambda & \leq &  4\left( 1.98-2\left(\zeta_{1}+\zeta_{2}\right)-\sqrt{8/\pi}(3\alpha_{h})^{-1}  \right)
%		-  \mu \cdot 16 (1.02 + 0.665/\alpha_h)   \\
%		& = & \text{ }  16 (1.02 + 0.665/\alpha_h) ( \mu_0 - \mu ),
%4 (0.994 - \zeta_1 - \zeta_2 - \sqrt{2/(9\pi)}\alpha_h^{-1} - \mu/\mu_0)	
% \end{eqnarray*}
%
%and the convergence rate is given by $1-\rho_0 = 1 - \mu \lambda$. 
This establishes Proposition \ref{prop-Contraction} and in turn Theorem
\ref{theorem-Truncated-WF} when $\mu_{t}$ is taken to be a fixed
constant. 

To justify the contraction under backtracking line search, it suffices
to prove that the resulting step size falls within this range
(\ref{eq:feasible-step-size}), which we defer to Appendix
\ref{sec:Backtracking-line-search}.

\section{Stability \label{sec:Proof-of-Theorem-Noisy}}

This section goes in the direction of establishing stability
guarantees of TWF. We concentrate on the iterative gradient stage, and
defer the analysis for the initialization stage to Appendix
\ref{proof-truncated-spectral}.

Before continuing, we collect two bounds that we shall use several times.
The first is the observation that
\begin{align}
	\frac{1}{m}\|\boldsymbol{y}-\mathcal{A}(\boldsymbol{z}\boldsymbol{z}^{\top})\|_1 
	&\leq \frac{1}{m}\|\mathcal{A}(\boldsymbol{x}\boldsymbol{x}^{\top}-\boldsymbol{z}\boldsymbol{z}^{\top})\|_1 + \frac{1}{m}\|\boldsymbol{\eta}\|_1	 \nonumber\\
	&\lesssim \|\boldsymbol{h}\| \|\boldsymbol{z} \| + \frac{1}{m}\|\boldsymbol{\eta}\|_1
	\lesssim  \|\boldsymbol{h}\| \| \boldsymbol{z} \| + \frac{1}{\sqrt{m}} \| \boldsymbol{\eta} \|, 
\end{align}
where the last inequality follows from Cauchy-Schwarz.  Setting 
\[
v_i := 2\frac{y_i - | \boldsymbol{a}_i^{\top} \boldsymbol{z} |^2
}{\boldsymbol{a}_i^{\top} \boldsymbol{z}} {\bf 1}_{\mathcal{E}_1^i
  \cap \mathcal{E}_2^i}
\]
as usual, this inequality together with the trimming rules
$\mathcal{E}_1^i$ and $\mathcal{E}_2^i$ gives
\begin{equation}
\begin{array}{c}
	|v_i|
	\lesssim \|\boldsymbol{h}\|  + \frac{\Vert \boldsymbol{\eta} \Vert }{  \sqrt{m} \| \boldsymbol{z} \| }\\
	\Longrightarrow \quad \left\| \frac{1}{m} \nabla\ell_{\mathrm{tr}}(\boldsymbol{z}) \right\| \text{ }=\text{ } \frac{1}{m} \| \boldsymbol{A}^{\top} \boldsymbol{v} \| 
	~ \leq~ \left\| \frac{1}{\sqrt{m}}\boldsymbol{A} \right\|  \frac{ 1 }{ \sqrt{m} } \|\boldsymbol{v}\|
	~ \overset{\text{(i)}}{\lesssim}~ \frac{1}{\sqrt{m}} \| \boldsymbol{v} \| 
	~\lesssim~  \| \boldsymbol{h} \| + \frac{\Vert \boldsymbol{\eta} \Vert }{  \sqrt{m} \| \boldsymbol{z} \| }, 
\end{array}
	\label{eq:norm_noisy}
\end{equation}
where (i) arises from \cite[Corollary 5.35]{Vershynin2012}.

%As discussed earlier

As discussed in Section \ref{sec:Why-it-works}, the estimation error is contractive if 
$-\frac{1}{m} \nabla \ell_{\mathrm{tr}}\left(\boldsymbol{z}\right)$ 
satisfies the regularity condition. With (\ref{eq:norm_noisy}) in place, $\mathsf{RC}$ reduces to
\begin{eqnarray}
	-\frac{1}{m}\left\langle \nabla \ell_{\mathrm{tr}}\left(\boldsymbol{z}\right),\boldsymbol{h}\right\rangle  
	& \gtrsim & ~ \| \boldsymbol{h} \|^2.
	\label{eq:RC-stability}
\end{eqnarray}
Unfortunately, (\ref{eq:RC-stability}) does not hold for all
$\boldsymbol{z}$ within the neighborhood of $\boldsymbol{x}$ due to
the existence of noise.  Instead we establish the following: 
%\ejc{Perhaps you
%  want to give names to these regimes?}
\begin{itemize}
	\item
 	The condition (\ref{eq:RC-stability}) holds for all $\boldsymbol{h}$ obeying
	\begin{equation}
		c_{3}\frac{\left\Vert
    	\boldsymbol{\eta}\right\Vert /\sqrt{m}}{\left\Vert
    	\boldsymbol{z}\right\Vert } \leq \Vert \boldsymbol{h}\Vert
	\leq c_4\Vert \boldsymbol{x} \Vert 
		\label{eq:case1}
	\end{equation}
	for some constants $c_3, c_4>0$ (we shall call it {\em Regime
          1}); this will be proved later. In this regime, the reasoning
        in Section \ref{sec:Why-it-works} gives %
	\begin{equation}
		\mathrm{dist} \Big(\boldsymbol{z}+ \frac{\mu}{m}  \nabla \ell_{\mathrm{tr}}(\boldsymbol{z}), ~\boldsymbol{x} \Big)  \leq  (1 - \rho) \mathrm{dist}(\boldsymbol{z}, \boldsymbol{x} )
		\label{eq:case1error}
	\end{equation}
	for some appropriate constants $\mu, \rho>0$ and, hence, error
        contraction occurs as in the noiseless setting. 
      \item However, once the iterate enters {\em Regime 2} where
	\begin{equation}
		\left\Vert \boldsymbol{h}\right\Vert \leq \frac{c_{3}\left\Vert
    		\boldsymbol{\eta}\right\Vert }{ \sqrt{m} \left\Vert
    		\boldsymbol{z}\right\Vert }, 
		\label{eq:case2}
	\end{equation}
	the estimation error might no longer be contractive.
        Fortunately, in this regime each move by
        $\frac{\mu}{m}\nabla {\ell}_{\mathrm{tr}}
        \left(\boldsymbol{z}\right)$
        is of size at most
        $O(\frac{\Vert \boldsymbol{\eta} \Vert}{ \sqrt{m} \|
          \boldsymbol{z} \| })$,
        compare (\ref{eq:norm_noisy}).  As a result, at each iteration
        the estimation error cannot increase by more than a numerical
        constant times
        $\frac{\Vert\boldsymbol{\eta} \Vert }{ \sqrt{m} \|
          \boldsymbol{z} \|}$
        before possibly jumping out (of this regime). Therefore,
	\begin{equation}
	\mathrm{dist} \big(\boldsymbol{z}+ \frac{\mu}{m}  \nabla \ell_{\mathrm{tr}}(\boldsymbol{z}), ~\boldsymbol{x} \big)  
	 \leq  c_5\frac{\|\boldsymbol{\eta}\|}{\sqrt{m}\|\boldsymbol{x}\|}  
	\label{eq:case2error}
	\end{equation}
	for some constant $c_5>0$.  Moreover, as long as
        $\|\boldsymbol{\eta}\|_{\infty} / \|\boldsymbol{x}\|^2$ is
        sufficiently small, one can guarantee that
        $c_5\frac{\|\boldsymbol{\eta}\|}{\sqrt{m}\|\boldsymbol{x}\|}
        \leq c_5 \frac{\| \boldsymbol{\eta}
          \|_{\infty}}{\|\boldsymbol{x}\|} \leq
        c_4\|\boldsymbol{x}\|$.
        In other words, if the iterate jumps out of Regime 2, it will
        still fall within Regime 1. 
\end{itemize}
To summarize, suppose the initial guess $\boldsymbol{z}^{(0)}$ obeys
$\mathrm{dist}(\boldsymbol{z}^{(0)}, \boldsymbol{x}) \leq c_4 \|
\boldsymbol{x} \|$.
Then the estimation error will shrink at a geometric rate $1-\rho$
before it enters Regime 2. Afterwards, $\boldsymbol{z}^{(t)}$ will
either stay within Regime 2 or jump back and forth between Regimes 1
and 2. Because of the bounds (\ref{eq:case2error}) and
(\ref{eq:case1error}), the estimation errors will never exceed the
order of $\frac{\|\boldsymbol{\eta}\|}{\sqrt{m}\|\boldsymbol{x}\|}$
from then on.
%\ejc{How do we know what we will not exceed $c_4\|x\|$?}  
Putting these together establishes
(\ref{eq:noisy-converge}), namely, the first part of the theorem.

Below we justify the condition (\ref{eq:RC-stability}) for Regime 1, for which we  
   start by gathering additional properties of the trimming
rules. By Cauchy-Schwarz,  $\frac{1}{m}\left\Vert
  \boldsymbol{\eta}\right\Vert _{1}\leq\frac{1}{\sqrt{m}}\left\Vert
  \boldsymbol{\eta}\right\Vert \leq\frac{1}{c_{3}}\left\Vert
  \boldsymbol{h}\right\Vert \left\Vert \boldsymbol{z}\right\Vert $.
When $c_{3}$ is sufficiently large, applying Lemmas
\ref{lemma:RIP-Candes} and \ref{lemma: bound-h} gives
\begin{align}
\begin{array}{ll}
	\frac{1}{m}\sum\nolimits_{l=1}^{m}\left|y_{l}-\left|\boldsymbol{a}_{l}^{\top}\boldsymbol{z}\right|^{2}\right| 
		 \leq  \frac{1}{m}\left\Vert \mathcal{A}\left(\boldsymbol{x}\boldsymbol{x}^{\top}-\boldsymbol{z}\boldsymbol{z}^{\top}\right)\right\Vert _{1}+\frac{1}{m}\left\Vert \boldsymbol{\eta}\right\Vert _{1}
		\leq 2.98 \Vert \boldsymbol{h} \Vert \Vert \boldsymbol{z} \Vert ; \\
	\frac{1}{m}\sum\nolimits_{l=1}^{m}\left|y_{l}-\left|\boldsymbol{a}_{l}^{\top}\boldsymbol{z}\right|^{2}\right| 
		 \geq  \frac{1}{m}\left\Vert \mathcal{A}\left(\boldsymbol{x}\boldsymbol{x}^{\top}-\boldsymbol{z}\boldsymbol{z}^{\top}\right)\right\Vert _{1}-\frac{1}{m}\left\Vert \boldsymbol{\eta}\right\Vert _{1}
		\geq 1.151 \Vert \boldsymbol{h} \Vert \Vert \boldsymbol{z} \Vert .
	\end{array}
	\label{eq:A_UB}
\end{align}
From now on, we shall denote
$
	\tilde{\mathcal{E}}_{2}^{i}:=\Big\{\left||\boldsymbol{a}_{i}^{\top}\boldsymbol{x}|^{2}-|\boldsymbol{a}_{i}^{\top}\boldsymbol{z}|^{2}\right|\leq\frac{\alpha_{h}}{m}\left\Vert \boldsymbol{y}-\mathcal{A}\left(\boldsymbol{z}\boldsymbol{z}^{\top}\right)\right\Vert _{1}\frac{|\boldsymbol{a}_{i}^{\top}\boldsymbol{z}|}{\left\Vert \boldsymbol{z}\right\Vert } \Big\}
$
to differentiate from $\mathcal{E}_{2}^{i}$. 
For any small constant $\epsilon>0$, we introduce the index set $\mathcal{G}:=\left\{ i: |\eta_{i}|\leq C_{\epsilon}\left\Vert \boldsymbol{\eta}\right\Vert /\sqrt{m}\right\} $
that satisfies $\left|\mathcal{G}\right| = (1-\epsilon)m$.  Note that $C_{\epsilon}$ must be bounded as $n$ scales, since  
\begin{equation}
	\| \boldsymbol{\eta} \|^2 ~\geq~ \sum\nolimits_{i\notin \mathcal{G}} \eta_i^2 ~\geq~ (m- |\mathcal{G}|) \cdot  C_{\epsilon}^2 \|\boldsymbol{\eta}\|^2/m  ~\geq~ \epsilon C_{\epsilon}^2 \| \boldsymbol{\eta} \|^2
	\quad \Rightarrow \quad
	C_{\epsilon} \leq 1 / \sqrt{\epsilon}.
\end{equation}

We are now ready to analyze the regularized gradient, which we separate into several components as follows
% component
% $\nabla_{\mathrm{tr}}^{\mathrm{clean}}\mathcal{\ell}\left(\boldsymbol{z}\right)$
% and a noise component
% $\nabla_{\mathrm{tr}}^{\mathrm{noise}}\mathcal{\ell}\left(\boldsymbol{z}\right)$
% as follows:
%
% \begin{eqnarray}
%	\nabla\mathcal{\ell}_{\mathrm{tr}} ( \boldsymbol{z} ) 
%	& = & \text{ }  \underset{:=\nabla_{\mathrm{tr}}^{\mathrm{clean}}\mathcal{\ell} (\boldsymbol{z})}
%		{\underbrace{ 2\sum_{i=1}^{m}\frac{ |\boldsymbol{a}_{i}^{\top}\boldsymbol{x} |^{2}- |\boldsymbol{a}_{i}^{\top}\boldsymbol{z} |^{2}}{\boldsymbol{a}_{i}^{\top}\boldsymbol{z}}\boldsymbol{a}_{i}{\bf 1}_{\mathcal{E}_{1}^{i}\cap\mathcal{E}_{2}^{i}}}}
%	+ \underset{:=\nabla_{\mathrm{tr}}^{\mathrm{noise}}\mathcal{\ell} (\boldsymbol{z} )}
%		{\underbrace{2\sum_{i=1}^{m}\frac{\eta_{i}}{\boldsymbol{a}_{i}^{\top}\boldsymbol{z}}\boldsymbol{a}_{i}{\bf 1}_{\mathcal{E}_{1}^{i}\cap\mathcal{E}_{2}^{i}}}}.
% \end{eqnarray}
%
\begin{align}
	&\nabla_{\mathrm{tr}}\mathcal{\ell}\left(\boldsymbol{z}\right) 
	 =  \text{ }
	\underset{:=\nabla_{\mathrm{tr}}^{\mathrm{clean}}\mathcal{\ell}\left(\boldsymbol{z}\right)}{\underbrace{2\sum_{i\in\mathcal{G}}\frac{\left|\boldsymbol{a}_{i}^{\top}\boldsymbol{x}\right|^{2}-\left|\boldsymbol{a}_{i}^{\top}\boldsymbol{z}\right|^{2}}{\boldsymbol{a}_{i}^{\top}\boldsymbol{z}}\boldsymbol{a}_{i}{\bf 1}_{\mathcal{E}_{1}^{i}\cap\mathcal{E}_{2}^{i}}+2\sum_{i\notin{\mathcal{G}}}\frac{\left|\boldsymbol{a}_{i}^{\top}\boldsymbol{x}\right|^{2}-\left|\boldsymbol{a}_{i}^{\top}\boldsymbol{z}\right|^{2}}{\boldsymbol{a}_{i}^{\top}\boldsymbol{z}}\boldsymbol{a}_{i}{\bf 1}_{\mathcal{E}_{1}^{i}\cap\tilde{\mathcal{E}}_{2}^{i}}}} \nonumber\\
 	& \quad 
	+~ \underset{:=\nabla_{\mathrm{tr}}^{\mathrm{noise}}\mathcal{\ell}\left(\boldsymbol{z}\right)}{\underbrace{2\sum_{i\in\mathcal{G}}\frac{\eta_{i}}{\boldsymbol{a}_{i}^{\top}\boldsymbol{z}}\boldsymbol{a}_{i}{\bf 1}_{\mathcal{E}_{1}^{i}\cap\mathcal{E}_{2}^{i}}}}
	+\text{ }\underset{:=\nabla_{\mathrm{tr}}^{\mathrm{extra}}\mathcal{\ell}\left(\boldsymbol{z}\right)}{\underbrace{2\sum_{i\notin{\mathcal{G}}}\left(\frac{y_{i}-\left|\boldsymbol{a}_{i}^{\top}\boldsymbol{z}\right|^{2}}{\boldsymbol{a}_{i}^{\top}\boldsymbol{z}}{\bf 1}_{\mathcal{E}_{1}^{i}\cap\mathcal{E}_{2}^{i}}-\frac{\left|\boldsymbol{a}_{i}^{\top}\boldsymbol{x}\right|^{2}-\left|\boldsymbol{a}_{i}^{\top}\boldsymbol{z}\right|^{2}}{\boldsymbol{a}_{i}^{\top}\boldsymbol{z}}{\bf 1}_{\mathcal{E}_{1}^{i}\cap\tilde{\mathcal{E}}_{2}^{i}}\right)\boldsymbol{a}_{i}}}.
\end{align}
\begin{itemize}
\item
For each index $i\in\mathcal{G}$,  the inclusion property (\ref{eq:Event2}) 
(i.e. $\mathcal{E}_{3}^{i}  \subseteq  \mathcal{E}_{2}^{i} \subseteq \mathcal{E}_{4}^{i}$) holds. To see this, observe that 
\[
	\big|y_i - |\boldsymbol{a}_i^{\top}\boldsymbol{z}|^2 \big| \in \Big[ \big||\boldsymbol{a}_i^{\top}\boldsymbol{x}|^2 - |\boldsymbol{a}_i^{\top}\boldsymbol{z}|^2 \big| \pm |\eta_i| \Big]. 
\]
Since $|\eta_i| \leq C_{\epsilon}\|\boldsymbol{\eta}\|/\sqrt{m} \ll \|\boldsymbol{h}\|\|\boldsymbol{z}\|$ when $c_3$ is sufficiently large,  one can derive the inclusion (\ref{eq:Event2})  immediately from (\ref{eq:A_UB}). As a result, all the proof arguments for Proposition
\ref{prop-regularity-noiseless} carry over to
$\nabla_{\mathrm{tr}}^{\mathrm{clean}}\mathcal{\ell}\left(\boldsymbol{z}\right)$,
suggesting that
\begin{equation}
	-\Big\langle\boldsymbol{h}, \frac{1}{m}\nabla_{\mathrm{tr}}^{\mathrm{clean}}\ell\left(\boldsymbol{z}\right)\Big\rangle\hspace{0.3em}
	\geq \hspace{0.3em}2\left\{ 1.99 - 2\left(\zeta_{1}+\zeta_{2}\right) - \sqrt{8/(9\pi)}\alpha_{h}^{-1} - \epsilon\right\} 
		\Vert \boldsymbol{h} \Vert ^{2}.
	\label{eq:regularity-noisy}
\end{equation}

\item
Next, letting
$w_{i}=\frac{2\eta_{i}}{\boldsymbol{a}_{i}^{\top}\boldsymbol{z}}{\bf
  1}_{\mathcal{E}_{1}^{i}\cap\mathcal{E}_{2}^{i}} {\bf 1}_{\{i\in\mathcal{G}\}}$, we see that for
any constant $\delta>0$, the noise component obeys
\begin{align}
	\left\Vert \frac{1}{m}\nabla_{\mathrm{tr}}^{\mathrm{noise}}\ell (\boldsymbol{z}) \right\Vert  
	 &~=~  \left\Vert \frac{1}{m}\boldsymbol{A}^{\top} \boldsymbol{w}\right\Vert 
		\text{ }\leq\text{ } \left\Vert \frac{1}{\sqrt{m}}\boldsymbol{A} \right\Vert \left\Vert \frac{1}{\sqrt{m}}\boldsymbol{w}\right\Vert \nonumber\\
		&~\overset{\text{(ii)}}{\leq}~  \frac{1+\delta}{\sqrt{m}} \Vert \boldsymbol{w} \Vert 
 	 \text{ }\leq\text{ }  (1 + \delta) \frac{ 2\Vert \boldsymbol{\eta} \Vert /\sqrt{m}}{ \alpha_z^{\mathrm{lb}}\Vert \boldsymbol{z} \Vert },
	\label{eq:noise-grad-norm}
\end{align}
provided that $m/n$ is sufficiently large. Here, (ii) arises from \cite[Corollary 5.35]{Vershynin2012}, and the
last inequality is a consequence of the upper estimate
\begin{equation}
	\Vert \boldsymbol{w} \Vert ^{2}\text{ } 
		\leq 4\sum_{i=1}^{m} \frac{ \left|\eta_{i}\right|^{2} }{ (\boldsymbol{a}_{i}^{\top}\boldsymbol{z})^{2} }{\bf 1}_{\mathcal{E}_{1}^{i}\cap\mathcal{E}_{2}^{i}}
	\text{ }\leq\text{ } 4 \sum_{i=1}^{m}\frac{|\eta_{i} |^{2}}{(\alpha_{z}^{\mathrm{lb}}\left\Vert \boldsymbol{z}\right\Vert )^{2}}
	\text{ }=\text{ } \frac{ 4\left\Vert \boldsymbol{\eta}\right\Vert ^{2}}{ (\alpha_{z}^{\mathrm{lb}} \Vert \boldsymbol{z} \Vert )^{2}}.
\end{equation}
In turn, this immediately gives
\begin{align}
	\left|\Big\langle\boldsymbol{h},\frac{1}{m} \nabla_{\mathrm{tr}}^{\mathrm{noise}}\ell \left(\boldsymbol{z}\right)\Big\rangle\right| 
	 ~\leq~ \left\Vert \boldsymbol{h}\right\Vert \left\Vert \frac{1}{m} \nabla_{\mathrm{tr}}^{\mathrm{noise}} \ell \left(\boldsymbol{z}\right)\right\Vert 
	\text{ }\leq\text{ } \frac{ 2\left(1 + \delta\right) }{ \alpha_{z}^{\mathrm{lb}} } \frac{ \left\Vert \boldsymbol{\eta}\right\Vert }{ \sqrt{m} \Vert \boldsymbol{z} \Vert } \Vert \boldsymbol{h} \Vert.
\end{align}

\item 
We now turn to the last term $\nabla_{\mathrm{tr}}^{\mathrm{extra}}\mathcal{\ell}\left(\boldsymbol{z}\right)$. 
According to the definition of $\mathcal{E}_{2}^{i}$ and
$\tilde{\mathcal{E}}_{2}^{i}$ as well as the property (\ref{eq:A_UB}),
the weight
$q_{i}:=2\Big(\frac{y_{i}-\left|\boldsymbol{a}_{i}^{\top}\boldsymbol{z}\right|^{2}}{\boldsymbol{a}_{i}^{\top}\boldsymbol{z}}{\bf 1}_{\mathcal{E}_{1}^{i}\cap\mathcal{E}_{2}^{i}}-\frac{\left|\boldsymbol{a}_{i}^{\top}\boldsymbol{x}\right|^{2}-\left|\boldsymbol{a}_{i}^{\top}\boldsymbol{z}\right|^{2}}{\boldsymbol{a}_{i}^{\top}\boldsymbol{z}}{\bf 1}_{\mathcal{E}_{1}^{i}\cap\tilde{\mathcal{E}}_{2}^{i}}\Big){\bf 1}_{\{i\notin\mathcal{G}\}}$
is bounded in magnitude  by $6\|\boldsymbol{h}\|$.  This gives
\[
	\left\Vert \boldsymbol{q}\right\Vert \leq\sqrt{m-\left|{\mathcal{G}}\right|}\cdot6\|\boldsymbol{h}\|\leq6 \sqrt{\epsilon m}\|\boldsymbol{h}\|,
\]
%
%which leads to
and hence
\begin{equation}
	%\Rightarrow\quad
	\Big|\Big\langle\frac{1}{m}\nabla_{\mathrm{tr}}^{\mathrm{extra}}\mathcal{\ell}\left(\boldsymbol{z}\right),\boldsymbol{h}\Big\rangle\Big|\leq\|\boldsymbol{h}\|\cdot\big\|\frac{1}{m}\nabla_{\mathrm{tr}}^{\mathrm{extra}}\mathcal{\ell}\left(\boldsymbol{z}\right)\big\|
	= \frac{1}{m}\|\boldsymbol{h}\|\cdot\big\|\boldsymbol{A}^{\top}\boldsymbol{q}\big\|
	\leq 6\left(1+\delta\right)\sqrt{\epsilon}\|\boldsymbol{h}\|^{2}.
\end{equation}
\end{itemize}
Taking the above bounds together yields
\begin{align*}
	&-\frac{1}{m}\left\langle \nabla \ell_{\mathrm{tr}} \left(\boldsymbol{z}\right),\boldsymbol{h}\right\rangle  
	~ \geq ~ \text{ }2 \left\{ 1.99 - 2\left(\zeta_{1}+\zeta_{2}\right) - \sqrt{\frac{8}{9\pi}}\frac{1}{\alpha_{h}} - 6(1+\delta)\sqrt{\epsilon} - \epsilon\right\} 
		\left\Vert \boldsymbol{h} \right\Vert ^{2} \\
	&\qquad\qquad\qquad\qquad - \frac{2\left(1+\delta\right)}{\alpha_{z}^{\mathrm{lb}}}\frac{\left\Vert \boldsymbol{\eta}\right\Vert }{\sqrt{m}\left\Vert \boldsymbol{z}\right\Vert } \Vert \boldsymbol{h} \Vert .
\end{align*}
Since $\Vert \boldsymbol{h} \Vert \geq c_{3}\frac{\Vert \boldsymbol{\eta} \Vert }{\sqrt{m} \Vert \boldsymbol{z} \Vert }$
for some large constant $c_3 > 0$, setting $\epsilon$ to be small one obtains 
\begin{eqnarray}
	-\frac{1}{m}\left\langle \nabla \ell_{\mathrm{tr}}\left(\boldsymbol{z}\right),\boldsymbol{h}\right\rangle  
	& \geq & \text{ } 
	2\left\{ 1.95-2\left(\zeta_{1}+\zeta_{2}\right)-\sqrt{8/(9\pi)}\alpha_{h}^{-1}\right\} 
	\left\Vert \boldsymbol{h}\right\Vert ^{2}
\end{eqnarray}
%
%Now recall from (\ref{eq:norm_noisy}) that in this regime, $\big
%\|\frac{1}{m}\nabla\ell_{\mathrm{tr}}(\boldsymbol{z}) \big\| \lesssim
%\| \boldsymbol{h} \|$.  Putting these together establishes the
%regularity condition 
for all $\boldsymbol{h}$ obeying 
\[
	\frac{c_{3} \Vert \boldsymbol{\eta} \Vert /\sqrt{m}}{\Vert \boldsymbol{z} \Vert }\text{ }\leq\text{ } \Vert \boldsymbol{h} \Vert 
	\text{ }\leq\text{ }   \min\left\{ \frac{1}{11},\frac{\alpha_{z}^{\mathrm{lb}}}{3\alpha_{h}},\text{ }\frac{\alpha_{z}^{\mathrm{lb}}}{6},\text{ }\frac{\sqrt{98/3}\left(\alpha_{z}^{\mathrm{lb}}\right)^{2}}{2\alpha_{z}^{\mathrm{ub}}
		+ \alpha_{z}^{\mathrm{lb}}}\right\} \Vert \boldsymbol{z} \Vert, 
\]
which finishes the proof of Theorem \ref{theorem-Truncated-WF-noisy} for general $\boldsymbol{\eta}$.

%The point is that the estimation error will not
%  exceed a constant times $\frac{\Vert \boldsymbol{\eta} \Vert }{
%    \sqrt{m} \| \boldsymbol{z} \| }$ in this regime.

Up until now, we have established the theorem for general
$\boldsymbol{\eta}$, and it remains to specialize it to the Poisson
model. Standard concentration results, which we omit, give
\begin{equation}
	\frac{1}{m}\left\Vert \boldsymbol{\eta}\right\Vert ^{2}\approx\frac{1}{m}\sum_{i=1}^{m}\mathbb{E}\left[\eta_{i}^{2}\right]=\frac{1}{m}\sum_{i=1}^{m}\left(\boldsymbol{a}_{i}^{\top}\boldsymbol{x}\right)^{2}\approx\|\boldsymbol{x}\|^{2}
	\label{eq:eta_2}
\end{equation}
with high probability. Substitution into (\ref{eq:noisy-converge}) completes the proof.

\section{Minimax lower bound \label{sec:proof-of-Minimax-Lower-Boundsi}}

The goal of this section is to establish the minimax lower bound given
in Theorem \ref{theorem-converse}. For notational simplicity, we
denote by $\mathbb{P}\left(\boldsymbol{y}\mid\boldsymbol{w}\right)$
the likelihood of $y_{i} \stackrel{\text{ind.}}{\sim}
\mathsf{Poisson}(|\boldsymbol{a}_{i}^{\top}\boldsymbol{w}|^{2})$, $1
\leq i \leq m$ conditional on $\{\boldsymbol{a}_i\}$. For any two
probability measures $P$ and $Q$, we denote by $\mathsf{KL}\left(P\|Q\right)$ the Kullback\textendash{}Leibler (KL) 
%and $\chi^{2}$-
divergence between them:
\begin{equation}
	\mathsf{KL}\left(P\|Q\right):={\displaystyle \int}\log\left(\frac{\mathrm{d}P}{\mathrm{d}Q}\right)\mathrm{d}P,
	%\quad
	%\text{and}
	\label{eq:KL-chi2-defn}
\end{equation}

The basic idea is to adopt the general reduction scheme discussed
in \cite[Section 2.2]{tsybakov2009introduction}, which amounts to
finding a finite collection of hypotheses that are minimally separated.
Below we gather one result useful for constructing and analyzing
such hypotheses. 

\begin{lemma}\label{lemma:MinimaxConstruction}
  %Fix any $0<\zeta<1$. 
  Suppose that
  $\boldsymbol{a}_{i}\sim\mathcal{N}\left({\bf
      0},\boldsymbol{I}_{n}\right)$,
  $n$ is sufficiently large, and $m=\kappa n$ for some sufficiently
  large constant $\kappa > 0$.  Consider any
    %fixed $0< \delta \leq 1$ and
    $\boldsymbol{x}\in\mathbb{R}^{n} \backslash\{{\bf 0}\}$. On an
    event $\mathcal{B}$ of probability approaching one, there
    exists a collection $\mathcal{M}$ of $M=\exp\left(n/30\right)$
    distinct vectors obeying the following properties:
  \begin{itemize}
  \item[(i)] $\boldsymbol{x}\in\mathcal{M}$; 
  \item[(ii)] for all
    $\boldsymbol{w}^{(l)},\boldsymbol{w}^{(j)}\in\mathcal{M}$,
\begin{equation}
	1/\sqrt{8} - (2n)^{-1/2} \leq \big\|\boldsymbol{w}^{(l)}-\boldsymbol{w}^{(j)}\big\|
	\leq  3/2 + n^{-1/2}  ;
	\label{eq:sep-from-x}
\end{equation}
\item[(iii)] for all $\boldsymbol{w}\in\mathcal{M}$,
\begin{equation}
	\frac{ | \boldsymbol{a}_{i}^{\top}\left(\boldsymbol{w}-\boldsymbol{x} \right) |^{2}}{ | \boldsymbol{a}_{i}^{\top}\boldsymbol{x} |^2}
	\leq  \frac{\Vert \boldsymbol{w} - \boldsymbol{x}\Vert ^{2}}  {\|\boldsymbol{x}\|^2 } \{ 2 + 17\log^3 m \},
	\quad1\leq i\leq m.
	\label{eq:ratio-bound}
\end{equation}
\end{itemize}
\end{lemma}

%\begin{proof}See Appendix \ref{sub:Proof-of-Lemma-Minimax-Construction}.\end{proof}

%\begin{remark}
%\end{remark}

In words, Lemma \ref{lemma:MinimaxConstruction} constructs a set
$\mathcal{M}$ of exponentially many vectors/hypotheses scattered
around $\boldsymbol{x}$ and yet well separated. From (ii) we see that
each pair of hypotheses in $\mathcal{M}$ is separated by a distance
roughly on the order of $1$, and all hypotheses reside within a
spherical ball centered at $\boldsymbol{x}$ of radius $3/2+o(1)
$. When $\|\boldsymbol{x}\| \geq \log^{1.5}m$, every hypothesis $\boldsymbol{w} \in \mathcal{M}$
satisfies $\|\boldsymbol{w}\| \approx \|\boldsymbol{x}\|\gg 1$.  In
addition, (iii) says that the quantities ${ |
  \boldsymbol{a}_{i}^{\top}\left(\boldsymbol{w}-\boldsymbol{x} \right)
  |}/{ | \boldsymbol{a}_{i}^{\top}\boldsymbol{x} |}$ are all very well
controlled (modulo some logarithmic factor).  In particular, when
 $\|\boldsymbol{x}\| \geq \log^{1.5}m$, one must
have
\begin{equation}
  \frac{ | \boldsymbol{a}_{i}^{\top}\left(\boldsymbol{w}-\boldsymbol{x} \right) |^{2}}{ | \boldsymbol{a}_{i}^{\top}\boldsymbol{x} |^2}
  \lesssim  \frac{\Vert \boldsymbol{w} - \boldsymbol{x}\Vert ^{2}}  {\|\boldsymbol{x}\|^2 }  \log^3 m 
  \lesssim \frac{ 1 }  { \log^{3}m }  \log^3 m  \lesssim 1. 
	\label{eq:ratio-bound-large-x}
\end{equation}
 In the Poisson
model, such a quantity turns out to be crucial in controlling the
information divergence between two hypotheses, as demonstrated in the
following lemma.

\begin{lemma}
\label{lemma-KL-UB}
%Suppose that $y_{i}\mid\boldsymbol{x}=\boldsymbol{w}\sim\mathsf{Poisson}(\left|\boldsymbol{a}_{i}^{\top}\boldsymbol{w}\right|^{2})$
%($i=1,\cdots,m$) are independently generated, Then 
Fix a family of design vectors $\{\boldsymbol{a}_i\}$. Then for any
$\boldsymbol{w}$ and $\boldsymbol{r}\in\mathbb{R}^{n}$,
\begin{align}
	\mathsf{KL}\big( ~\mathbb{P}\left(\boldsymbol{y} \mid \boldsymbol{w}+\boldsymbol{r}\right) \hspace{0.3em}\|\hspace{0.3em}
			  \mathbb{P}\left(\boldsymbol{y} \mid \boldsymbol{w}\right) \big) 
	\text{ } \leq \text{ } 
	\sum\nolimits_{i=1}^{m} | \boldsymbol{a}_{i}^{\top}\boldsymbol{r} |^{2} 
		\bigg(8+\frac{ 2 |\boldsymbol{a}_{i}^{\top}\boldsymbol{r} |^{2}} { |\boldsymbol{a}_{i}^{\top}\boldsymbol{w} |^{2}}\bigg).
	\label{eq:KL-UB-Poisson}
\end{align}
\end{lemma}

%\begin{proof}See Appendix \ref{sub:Proof-of-Lemma-KL-UB}.\end{proof}

Lemma \ref{lemma-KL-UB} and (\ref{eq:ratio-bound-large-x}) taken
collectively suggest that on the event $\mathcal{B}\cap\mathcal{C}$ ($\mathcal{B}$ is in Lemma
\ref{lemma:MinimaxConstruction} and $\mathcal{C}:=\{\|\boldsymbol{A}\|\leq \sqrt{2m}\}$), the conditional KL divergence (we
condition on the $\boldsymbol{a}_i$'s) obeys 
%then with exponentially high probability, 
%
\begin{align}
	\mathsf{KL}\big(~\mathbb{P}\left(\boldsymbol{y}\mid\boldsymbol{w}\right)\hspace{0.3em}\|\hspace{0.3em}\mathbb{P}\left(\boldsymbol{y}\mid\boldsymbol{x}\right)\big) 
	& \leq c_3 \sum\nolimits_{i=1}^{m} \left|\boldsymbol{a}_{i}^{\top}\left(\boldsymbol{w}-\boldsymbol{x}\right)\right|^{2}
	\leq 2c_3 m \left\Vert \boldsymbol{w}-\boldsymbol{x}\right\Vert ^{2},
	\quad \forall \boldsymbol{w}\in\mathcal{M}; 
	\label{eq:KL-ub}
\end{align}
here, the inequality holds for some constant $c_3 > 0$ provided that
$\left\Vert \boldsymbol{x}\right\Vert \geq\log^{1.5}m$, and the last
inequality is a result of $\mathcal{C}$ (which occurs with high
probability).  We now use hypotheses as in Lemma \ref{lemma:MinimaxConstruction} but rescaled in such a way that 
\begin{equation}
  \|\boldsymbol{w}-\boldsymbol{x}\| \asymp \delta, \quad\text{and}\quad \|\boldsymbol{w}- \tilde{\boldsymbol{w}}\|\asymp \delta, \quad \forall \boldsymbol{w}, \tilde{\boldsymbol{w}}\in \mathcal{M} \text{ with }\boldsymbol{w}\neq \tilde{\boldsymbol{w}}.
\end{equation}
for some $ 0< \delta < 1$. This is achieved via the substitution
$\boldsymbol{w} \longleftarrow \boldsymbol{x} + \delta (\boldsymbol{w}
- \boldsymbol{x})$;
with a slight abuse of notation, $\mathcal{M}$ denotes the new set.

The hardness of a minimax estimation problem is known to be dictated
by information divergence inequalities such as \eqref{eq:KL-ub}.
Indeed, suppose that
\begin{equation}
	\frac{1}{M-1}\sum\nolimits_{\boldsymbol{w}\in\mathcal{M}\backslash\left\{ \boldsymbol{x}\right\} }
		\mathsf{KL}\big( ~\mathbb{P}\left(\boldsymbol{y}\mid\boldsymbol{w}\right)\hspace{0.3em}\|\hspace{0.3em}\mathbb{P}\left(\boldsymbol{y}\mid\boldsymbol{x}\right) \big)
	\leq \frac{1}{10}\log\left(M-1\right)
	\label{eq:KL-condition}
\end{equation}
holds, then the Fano-type minimax lower bound \cite[Theorem 2.7]{tsybakov2009introduction} asserts that
\begin{equation}
	%\inf_{\hat{\boldsymbol{x}}}\sup_{\boldsymbol{x}\in\mathcal{M}}\mathbb{E}\left[\mathrm{dist}\left(\hat{\boldsymbol{x}},\boldsymbol{x}\right) ~\big|~ \boldsymbol{a}_i, 1\leq i\leq m\right]
	%= 
  \inf_{\hat{\boldsymbol{x}}}\sup_{\boldsymbol{x}\in\mathcal{M}}\mathbb{E}\left[\Vert \hat{\boldsymbol{x}}-\boldsymbol{x} \Vert ~\big|~ \{\boldsymbol{a}_i\} \right]
  \gtrsim \min_{\boldsymbol{w},\tilde{\boldsymbol{w}}\in\mathcal{M}, \boldsymbol{w}\neq \tilde{\boldsymbol{w}}} \Vert \boldsymbol{w} - \tilde{\boldsymbol{w}} \Vert .
	\label{eq:minimax_Fano_lb}
\end{equation}
Since $M = \exp( n/30)$, (\ref{eq:KL-condition}) would follow from 
\begin{equation}
  2c_3  \Vert \boldsymbol{w}-\boldsymbol{x} \Vert ^{2}\leq n/(300 m). 
  \quad \boldsymbol{w}\in\mathcal{M}. 
	\label{eq:w-x-dist}
\end{equation}
Hence, we just need to select $\delta$ to be a small multiple of $\sqrt{n/m}$. This in turn gives
\begin{equation}
\inf_{\hat{\boldsymbol{x}}}\sup_{\boldsymbol{x}\in\mathcal{M}}\mathbb{E}\left[\Vert
  \hat{\boldsymbol{x}}-\boldsymbol{x} \Vert ~\big|~ \{\boldsymbol{a}_i\}
  \right] \gtrsim
\min_{\boldsymbol{w},\tilde{\boldsymbol{w}}\in\mathcal{M}, \boldsymbol{w}\neq \tilde{\boldsymbol{w}}} \Vert
\boldsymbol{w} - \tilde{\boldsymbol{w}} \Vert \gtrsim \sqrt{n/m}. 
\end{equation}

Finally, it remains to connect $\Vert
\hat{\boldsymbol{x}}-\boldsymbol{x} \Vert$ with
$\mathrm{dist}\left(\hat{\boldsymbol{x}},\boldsymbol{x}\right)$. Since
all the $\boldsymbol{w} \in \mathcal{M}$ are clustered around
$\boldsymbol{x}$ and are at a mutual distance about $\delta$ that is
much smaller than $\|\boldsymbol{x}\|$, we can see that for any
reasonable estimator, $\mathrm{dist}(\hat{\boldsymbol{x}},
\boldsymbol{x}) = \|\hat{\boldsymbol{x}} - \boldsymbol{x}\|$. This
finishes the proof. 

%  Apparently, a minimax
% estimator $\hat{\boldsymbol{x}}$ when restricted to $\mathcal{M}$ must
% return an estimate that lies within or close to the set $\mathcal{M}$.
% But when $\|\hat{\boldsymbol{x}} - \boldsymbol{x}\|<
% \|\boldsymbol{x}\|$, one necessarily has $
% \mathrm{dist}(\hat{\boldsymbol{x}}, \boldsymbol{x}) =
% \|\hat{\boldsymbol{x}} - \boldsymbol{x}\| $, which implies that
% \[
% 	\inf_{\hat{\boldsymbol{x}}}\sup_{\boldsymbol{x}\in\mathcal{M}}\mathbb{E}\left[\mathrm{dist}\left(\hat{\boldsymbol{x}},\boldsymbol{x}\right) ~\big|~ \boldsymbol{a}_i, 1\leq i\leq m\right]
% 	=
% 	\inf_{\hat{\boldsymbol{x}}}\sup_{\boldsymbol{x}\in\mathcal{M}}\mathbb{E}\left[\Vert \hat{\boldsymbol{x}}-\boldsymbol{x} \Vert ~\big|~ \boldsymbol{a}_i, 1\leq i\leq m \right]
% \] and establishes Lemma \ref{lemma:MinimaxConstruction}.

%\yxc{The below statement might not be rigorous}
% Since $\mathbb{P}(\mathcal{B})$ is bounded away from
% 0, we immediately arrive at the minimax bound asserted in Theorem
% \ref{theorem-converse}.

\section{Discussion}
\label{sec:Discussion}

To keep our treatment concise, this paper does not strive to explore
all possible generalizations of the theory. There are nevertheless a
few extensions worth pointing out.

\begin{itemize}
\item {\bf More general objective functions}. For concreteness, we
  restrict our analysis to the Poisson log-likelihood function, but
  the analysis framework we laid out easily carries over to a broad
  class of (nonconvex) objective functions. For instance, all results
  continue to hold if we replace the Poisson log-likelihood by the
  Gaussian log-likelihood; that is, the polynomial function
  $-\sum_{i=1}^{m}(y_{i}-|\boldsymbol{a}_{i}^{\top}\boldsymbol{z}|^{2}
  )^{2}$ studied in \cite{candes2014wirtinger}. A general guideline is
  to first check whether the expected regularity condition
\begin{eqnarray*}
	\begin{array}{l}
	\mathbb{E}\big[-\big\langle \frac{1}{m}\nabla\ell_{\mathrm{tr}}\left(\boldsymbol{z}\right),\boldsymbol{h} \big\rangle \big]
	\text{ } \gtrsim \text{ } \left\Vert \boldsymbol{h}\right\Vert ^{2}
	\end{array}
\end{eqnarray*}
holds for any fixed $\boldsymbol{z}$ within a neighborhood around
$\boldsymbol{x}$. If so, then often times $\mathrm{RC}$ holds
uniformly within this neighborhood due to sharp concentration of
measure ensured by the regularization procedure.

%\item {\bf Universal recovery}. We concentrated on the situation where
%  the planted solution is fixed and independent of the design
%  vectors. This can be extended without difficulty to a more universal
%  theory: once the design vectors are selected and fixed, then TWF
%  should be able to succeed simultaneously for all signals
%  $\boldsymbol{x} \in \mathbb{R}^n$.  In fact, our theory for the
%  iterative refinement stage already delivers universal guarantees, so
%  the only remaining step is to justify the universality of the
%  truncated spectral method.
	
\item {\bf Sub-Gaussian measurements}. The theory extends  to the situation where the $\boldsymbol{a}_i$'s are
  i.i.d.~sub-Gaussian random vectors, although the truncation
  threshold might need to be tweaked based on the sub-Gaussian norm of
  $\boldsymbol{a}_i$. A more challenging scenario, however, is the
  case where the $\boldsymbol{a}_i$'s are generated according to the
  CDP model, since there is much less randomness to exploit in the
  mathematical analysis. We leave this to future research. 
\end{itemize}

Having demonstrated the power of TWF in recovering a rank-one matrix
$\boldsymbol{x}\boldsymbol{x}^*$ from quadratic equations, we remark
on the potential of TWF towards recovering low-rank matrices from
rank-one measurements. Imagine that we wish to estimate a rank-$r$
matrix $\boldsymbol{X}\succeq {\bf 0}$ and that all we know about
$\boldsymbol{X}$ is
\[ 
	y_{i}=\boldsymbol{a}_{i}^{\top}\boldsymbol{X}\boldsymbol{a}_{i},\quad1\leq i\leq m.	
\]
It is known that this problem can be efficiently solved by using more
computational-intensive semidefinite programs \cite{chen2013exact,
  cai2015rop}. 
With the hope of
developing a linear-time algorithm, one might consider a modified TWF
scheme, which would maintain a rank-$r$ matrix variable and operate as
follows: perform truncated spectral initialization, and then successively update
the current guess via a regularized gradient descent rule applied to a
presumed log-likelihood function. 

Moving away from i.i.d.~sub-Gaussian measurements, there is a
proliferation of problems that involve completion of a low-rank matrix
$\boldsymbol{X}$ from partial entries, where the rank is known {\em a
  priori}.  It is self-evident that such entry-wise observations can
also be cast as rank-one measurements of $\boldsymbol{X}$. Therefore,
the preceding modified TWF may add to recent literature in applying
non-convex schemes for low-rank matrix completion \cite{KesMonSew2010,
  jain2013low, hardt2013provable, sun2014guaranteed}, robust
PCA \cite{netrapalli2014non}, or even a broader family of latent-variable models (e.g.~dictionary learning \cite{sun2015complete,sun2015nonconvex}, sparse coding \cite{arora2015simple}, and mixture problems \cite{yi2013alternating,balakrishnan2014statistical}).  A concrete application of this flavor
is a simple form of the fundamental alignment/matching problem
\cite{bandeira2014multireference,huang2013consistent,chen2014matching}. Imagine a
collection of $n$ instances, each representing an image of the same
physical object but with different shift $r_i\in \{0,\cdots, M-1\}$.
The goal is to align all these instances from observations on the
relative shift between pairs of them. Denoting by $\boldsymbol{X}_{i}$
the cyclic shift by an amount $r_i$ of $\boldsymbol{I}_M$, one sees
that the collection matrix $\boldsymbol{X}:=[\boldsymbol{X}_{i}^\top
\boldsymbol{X}_{j}]_{1\leq i,j\leq k}$ is a rank-$M$ matrix, and the
relative shift observations can be treated as rank-one measurements of
$\boldsymbol{X}$.  Running TWF over this problem instance might result
in a statistically and computationally efficient solution. This would
be of great practical interest.

%\appendix
\appendix

\section{Proofs for Section \ref{sec:Proof-of-Theorem-TruncatedWF}}

\subsection{Proof of Lemma \ref{Lemma:superset-D}\label{proof-of-Lemma:superset-D}}

%Before proceeding to the proof, we note that the hypothesis $\left\Vert \boldsymbol{h}\right\Vert \leq\delta\left\Vert \boldsymbol{x}\right\Vert $
%implies
%
%\begin{equation}
%	\left\Vert \boldsymbol{h}\right\Vert \leq\frac{\delta}{1-\delta}\left\Vert \boldsymbol{z}\right\Vert ,
%	\quad \left\Vert \boldsymbol{z}\right\Vert \leq\left(1+\delta\right)\left\Vert \boldsymbol{x}\right\Vert ,
%	\quad \text{ and }\quad
%	\left\Vert \boldsymbol{x}\right\Vert \leq\frac{1}{1-\delta}\left\Vert \boldsymbol{z}\right\Vert ,
%	\label{eq:bound-h-z}
%\end{equation}
%
%which we shall use several times in the proof.

First, we make the observation that 
$(\boldsymbol{a}_{i}^{\top}\boldsymbol{z} )^2 - ( \boldsymbol{a}_{i}^{\top}\boldsymbol{x} )^{2}
= \left(2\boldsymbol{a}_{i}^{\top}\boldsymbol{z}-\boldsymbol{a}_{i}^{\top}\boldsymbol{h}\right)\boldsymbol{a}_{i}^{\top}\boldsymbol{h}$
is a quadratic function in $\boldsymbol{a}_{i}^{\top}\boldsymbol{h}$.
If we assume $\gamma \leq\frac{ \alpha_{z}^{\text{lb}} \Vert \boldsymbol{z}\Vert }{\Vert \boldsymbol{h} \Vert }$,
then on the event $\mathcal{E}_{1}^{i}$ one has 
\begin{equation}
	(\boldsymbol{a}_{i}^{\top}\boldsymbol{z})^2
	\text{ }\geq\text{ } \alpha_{z}^{\text{lb}} \| \boldsymbol{z} \| \cdot | \boldsymbol{a}_{i}^{\top}\boldsymbol{z} | 
	\text{ }\geq\text{ } \gamma \left\Vert \boldsymbol{h} \right\Vert   \left|\boldsymbol{a}_{i}^{\top}\boldsymbol{z}\right|.
	\label{eq:hypothesis-beta}
\end{equation}
Solving the quadratic inequality that specifies $\mathcal{D}_{\gamma}^{i}$ gives
\begin{eqnarray*}
	\boldsymbol{a}_{i}^{\top}\boldsymbol{h} 
	& \in & 
	\left[\boldsymbol{a}_{i}^{\top}\boldsymbol{z}-\sqrt{\left(\boldsymbol{a}_{i}^{\top}\boldsymbol{z}\right)^{2} + \gamma\left\Vert \boldsymbol{h}\right\Vert \left|\boldsymbol{a}_{i}^{\top}\boldsymbol{z}\right|},
	\text{ }\text{ }
	\boldsymbol{a}_{i}^{\top}\boldsymbol{z}-\sqrt{\left(\boldsymbol{a}_{i}^{\top}\boldsymbol{z}\right)^{2} - \gamma\left\Vert \boldsymbol{h}\right\Vert \left|\boldsymbol{a}_{i}^{\top}\boldsymbol{z}\right|}\right],  \\
	\text{or}\quad
	\boldsymbol{a}_{i}^{\top}\boldsymbol{h} 
	& \in & \left[\boldsymbol{a}_{i}^{\top}\boldsymbol{z}+\sqrt{\left(\boldsymbol{a}_{i}^{\top}\boldsymbol{z}\right)^{2} - \gamma\left\Vert \boldsymbol{h}\right\Vert \left|\boldsymbol{a}_{i}^{\top}\boldsymbol{z}\right|},
	\text{ }\text{ }
	\boldsymbol{a}_{i}^{\top}\boldsymbol{z}+\sqrt{\left(\boldsymbol{a}_{i}^{\top}\boldsymbol{z}\right)^{2} + \gamma\left\Vert \boldsymbol{h}\right\Vert \left|\boldsymbol{a}_{i}^{\top}\boldsymbol{z}\right|}\right],
\end{eqnarray*}
which we will simplify in the sequel. 

Suppose for the moment that $\boldsymbol{a}_{i}^{\top}\boldsymbol{z}\geq0$,
then the preceding two intervals are respectively equivalent to
\begin{eqnarray*}
	\boldsymbol{a}_{i}^{\top}\boldsymbol{h} 
	& \in & \underset{:= I_1}{\underbrace{  \left[ \frac{ -\text{ } \gamma\left\Vert \boldsymbol{h}\right\Vert \left|\boldsymbol{a}_{i}^{\top}\boldsymbol{z}\right|}{\boldsymbol{a}_{i}^{\top}\boldsymbol{z} + \sqrt{\left(\boldsymbol{a}_{i}^{\top}\boldsymbol{z}\right)^{2}+\gamma\left\Vert \boldsymbol{h}\right\Vert \left|\boldsymbol{a}_{i}^{\top}\boldsymbol{z}\right|}},
	\text{ }\frac{\gamma \left\Vert \boldsymbol{h}\right\Vert \left|\boldsymbol{a}_{i}^{\top}\boldsymbol{z}\right|}{\boldsymbol{a}_{i}^{\top}\boldsymbol{z}+\sqrt{\left(\boldsymbol{a}_{i}^{\top}\boldsymbol{z}\right)^{2} 
		- \gamma \left\Vert \boldsymbol{h}\right\Vert \left|\boldsymbol{a}_{i}^{\top}\boldsymbol{z}\right|}}\right] }} ;  \\
	\boldsymbol{a}_{i}^{\top}\boldsymbol{h}-2\boldsymbol{a}_{i}^{\top}\boldsymbol{z} 
	& \in & 
	\underset{:= I_2}{\underbrace{ \left[ \frac{ -\text{ } \gamma \left\Vert \boldsymbol{h}\right\Vert \left|\boldsymbol{a}_{i}^{\top}\boldsymbol{z}\right|}{\boldsymbol{a}_{i}^{\top}\boldsymbol{z}+\sqrt{\left(\boldsymbol{a}_{i}^{\top}\boldsymbol{z}\right)^{2} 
		- \gamma\left\Vert \boldsymbol{h}\right\Vert \left|\boldsymbol{a}_{i}^{\top}\boldsymbol{z}\right|}},
	\text{ }\text{ }\frac{\gamma \left\Vert \boldsymbol{h}\right\Vert \left|\boldsymbol{a}_{i}^{\top}\boldsymbol{z}\right|}{\boldsymbol{a}_{i}^{\top}\boldsymbol{z}+\sqrt{\left(\boldsymbol{a}_{i}^{\top}\boldsymbol{z}\right)^{2} 
		+ \gamma\left\Vert \boldsymbol{h}\right\Vert \left|\boldsymbol{a}_{i}^{\top}\boldsymbol{z}\right|}}\right] }}.
\end{eqnarray*}
Assuming (\ref{eq:hypothesis-beta}) and making use of the observations
\begin{eqnarray*}
	\frac{\gamma \left\Vert \boldsymbol{h}\right\Vert \left|\boldsymbol{a}_{i}^{\top}\boldsymbol{z}\right|}{\boldsymbol{a}_{i}^{\top}\boldsymbol{z}+\sqrt{\left(\boldsymbol{a}_{i}^{\top}\boldsymbol{z}\right)^{2} - \gamma \left\Vert \boldsymbol{h}\right\Vert \left|\boldsymbol{a}_{i}^{\top}\boldsymbol{z}\right|}} 
	& \leq & \frac{\gamma \left\Vert \boldsymbol{h}\right\Vert \left|\boldsymbol{a}_{i}^{\top}\boldsymbol{z}\right|}{\boldsymbol{a}_{i}^{\top}\boldsymbol{z}}
	= \gamma \left\Vert \boldsymbol{h}\right\Vert \\
	\text{and} \quad\frac{\gamma \left\Vert \boldsymbol{h}\right\Vert \left|\boldsymbol{a}_{i}^{\top}\boldsymbol{z}\right|}{\boldsymbol{a}_{i}^{\top}\boldsymbol{z}+\sqrt{\left(\boldsymbol{a}_{i}^{\top}\boldsymbol{z}\right)^{2} 
		+ \gamma\left\Vert \boldsymbol{h}\right\Vert \left|\boldsymbol{a}_{i}^{\top}\boldsymbol{z}\right|}} 
	& \geq & \frac{\gamma\left\Vert \boldsymbol{h}\right\Vert \left|\boldsymbol{a}_{i}^{\top}\boldsymbol{z}\right|}{\left(1+\sqrt{2}\right)\left|\boldsymbol{a}_{i}^{\top}\boldsymbol{z}\right|}
	= \frac{\gamma}{1+\sqrt{2}}\left\Vert \boldsymbol{h}\right\Vert ,
\end{eqnarray*}
we obtain the inner and outer bounds
\begin{eqnarray*}
	\left[\pm \big( 1+\sqrt{2}\big)^{-1} \gamma \left\Vert \boldsymbol{h}\right\Vert \right] 
	& \subseteq & I_{1},I_{2}\text{ }
	\subseteq \big[\pm\gamma \left\Vert \boldsymbol{h} \right\Vert \big].
\end{eqnarray*}
Setting $\gamma_{1}:=\frac{\gamma}{1+\sqrt{2}}$ gives
\begin{eqnarray*}
	\left(\mathcal{D}_{\gamma_{1}}^{i,1}\cap\mathcal{E}_{i,1}\right)\cup\left(\mathcal{D}_{\gamma_{1}}^{i,2}\cap\mathcal{E}_{i,1}\right)\text{ }\subseteq\text{ }\mathcal{D}_{\gamma}\cap\mathcal{E}_{i,1} & \subseteq & \left(\mathcal{D}_{\gamma}^{i,1}\cap\mathcal{E}_{i,1}\right)\cup\left(\mathcal{D}_{\gamma}^{i,2}\cap\mathcal{E}_{i,1}\right).
\end{eqnarray*}
%
%with  and $\gamma_{2}:=\gamma$.

Proceeding with the same argument, we can derive exactly the same
inner and outer bounds in the regime where $\boldsymbol{a}_{i}^{\top}\boldsymbol{z}<0$,
concluding the proof.%

\subsection{Proof of Lemma \ref{Lemma:I1_I2}\label{sub:Proof-of-Lemma-I1_I2}}

%We start by observing that $\mathcal{E}_{1}^{i}$ and $\mathcal{D}_{\gamma,1}^{i}$
%depend only on $\boldsymbol{z} / \left\Vert \boldsymbol{z}\right\Vert $ and $\boldsymbol{h} / \left\Vert \boldsymbol{h}\right\Vert$
%regardless of their relative size
%$\Vert \boldsymbol{h} \Vert / \Vert \boldsymbol{z} \Vert$,
%and that $( \boldsymbol{a}_{i}^{\top}\boldsymbol{h} )^2 {\bf 1}_{\mathcal{E}_{1}^{i}\cap\mathcal{D}_{\gamma,1}^{i}}$
%is a homogeneous function of $\boldsymbol{h}$ of degree 2. Therefore,
By homogeneity, it suffices to establish the claim for the case where both $\boldsymbol{h}$ and
$\boldsymbol{z}$ are \emph{unit vectors}. 

Suppose for the moment that $\boldsymbol{h}$ and $\boldsymbol{z}$ are
\emph{statistically independent} from $\{\boldsymbol{a}_{i}\}$.  We
introduce two auxiliary Lipschitz functions approximating indicator
functions:
\begin{align}
	&\chi_{z}\left(\tau\right) :=  \begin{cases}
		1,\quad & \text{if }\left|\tau\right|\in\left[\sqrt{1.01}\alpha_{z}^{\text{lb}},\sqrt{0.99}\alpha_{z}^{\text{ub}}\right];\\
		-100 \left(\alpha_{z}^{\text{ub}}\right)^{-2}\tau^{2}+100,\quad & \text{if }\left|\tau\right|\in\left[\sqrt{0.99}\alpha_{z}^{\text{ub}},\alpha_{z}^{\text{ub}}\right];\\
		100 \left(\alpha_{z}^{\text{lb}}\right)^{-2}\tau^{2}-100,\quad & \text{if }\left|\tau\right|\in\left[\alpha_{z}^{\text{lb}},\sqrt{1.01}\alpha_{z}^{\text{lb}}\right];\\
		0,\quad & \text{else}.
	\end{cases} \\
	&\chi_{h}\left(\tau\right) :=  \begin{cases}
		1, \quad & \text{if }\left|\tau\right|\in\left[0,\sqrt{0.99}\gamma\right];\\
		-\frac{100}{\gamma^{2}}\tau^{2}+100,\quad & \text{if }\left|\tau\right|\in\left[\sqrt{0.99}\gamma,\gamma\right];\\
		0, \quad & \text{else}.
	\end{cases}
\end{align}
Since $\boldsymbol{h}$ and $\boldsymbol{z}$ are assumed to be unit
vectors, these two functions obey 
\begin{equation}
	0\leq\chi_{z}\left(\boldsymbol{a}_{i}^{\top}\boldsymbol{z}\right)\leq{\bf 1}_{\mathcal{E}_{1}^{i}},\quad\text{and}\quad0\leq\chi_{h}\left(\boldsymbol{a}_{i}^{\top}\boldsymbol{h}\right)\leq{\bf 1}_{\mathcal{D}_{\gamma}^{i,1}}
\end{equation}
and thus, 
\begin{eqnarray}
	\frac{1}{m}\sum_{i=1}^{m}\left(\boldsymbol{a}_{i}^{\top}\boldsymbol{h}\right)^{2}{\bf 1}_{\mathcal{E}_{1}^{i}\cap\mathcal{D}_{\gamma}^{i,1}} 
	& \geq & \frac{1}{m}\sum_{i=1}^{m}  (\boldsymbol{a}_{i}^{\top}\boldsymbol{h})^{2}\chi_{z} (\boldsymbol{a}_{i}^{\top}\boldsymbol{z})  \chi_{h} (\boldsymbol{a}_{i}^{\top}\boldsymbol{h} ).
	\label{eq:LB-chi}
\end{eqnarray}
We proceed to lower bound
$\frac{1}{m}\sum_{i=1}^{m}\left(\boldsymbol{a}_{i}^{\top}\boldsymbol{h}\right)^{2}\chi_{z}\left(\boldsymbol{a}_{i}^{\top}\boldsymbol{z}\right)\chi_{h}\left(\boldsymbol{a}_{i}^{\top}\boldsymbol{h}\right)$.

Firstly, to compute the mean of 
$(\boldsymbol{a}_{i}^{\top}\boldsymbol{h})^{2}\chi_{z} (\boldsymbol{a}_{i}^{\top}\boldsymbol{z} )\chi_{h} (\boldsymbol{a}_{i}^{\top}\boldsymbol{h} )$,
we introduce an auxiliary orthonormal matrix 
\begin{equation}
	\boldsymbol{U}_{\boldsymbol{z}}=\left[\begin{array}{c}
\boldsymbol{z}^{\top}/\left\Vert \boldsymbol{z}\right\Vert \\
	\vdots
	\end{array}\right]
\end{equation}
whose first row is along the direction of $\boldsymbol{z}$, and set
\begin{equation}
	\tilde{\boldsymbol{h}} := \boldsymbol{U}_{\boldsymbol{z}}\boldsymbol{h},
	\quad\text{and}\quad
	\tilde{\boldsymbol{a}}_{i} := \boldsymbol{U}_{\boldsymbol{z}}\boldsymbol{a}_{i}.	
	\label{eq:defn_a_tilde-1}
\end{equation}
Also, denote by $\tilde{a}_{i,1}$ (resp. $\tilde{h}_{1}$) the first
entry of $\tilde{\boldsymbol{a}}_{i}$ (resp. $\tilde{\boldsymbol{h}}$),
and $\boldsymbol{\tilde{a}}_{i,\backslash1}$ (resp. $\tilde{\boldsymbol{h}}_{\backslash1}$)
the remaining entries of $\tilde{\boldsymbol{a}}_{i}$ (resp. $\tilde{\boldsymbol{h}}$),
and let $\xi\sim\mathcal{N}\left(0,1\right)$. We have 
\begin{align}
	 &  \mathbb{E}\Big[ \left(\boldsymbol{a}_{i}^{\top}\boldsymbol{h}\right)^{2}\chi_{z}\left(\boldsymbol{a}_{i}^{\top}\boldsymbol{z}\right)\chi_{h}\left(\boldsymbol{a}_{i}^{\top}\boldsymbol{h}\right)\Big] \nonumber\\
	 &   \geq\text{ } \mathbb{E}\big[ (\boldsymbol{a}_{i}^{\top}\boldsymbol{h} )^{2}\chi_{z}\left(\boldsymbol{a}_{i}^{\top}\boldsymbol{z}\right) \big] 
		- \mathbb{E}\left[\left(\boldsymbol{a}_{i}^{\top}\boldsymbol{h}\right)^{2}\left(1-\chi_{h}\left(\boldsymbol{a}_{i}^{\top}\boldsymbol{h}\right)\right)\right]\text{ } \nonumber\\
	 &   \geq\text{ }\mathbb{E}\left[ \big(\tilde{a}_{i,1}\tilde{h}_{1} \big)^{2}\chi_{z}\left(\boldsymbol{a}_{i}^{\top}\boldsymbol{z}\right)\right] 
	   + \mathbb{E}\left[ \big(\tilde{\boldsymbol{a}}_{i,\backslash1}^{\top}\tilde{\boldsymbol{h}}_{\backslash1} \big)^{2}\right] \mathbb{E}\left[ \chi_z \left(\boldsymbol{a}_i^{\top} \boldsymbol{z} \right) \right]
	   - \left\Vert \boldsymbol{h}\right\Vert^2 \mathbb{E}\left[\xi^2 {\bf 1}_{\left\{ \left|\xi\right|>\sqrt{0.99}\gamma\right\} }\right] \nonumber\\
	 &    \geq \text{ } | \tilde{h}_{1} |^{2} (1-\zeta_{1})+ \Vert \tilde{\boldsymbol{h}}_{\backslash1}\Vert ^2  (1-\zeta_{1})
	   - \zeta_2 \Vert \boldsymbol{h} \Vert ^2
		\label{eq:second_event-1}\\
	 &   \geq\text{ }\left(1-\zeta_1 - \zeta_2 \right)\left\Vert \boldsymbol{h}\right\Vert ^2, \nonumber
		%\label{eq:first-term-expectation-1}
\end{align}
where the identity (\ref{eq:second_event-1}) arises from 
(\ref{eq:Assumption-zeta1}) and (\ref{eq:Assumption-zeta2}). Since 
$\left(\boldsymbol{a}_{i}^{\top}\boldsymbol{h}\right)^{2}\chi_{z}\left(\boldsymbol{a}_{i}^{\top}\boldsymbol{z}\right)\chi_{h}\left(\boldsymbol{a}_{i}^{\top}\boldsymbol{h}\right)$
is bounded in magnitude by 
$\gamma^{2}\left\Vert \boldsymbol{h}\right\Vert ^{2}$,
it is a sub-Gaussian random variable with sub-Gaussian norm $O (\gamma^{2}\left\Vert \boldsymbol{h}\right\Vert ^2 )$.
Apply the Hoeffding-type inequality \cite[Proposition 5.10]{Vershynin2012}
to deduce that for any $\epsilon>0$,
\begin{align}
	\frac{1}{m}\sum_{i=1}^{m}\left(\boldsymbol{a}_{i}^{\top}\boldsymbol{h}\right)^{2}\chi_{z}\left(\boldsymbol{a}_{i}^{\top}\boldsymbol{z}\right)\chi_{h}\left(\boldsymbol{a}_{i}^{\top}\boldsymbol{h}\right) 
	& \geq  \mathbb{E}\left[\left(\boldsymbol{a}_{i}^{\top}\boldsymbol{h}\right)^{2}\chi_{z}\left(\boldsymbol{a}_{i}^{\top}\boldsymbol{z}\right)\chi_{h}\left(\boldsymbol{a}_{i}^{\top}\boldsymbol{h}\right)\right]-\epsilon\left\Vert \boldsymbol{h}\right\Vert ^{2}\\
 	& \geq  \left(1-\zeta_{1}-\zeta_{2}-\epsilon\right)\left\Vert \boldsymbol{h}\right\Vert ^{2}\label{eq:LB-h}
\end{align}
with probability at least $1-\exp (-\Omega (\epsilon^{2}m ) )$. 

The next step is to obtain uniform control over all \emph{unit vectors},
for which we adopt a basic version of an $\epsilon$-net argument.
Specifically, we construct an $\epsilon$-net $\mathcal{N}_{\epsilon}$
with cardinality 
$\left|\mathcal{N}_{\epsilon}\right|\leq\left(1+ 2/\epsilon \right)^{2n}$ (cf. \cite{Vershynin2012})
 such that for any $\left(\boldsymbol{h},\boldsymbol{z}\right)$
with $\left\Vert \boldsymbol{h}\right\Vert =\left\Vert \boldsymbol{z}\right\Vert =1$,
there exists a pair $\boldsymbol{h}_{0},\boldsymbol{z}_{0}\in\mathcal{N}_{\epsilon}$
satisfying $\left\Vert \boldsymbol{h}-\boldsymbol{h}_{0}\right\Vert \leq\epsilon$
and $\left\Vert \boldsymbol{z}-\boldsymbol{z}_{0}\right\Vert \leq\epsilon$.
%In particular,  tells us that $\mathcal{N}_{\epsilon}$
%can be selected with cardinality at most $\left|\mathcal{N}_{\epsilon}\right|\leq\left(1+\frac{2}{\epsilon}\right)^{n}.$
Now that we have discretized the unit spheres using a finite set, taking
the union bound gives 
\begin{equation}
	\frac{1}{m}\sum_{i=1}^{m}\left(\boldsymbol{a}_{i}^{\top}\boldsymbol{h}_{0}\right)^{2}\chi_{z}\left(\boldsymbol{a}_{i}^{\top}\boldsymbol{z}_{0}\right)\chi_{h}\left(\boldsymbol{a}_{i}^{\top}\boldsymbol{h}_{0}\right)
	\geq \left(1-\zeta_{1}-\zeta_{2}-\epsilon\right)\left\Vert \boldsymbol{h}_{0}\right\Vert ^{2}, \quad
	\forall\boldsymbol{h}_{0},\boldsymbol{z}_{0}\in\mathcal{N}_{\epsilon}
	\label{eq:Property-Epsilon-Net}
\end{equation}
 with probability at least $1 - (1+ 2/\epsilon )^{2n}\exp (-\Omega( \epsilon^{2}m ))$.

Define $f_1(\cdot)$ and $f_2(\cdot)$ such that $f_1(\tau) := \tau\chi_{h}(\sqrt{\tau})$ and $f_2(\tau) := \chi_{z}(\sqrt{\tau})$, which are both bounded functions with Lipschitz constant $O(1)$.
%
% bounded above by 101 and $100/\left(\alpha_{z}^{\mathrm{lb}}\right)^{2}$, respectively. 
%
% Note that $\tau^{2}\chi_{h}\left(\tau\right)$ is a bounded Lipschitz
% function of $\tau^{2}$ with Lipschitz constant bounded above by $101$,
% and that $\chi_{z}\left(\tau\right)$ is a Lipschitz function of $\tau^{2}$
% with Lipschitz constant at most $100/\left(\alpha_{z}^{\mathrm{lb}}\right)^{2}$.
%
This guarantees that for each \emph{unit} vector pair $\boldsymbol{h}$
and $\boldsymbol{z}$, 
\begin{align*}
 	& \left|\left(\boldsymbol{a}_{i}^{\top}\boldsymbol{h}\right)^{2}\chi_{z}\left(\boldsymbol{a}_{i}^{\top}\boldsymbol{z}\right)\chi_{h}\left(\boldsymbol{a}_{i}^{\top}\boldsymbol{h}\right)-\left(\boldsymbol{a}_{i}^{\top}\boldsymbol{h}_{0}\right)^{2}\chi_{z}\left(\boldsymbol{a}_{i}^{\top}\boldsymbol{z}_{0}\right)\chi_{h}\left(\boldsymbol{a}_{i}^{\top}\boldsymbol{h}_{0}\right)\right|  \\
 	& \quad\leq\text{ }  |\chi_{h}\left(\boldsymbol{a}_{i}^{\top}\boldsymbol{z}\right)| \cdot |\left(\boldsymbol{a}_{i}^{\top}\boldsymbol{h}\right)^{2}\chi_{h}\left(\boldsymbol{a}_{i}^{\top}\boldsymbol{h}\right)
		- \left(\boldsymbol{a}_{i}^{\top}\boldsymbol{h}_{0}\right)^{2}\chi_{h}\left(\boldsymbol{a}_{i}^{\top}\boldsymbol{h}_{0}\right) |   \\
	& \qquad \qquad \qquad + |(\boldsymbol{a}_{i}^{\top}\boldsymbol{h}_{0})^{2}\chi_{h}\left(\boldsymbol{a}_{i}^{\top}\boldsymbol{h}_{0}\right)| \cdot |\chi_{h}\left(\boldsymbol{a}_{i}^{\top}\boldsymbol{z}\right)-\chi_{h}\left(\boldsymbol{a}_{i}^{\top}\boldsymbol{z}_{0}\right) |  \\
	& \quad\leq\text{ }  |\chi_{h}\left(\boldsymbol{a}_{i}^{\top}\boldsymbol{z}\right)| \cdot \big| f_1\big( |\boldsymbol{a}_i^{\top}\boldsymbol{h}|^2\big) - f_1 \big( |\boldsymbol{a}_i^{\top}\boldsymbol{h}_0|^2\big) \big|  \\
	& \qquad \qquad \qquad + \big|(\boldsymbol{a}_{i}^{\top}\boldsymbol{h}_{0})^{2}\chi_{h}\left(\boldsymbol{a}_{i}^{\top}\boldsymbol{h}_{0} \right) \big| \cdot \big|f_2\left(|\boldsymbol{a}_{i}^{\top}\boldsymbol{z}|^2\right)- f_2\left( |\boldsymbol{a}_{i}^{\top}\boldsymbol{z}_{0}|^2 \right) \big|  \\
 	& \quad\lesssim \text{ } \big| (\boldsymbol{a}_{i}^{\top}\boldsymbol{h})^{2} - (\boldsymbol{a}_{i}^{\top}\boldsymbol{h}_{0} )^{2}\big| + 
	%\gamma^{2}\frac{100}{\left(\alpha_{z}^{\mathrm{lb}}\right)^{2}}
	\big (\boldsymbol{a}_{i}^{\top}\boldsymbol{z} )^{2} - (\boldsymbol{a}_{i}^{\top}\boldsymbol{z}_{0})^{2} \big|.
\end{align*}
Consequently,  there exists some universal constant $c_{3}>0$ such that 
\begin{align*}
	& \Big|\frac{1}{m}\sum_{i=1}^{m}\left(\boldsymbol{a}_{i}^{\top}\boldsymbol{h}\right)^{2}\chi_{z}\left(\boldsymbol{a}_{i}^{\top}\boldsymbol{z}\right)\chi_{h}\left(\boldsymbol{a}_{i}^{\top}\boldsymbol{h}\right)-\frac{1}{m}\sum_{i=1}^{m}\left(\boldsymbol{a}_{i}^{\top}\boldsymbol{h}_{0}\right)^{2}\chi_{z}\left(\boldsymbol{a}_{i}^{\top}\boldsymbol{z}_{0}\right)\chi_{h}\left(\boldsymbol{a}_{i}^{\top}\boldsymbol{h}_{0}\right) \Big|\\
	& \quad\lesssim \text{ }\frac{1}{m}\left\Vert \mathcal{A}\big(\boldsymbol{h}\boldsymbol{h}^{\top}-\boldsymbol{h}_{0}\boldsymbol{h}_{0}^{\top}\big)\right\Vert _{1} +
	 % \frac{100\gamma^{2}}{\left(\alpha_{z}^{\mathrm{lb}}\right)^{2}}\cdot
	\frac{1}{m}\left\Vert \mathcal{A}\big(\boldsymbol{z}\boldsymbol{z}^{\top}-\boldsymbol{z}_{0}\boldsymbol{z}_{0}^{\top}\big)\right\Vert _{1}\\
	& \quad\overset{\text{(i)}}{\leq}\text{ }c_{3}\left\{ \big\|\boldsymbol{h}\boldsymbol{h}^{\top}-\boldsymbol{h}_{0}\boldsymbol{h}_{0}^{\top}\big\|_{\mathrm{F}}+\big\|\boldsymbol{z}\boldsymbol{z}^{\top}-\boldsymbol{z}_{0}\boldsymbol{z}_{0}^{\top}\big\|_{\mathrm{F}}\right\} \\
	& \quad\overset{\text{(ii)}}{\leq}\text{ }2.5c_{3} \left\{ \big\|\boldsymbol{h}-\boldsymbol{h}_{0}\big\|\cdot\big\|\boldsymbol{h}\big\|+\big\|\boldsymbol{z}-\boldsymbol{z}_{0}\big\|\cdot\big\|\boldsymbol{z}\big\|\right\} \leq 5 c_{3}\epsilon,
\end{align*}
where (i) results from Lemma \ref{lemma:RIP-Candes}, and (ii) arises
from Lemma \ref{lemma: bound-h} whenever $\epsilon < 1/2$. 

With the assertion (\ref{eq:Property-Epsilon-Net}) in place, we see
that with high probability, 
\begin{eqnarray*}
	\frac{1}{m}\sum_{i=1}^{m}\left(\boldsymbol{a}_{i}^{\top}\boldsymbol{h}\right)^{2}\chi_{z}\left(\boldsymbol{a}_{i}^{\top}\boldsymbol{z}\right)\chi_{h}\left(\boldsymbol{a}_{i}^{\top}\boldsymbol{h}\right) & \geq & \left(1-\zeta_{1}-\zeta_{2}-\left(5c_{3}+1\right)\epsilon\right)\left\Vert \boldsymbol{h}\right\Vert ^{2}
\end{eqnarray*}
for all unit vectors $\boldsymbol{h}$ and $\boldsymbol{z}$.
Since $\epsilon$ can be arbitrary, 
%we conclude that
%\begin{eqnarray*}
%\frac{1}{m}\sum_{i=1}^{m}\left(\boldsymbol{a}_{i}^{\top}\boldsymbol{h}\right)^{2}\chi_{z}\left(\boldsymbol{a}_{i}^{\top}\boldsymbol{z}\right)\chi_{h}\left(\boldsymbol{a}_{i}^{\top}\boldsymbol{h}\right) & \geq & \left(1-\zeta_{1}-\zeta_{2}-\epsilon\right)\left\Vert \boldsymbol{h}\right\Vert ^{2}
%\end{eqnarray*}
%uniformly over all $\boldsymbol{h}$ and $\boldsymbol{z}$. 
putting this and (\ref{eq:LB-chi}) together completes the proof.

\subsection{Proof of Lemma \ref{lemma:I2}\label{sub:Proof-of-Lemma-I2}}

The proof makes use of standard concentration of measure and covering
arguments, and it suffices to restrict our attention to \emph{unit
vectors} $\boldsymbol{h}$. We find it convenient to work with an
auxiliary function 
\[
	\chi_2\left(\tau\right) = \begin{cases}
		\left|\tau\right|^{\frac{3}{2}},\quad & 
			\text{if } |\tau| \leq \gamma^{2}, \\
		-\gamma\left(\left|\tau\right|-\gamma^{2}\right)+\gamma^{3},\quad & 
			\text{if }\gamma^{2} < |\tau|\leq2\gamma^{2}, \\
		0, & \text{else}.
	\end{cases}
\]
Apparently, $\chi_2\left(\tau\right)$
is a Lipschitz function of $\tau$ with Lipschitz norm $O \left(\gamma\right)$.
Recalling the definition of $\mathcal{D}_{\gamma}^{i,1}$, we see
that each summand is bounded above by
\[
	|\boldsymbol{a}_{i}^{\top}\boldsymbol{h} |^3 \text{ } {\bf 1}_{\mathcal{D}_{\gamma}^{i,1}}
	\leq \chi_2 \big( |\boldsymbol{a}_{i}^{\top}\boldsymbol{h}|^{2} \big).
\]
For each fixed $\boldsymbol{h}$ and $\epsilon>0$, applying the Bernstein
inequality \cite[Proposition 5.16]{Vershynin2012} gives
\begin{eqnarray*}
	\frac{1}{m}\sum_{i=1}^{m}\left|\boldsymbol{a}_{i}^{\top}\boldsymbol{h}\right|^{3}{\bf 1}_{\mathcal{D}_{\gamma}^{i,1}} 
	& \leq & \frac{1}{m}\sum_{i=1}^{m}\chi_2\left(\left|\boldsymbol{a}_{i}^{\top}\boldsymbol{h}\right|^{2}\right)
	\text{ }\leq\text{ }\mathbb{E}\left[\chi_2\left(\left|\boldsymbol{a}_{i}^{\top}\boldsymbol{h}\right|^{2}\right)\right]+\epsilon\\
 	& \leq & \mathbb{E}\big[\left|\boldsymbol{a}_{i}^{\top}\boldsymbol{h}\right|^{3} \big] + \epsilon
	\text{ }=\text{ } \sqrt{ 8 / \pi} + \epsilon
\end{eqnarray*}
with probability exceeding $1-\exp\left(-\Omega\left(\epsilon^{2}m\right)\right)$. 

From \cite[Lemma 5.2]{Vershynin2012}, there exists an $\epsilon$-net
$\mathcal{N}_{\epsilon}$ of the unit sphere with cardinality $\left|\mathcal{N}_{\epsilon}\right|\leq\left(1+\frac{2}{\epsilon}\right)^{n}$.
For each $\boldsymbol{h}$, suppose that $\left\Vert \boldsymbol{h}_{0}-\boldsymbol{h}\right\Vert \leq\epsilon$
for some $\boldsymbol{h}_{0}\in\mathcal{N}_{\epsilon}$. The Lipschitz
property of $\chi_2$ implies 
%there exists some constant $c_{5}>0$ such that 
%
\begin{align*}
	 \frac{1}{m}\sum_{i=1}^{m}\left\{ \chi_2 \left(\left|\boldsymbol{a}_{i}^{\top}\boldsymbol{h}\right|^{2}\right)-\chi_2 \left(\left|\boldsymbol{a}_{i}^{\top}\boldsymbol{h}_{0}\right|^{2}\right)\right\}   
	& \text{ }\lesssim \text{ } \frac{1}{m}\sum_{i=1}^{m}\left|\left|\boldsymbol{a}_{i}^{\top}\boldsymbol{h}\right|^{2}-\left|\boldsymbol{a}_{i}^{\top}\boldsymbol{h}_{0}\right|^{2}\right| \\
 	& \text{ }\overset{(\text{i})}{ \asymp }\text{ }  \left\Vert \boldsymbol{h}-\boldsymbol{h}_{0}\right\Vert \left\Vert \boldsymbol{h}\right\Vert 
 	\text{ }\asymp\text{ } \epsilon,
\end{align*}
where (i) arises by combining Lemmas \ref{lemma:RIP-Candes} and
\ref{lemma: bound-h}. This demonstrates that with high probability,
\[
	\frac{1}{m}\sum_{i=1}^{m}\left|\boldsymbol{a}_{i}^{\top}\boldsymbol{h}\right|^{3}{\bf 1}_{\mathcal{D}_{\gamma}^{i,1}}
	\leq \frac{1}{m}\sum_{i=1}^{m} \chi_2 \left( |\boldsymbol{a}_{i}^{\top}\boldsymbol{h}|^2\right)
	\leq \sqrt{ 8 / \pi} + O\left(\epsilon\right)
\]
for all unit vectors $\boldsymbol{h}$, as claimed.

\subsection{Proof of Lemma \ref{Lemma:I3}\label{sub:Proof-of-Lemma-I3}}

Without loss of generality, the proof focuses on the case where $\left\Vert \boldsymbol{h}\right\Vert =1$.
Fix an arbitrary small constant $\delta>0$. One can eliminate the difficulty
of handling the discontinuous indicator functions by working with
the following auxiliary function 
\begin{equation}
	\chi_{3}\left(\tau,\gamma\right):=\begin{cases}
		1,\quad & \text{if }\sqrt{\tau}\geq\psi_{\mathrm{lb}}\left(\gamma\right);\\
		\frac{100\tau}{\psi_{\mathrm{lb}}^{2}\left(\gamma\right)}-99,\quad & \text{if }\sqrt{\tau}\in\left[\sqrt{0.99}\psi_{\mathrm{lb}}\left(\gamma\right),\text{ }\psi_{\mathrm{lb}}\left(\gamma\right)\right];\\
		0, & \text{else}.
\end{cases}
	\label{eq:chi3}
\end{equation}
Here, $\psi_{\mathrm{lb}}\left(\cdot\right)$ is a piecewise constant function defined as
\[
	\psi_{\mathrm{lb}}\left(\gamma\right):=\left(1+\delta\right)^{\left\lfloor \frac{\log\gamma}{ \log( 1 + \delta ) } \right\rfloor },
\]
which clearly satisfy $\frac{\gamma}{ 1 + \delta }\leq\psi_{\mathrm{lb}}\left(\gamma\right) \leq \gamma$.
Such a function is useful for our purpose since for any $0<\delta\leq0.005$, 
\begin{equation}
	{\bf 1}_{\left\{ \left|\boldsymbol{a}_i^{\top}\boldsymbol{h}\right|\geq\gamma\right\} }
	\leq \chi_3 \big( \left|\boldsymbol{a}_i^{\top}\boldsymbol{h}\right|^2, \text{ \ensuremath{\gamma}} \big)
	\leq {\bf 1}_{\left\{ \left|\boldsymbol{a}_i^{\top}\boldsymbol{h}\right| \geq \sqrt{0.99}\psi_{\mathrm{lb}}\left(\gamma\right)\right\} }\leq{\bf 1}_{\left\{ \left| \boldsymbol{a}_i^{\top}\boldsymbol{h} \right|\geq 0.99\gamma \right\} }.
	\label{eq:chi3-chi4}
\end{equation}
%
%This implies that an upper bound w.r.t.  $\chi_3 \big( \left|\boldsymbol{a}_i^{\top}\boldsymbol{h}\right|^2, \text{ \ensuremath{\gamma}} \big)$

For any fixed unit vector $\boldsymbol{h}$, the above argument leads to an upper
tail estimate: for any $0<t\leq1$,
\begin{eqnarray}
	\mathbb{P}\left\{ \chi_3 \big(\left|\boldsymbol{a}_{i}^{\top}\boldsymbol{h}\right|^{2},\text{ \ensuremath{\gamma}} \big) \geq t\right\} \text{ } 
	& \leq & \mathbb{P}\left\{ {\bf 1}_{\left\{ \left|\boldsymbol{a}_{i}^{\top}\boldsymbol{h}\right| \geq 0.99\gamma\right\} }\geq t\right\} 
	\text{ }  = \text{ } \mathbb{P}\left\{ {\bf 1}_{\left\{ \left|\boldsymbol{a}_{i}^{\top}\boldsymbol{h}\right| \geq 0.99\gamma\right\} } = 1 \right\} \nonumber \\
\text{ } 
	& = & 2{\displaystyle \int}_{0.99\gamma}^{\infty}\phi\left(x\right)\mathrm{d}x\text{ }
	\leq \frac{2}{0.99\gamma} \phi \left(0.99\gamma\right),
	\label{eq:Mill}
\end{eqnarray}
where $\phi(x)$ is the density of a standard normal, and (\ref{eq:Mill}) follows from the tail bound $\int_{x}^{\infty} \phi(x)\mathrm{d}x\leq\frac{1}{x}\phi\left(x\right)$
for all $x>0$. 
This implies that when $\gamma \geq 2$, both 
$\chi_3 \big( | \boldsymbol{a}_i^{\top}\boldsymbol{h}|^2, \gamma \big)$
and ${\bf 1}_{\left\{ | \boldsymbol{a}_i^{\top}\boldsymbol{h}|\geq 0.99\gamma\right\} }$
are sub-exponential with sub-exponential norm  
$O (\gamma^{-2})$ (cf. \cite[Definition 5.13]{Vershynin2012}).
We apply the Bernstein-type inequality for the sum of sub-exponential random
variables \cite[Corollary 5.17]{Vershynin2012}, which indicates that
for any fixed $\boldsymbol{h}$ and $\gamma$ as well as any sufficiently
small $\epsilon\in(0,1)$, 
\begin{eqnarray*}
	\frac{1}{m}\sum_{i=1}^{m} \chi_3 \big( \left|\boldsymbol{a}_{i}^{\top}\boldsymbol{h}\right|^{2},\text{ \ensuremath{\gamma}} \big) 
	& \leq & \frac{1}{m}\sum_{i=1}^{m}{\bf 1}_{\left\{ \left|\boldsymbol{a}_{i}^{\top}\boldsymbol{h}\right|\geq0.99\gamma\right\} }\text{ }\leq\mathbb{E}\left[{\bf 1}_{\left\{ \left|\boldsymbol{a}_{i}^{\top}\boldsymbol{h}\right|\geq0.99\gamma\right\} }\right]
	+ \epsilon \frac{1}{ \gamma^2 } \\
 	& \leq & \frac{2}{0.99\gamma}\exp\left(-0.49\gamma^{2}\right) + \epsilon \frac{1}{ \gamma^2 }
\end{eqnarray*}
holds with probability exceeding $1- \exp\left(-\Omega(\epsilon^{2}m)\right)$.
%for some absolute constants $C,c>0$.

We now proceed to obtain uniform control over all $\boldsymbol{h}$ and
$2\leq \gamma\leq2^{n}$.  To begin with, we consider all
$2\leq \gamma \leq m$ and construct an $\epsilon$-net
$\mathcal{N}_{\epsilon}$ over the unit sphere such that: (i)
$\left|\mathcal{N}_{\epsilon}\right|\leq\left(1+\frac{2}{\epsilon}\right)^{n}$;
(ii) for any $\boldsymbol{h}$ with
$\left\Vert \boldsymbol{h}\right\Vert =1$, there exists a unit vector
$\boldsymbol{h}_{0}\in\mathcal{N}_{\epsilon}$ obeying
$\left\Vert \boldsymbol{h}-\boldsymbol{h}_{0}\right\Vert
\leq\epsilon$.
Taking the union bound gives the following: with probability at least
$1 - \frac{\log
  m}{\log\left(1+\delta\right)}\left(1+\frac{2}{\epsilon}\right)^n
\exp( -\Omega( \epsilon^{2}m ) )$,
\begin{eqnarray*}
	\frac{1}{m}\sum_{i=1}^{m} \chi_3 \big(\left|\boldsymbol{a}_{i}^{\top}\boldsymbol{h}_{0}\right|^{2},\text{ \ensuremath{\gamma}}_{0} \big) 
	& \leq & \text{ } (0.495\gamma_{0})^{-1} \exp \left(-0.49\gamma_{0}^{2}\right)+\epsilon\gamma_{0}^{-2}
\end{eqnarray*}
holds simultaneously for all $\boldsymbol{h}_{0}\in\mathcal{N}_{\epsilon}$
and $\gamma_{0}\in\left\{ \left(1+\delta\right)^{k}\mid1\leq k\leq\frac{\log m}{\log\left(1+\delta\right)}\right\} $. 

Note that $\chi_{3}\left(\tau,\gamma_{0}\right)$ is a Lipschitz function in $\tau$
with the Lipschitz constant bounded above by $\frac{100}{\psi_{\mathrm{lb}}^{2}\left(\gamma_{0}\right)}$.
With this in mind, for any $\left(\boldsymbol{h},\gamma\right)$ with
$\left\Vert \boldsymbol{h}\right\Vert =1$ and $\gamma_{0}:=\left(1+\delta\right)^{k}\leq\gamma<\left(1+\delta\right)^{k+1}$,
one has
\begin{eqnarray*}
	\left|\chi_{3}\big( \left|\boldsymbol{a}_{i}^{\top}\boldsymbol{h}_{0}\right|^{2},\text{ \ensuremath{\gamma}}_{0} \big)
		- \chi_{3} \big( \left|\boldsymbol{a}_{i}^{\top}\boldsymbol{h}\right|^{2},\text{ \ensuremath{\gamma}} \big) \right| 
	& = & \left|\chi_{3} \big( \left|\boldsymbol{a}_{i}^{\top}\boldsymbol{h}_{0}\right|^{2},\text{ \ensuremath{\gamma}}_{0} \big)
		- \chi_{3} \big( \left|\boldsymbol{a}_{i}^{\top}\boldsymbol{h}\right|^{2},\text{ \ensuremath{\gamma}}_{0} \big) \right|\\
 	& \leq & \text{ }\frac{100}{\psi_{\mathrm{lb}}^{2}\left(\gamma_{0}\right)}\left|\left|\boldsymbol{a}_{i}^{\top}\boldsymbol{h}\right|^{2}-\left|\boldsymbol{a}_{i}^{\top}\boldsymbol{h}_{0}\right|^2 \right|.
\end{eqnarray*}
It then follows from Lemmas \ref{lemma:RIP-Candes}-\ref{lemma: bound-h} that 
\begin{align*}
	& \frac{1}{m}\left|\sum_{i=1}^{m}\chi_{3}\left(\left|\boldsymbol{a}_{i}^{\top}\boldsymbol{h}_{0}\right|^{2},\text{ \ensuremath{\gamma}}_{0}\right)-\sum_{i=1}^{m}\chi_{3}\left(\left|\boldsymbol{a}_{i}^{\top}\boldsymbol{h}\right|^{2},\text{ \ensuremath{\gamma}}\right)\right| 
	 \text{ }\leq\text{ }  \frac{100}{\psi_{\mathrm{lb}}^{2}\left(\gamma_{0}\right)}\frac{1}{m}\left\Vert \mathcal{A}\left(\boldsymbol{h}\boldsymbol{h}^{\top}-\boldsymbol{h}_{0}\boldsymbol{h}_{0}^{\top}\right)\right\Vert _{1}\\
 	& \quad \leq  \frac{250\left(1+\delta\right)^{2}}{\gamma^{2}} \Vert \boldsymbol{h}-\boldsymbol{h}_{0}\Vert \Vert \boldsymbol{h} \Vert 
	\text{ }\leq\text{ } \frac{250 (1+\delta)^2 \epsilon}{\gamma^2}.
\end{align*}
Putting the above results together gives that for all $2\leq\gamma\leq\left(1+\delta\right)^{\frac{\log m}{\log\left(1+\delta\right)}}=m$,
\begin{eqnarray*}
	\frac{1}{m}\sum_{i=1}^{m}\chi_{3}\left(\left|\boldsymbol{a}_{i}^{\top}\boldsymbol{h}\right|^{2},\text{ \ensuremath{\gamma}}\right) & \leq & \frac{1}{m}\sum_{i=1}^{m}\chi_{3}\left(\left|\boldsymbol{a}_{i}^{\top}\boldsymbol{h}_{0}\right|^{2},\text{ \ensuremath{\gamma}}_{0}\right)
		+ \frac{250\left(1+\delta\right)^{2}}{\gamma^{2}}\epsilon\\
 	& \leq & \text{ }\frac{1}{0.495\gamma_{0}}\exp\left(-0.49\gamma_{0}^{2}\right)+251\left(1+\delta\right)^{2}\frac{\epsilon}{\gamma^{2}}\\
 	& \leq & \frac{1}{0.49\gamma}\exp\left(-0.485\gamma^{2}\right)+251\left(1+\delta\right)^{2}\frac{\epsilon}{\gamma^{2}}
\end{eqnarray*}
with probability exceeding $1-\frac{\log m}{\log\left(1+\delta\right)}\left(1+\frac{2}{\epsilon}\right)^{n}\exp\left(-c\epsilon^{2}m\right)$.
This establishes (\ref{eq:I3-indicator}) for all $2\leq\gamma\leq m$. 

It remains to deal with the case where $\gamma>m$. To this end, we
rely on the following observation:
\[
	\frac{1}{m}\sum_{i=1}^{m}{\bf 1}_{\left\{ \left|\boldsymbol{a}_{i}^{\top}\boldsymbol{h}\right|\geq m\right\} }
	\text{ }\leq\text{ } \frac{1}{m}\sum_{i=1}^{m}\frac{\left|\boldsymbol{a}_{i}^{\top}\boldsymbol{h}\right|^{2}}{m^{2}}
	\text{ }\overset{(\text{i})}{\leq}\text{ } \frac{1+\delta}{m^{2}}\left\Vert \boldsymbol{h}\right\Vert ^{2}
	\text{ }\ll\text{ } \frac{1}{m},\quad\forall\boldsymbol{h}\text{ with }\|\boldsymbol{h}\|=1,
\]
where (i) comes from \cite[Lemmas 3.1]{candes2012phaselift}.  This
basically tells us that with high probability, none of the indicator
variables can be equal to 1.  Consequently,
$\frac{1}{m}\sum_{i=1}^{m}{\bf 1}_{\left\{
    \left|\boldsymbol{a}_{i}^{\top}\boldsymbol{h}\right|\geq m\right\}
}=0$, which proves the claim.

\subsection{Proof of Lemma \ref{Lemma:norm-score}\label{sec:Proof-of-Lemma-norm-score}}

Fix $\delta>0$. Recalling the notation 
% $\boldsymbol{A}_{1:m}:=\left[\boldsymbol{a}_{1},\cdots,\boldsymbol{a}_{m}\right]$ and
%
%\[
$	v_{i}:=2\Big\{ 2\boldsymbol{a}_{i}^{\top}\boldsymbol{h}-\frac{\left|\boldsymbol{a}_{i}^{\top}\boldsymbol{h}\right|^{2}}{\boldsymbol{a}_{i}^{\top}\boldsymbol{z}} \Big\} {\bf 1}_{\mathcal{E}_{1}^{i}\cap\mathcal{E}_{2}^{i}},
$
%\]
%
we see from the expansion (\ref{eq:simplify}) that 
\begin{equation}
	\Big\Vert \frac{1}{m}\nabla_{\mathrm{tr}}\ell (\boldsymbol{z}) \Big\Vert 
	= \Big\Vert \frac{1}{m}\boldsymbol{A}^{\top} \boldsymbol{v}\Big\Vert \leq \frac{1}{m} \Vert \boldsymbol{A}\Vert \cdot \Vert \boldsymbol{v}\Vert 
	\leq \left(1+\delta\right)\frac{\left\Vert \boldsymbol{v}\right\Vert }{\sqrt{m}}
	\label{eq:grad-ub}
\end{equation}
as soon as $m\geq c_{1}n$ for some sufficiently large $c_{1}>0$.
Here, the norm estimate $\left\Vert \boldsymbol{A}\right\Vert \leq\sqrt{m}\left(1+\delta\right)$
arises from standard random matrix results \cite[Corollary 5.35]{Vershynin2012}. 

Everything then comes down to controlling $\Vert \boldsymbol{v}\Vert $.
To this end, making use of the inclusion (\ref{eq:Event-Containment})
yields
\begin{eqnarray*}
	\frac{1}{4m}\left\Vert \boldsymbol{v}\right\Vert ^{2} 
	& = & \frac{1}{m}\sum_{i=1}^{m} \bigg(2\boldsymbol{a}_{i}^{\top}\boldsymbol{h}-\frac{ |\boldsymbol{a}_{i}^{\top}\boldsymbol{h} |^{2}}{\boldsymbol{a}_{i}^{\top}\boldsymbol{z}} \bigg)^2 {\bf 1}_{\mathcal{E}_{1}^{i}\cap\mathcal{E}_{2}^{i}} \\
	& \leq & \frac{1}{m}\sum_{i=1}^{m} \bigg( 2\left|\boldsymbol{a}_{i}^{\top}\boldsymbol{h}\right|+\frac{ |\boldsymbol{a}_{i}^{\top}\boldsymbol{h} |^{2}}{ |\boldsymbol{a}_{i}^{\top}\boldsymbol{z} |} \bigg)^2  
		{\bf 1}_{\mathcal{E}_{1}^{i}\cap\left(\mathcal{D}_{\gamma_{4}}^{i,1}\cup\mathcal{D}_{\gamma_{4}}^{i,2}\right)}\\
 	& \leq & \frac{1}{m}\sum_{i=1}^{m} \left\{ 4 (\boldsymbol{a}_{i}^{\top}\boldsymbol{h} )^{2}
		+ \left(\frac{4 |\boldsymbol{a}_{i}^{\top}\boldsymbol{h}|^{3}}{ |\boldsymbol{a}_{i}^{\top}\boldsymbol{z}|}+\frac{|\boldsymbol{a}_{i}^{\top}\boldsymbol{h}|^{4}}{|\boldsymbol{a}_{i}^{\top}\boldsymbol{z}|^{2}}\right)
	  {\bf 1}_{\mathcal{E}_{1}^{i}\cap\left(\mathcal{D}_{\gamma_{4}}^{i,1}\cup\mathcal{D}_{\gamma_{4}}^{i,2}\right)}   \right\} \\
 	& = & \frac{1}{m}\sum_{i=1}^{m}\left\{ 4\left(\boldsymbol{a}_{i}^{\top}\boldsymbol{h}\right)^{2}
		+ \left(4+\frac{ |\boldsymbol{a}_{i}^{\top}\boldsymbol{h} |}{ |\boldsymbol{a}_{i}^{\top}\boldsymbol{z}|}\right) \frac{|\boldsymbol{a}_{i}^{\top}\boldsymbol{h}|^{3}}{|\boldsymbol{a}_{i}^{\top}\boldsymbol{z}|}\left({\bf 1}_{\mathcal{E}_{1}^{i}\cap\mathcal{D}_{\gamma_{4}}^{i,1}}+{\bf 1}_{\mathcal{E}_{1}^{i}\cap\mathcal{D}_{\gamma_{4}}^{i,2}}\right)\right\} .
\end{eqnarray*}
The first term is controlled by \cite[Lemma 3.1]{candes2012phaselift}
in such a way that with probability $1-\exp(-\Omega(m))$,
\[
	\frac{1}{m}\sum_{i=1}^{m}4\left(\boldsymbol{a}_{i}^{\top}\boldsymbol{h}\right)^{2}
	\leq 4\left(1+\delta\right)\left\Vert \boldsymbol{h}\right\Vert ^{2}.
\]
Turning to the remaining terms,
we see from the definition of $\mathcal{D}_{\gamma}^{i,1}$ and $\mathcal{D}_{\gamma}^{i,2}$
that 
\begin{eqnarray*}
	\frac{\left|\boldsymbol{a}_{i}^{\top}\boldsymbol{h}\right|}{\left|\boldsymbol{a}_{i}^{\top}\boldsymbol{z}\right|} 
	& \leq & \begin{cases}
\frac{\gamma\left\Vert \boldsymbol{h}\right\Vert }{\alpha_{z}^{\mathrm{lb}}\left\Vert \boldsymbol{z}\right\Vert },\quad & \text{on }\mathcal{E}_{1}^{i}\cap\mathcal{D}_{\gamma}^{i,1}\\
	2+\frac{\gamma\left\Vert \boldsymbol{h}\right\Vert }{\alpha_{z}^{\mathrm{lb}}\left\Vert \boldsymbol{z}\right\Vert },\quad & \text{on }\mathcal{E}_{1}^{i}\cap\mathcal{D}_{\gamma}^{i,2}
	\end{cases}\text{ }
	\leq\text{ }\begin{cases}
		1,\quad & \text{on }\mathcal{E}_{1}^{i}\cap\mathcal{D}_{\gamma}^{i,1}\\
		3,\quad & \text{on }\mathcal{E}_{1}^{i}\cap\mathcal{D}_{\gamma}^{i,2}
	\end{cases}
\end{eqnarray*}
as long as $\gamma\leq\frac{\alpha_{z}^{\mathrm{lb}}\left\Vert \boldsymbol{z}\right\Vert }{\left\Vert \boldsymbol{h}\right\Vert }$.
Consequently, one can bound
\begin{eqnarray*}
 	&  & \frac{1}{m}\sum_{i=1}^{m}\left(4+\frac{|\boldsymbol{a}_{i}^{\top}\boldsymbol{h}|}{|\boldsymbol{a}_{i}^{\top}\boldsymbol{z}|}\right)\frac{ |\boldsymbol{a}_{i}^{\top}\boldsymbol{h} |^{3}}{ |\boldsymbol{a}_{i}^{\top}\boldsymbol{z} |}\left({\bf 1}_{\mathcal{E}_{1}^{i}\cap\mathcal{D}_{\gamma}^{i,1}}+{\bf 1}_{\mathcal{E}_{1}^{i}\cap\mathcal{D}_{\gamma}^{i,2}}\right)\text{ }\\
 	&  & \quad\leq\text{ }\frac{5}{m}\sum_{i=1}^{m}\frac{ |\boldsymbol{a}_{i}^{\top}\boldsymbol{h} |^{3}}{ |\boldsymbol{a}_{i}^{\top}\boldsymbol{z}|}{\bf 1}_{\mathcal{E}_{1}^{i}\cap\mathcal{D}_{\gamma}^{i,1}}+\frac{7}{m}\sum_{i=1}^{m}\frac{ |\boldsymbol{a}_{i}^{\top}\boldsymbol{h}|^{3}}{ |\boldsymbol{a}_{i}^{\top}\boldsymbol{z} |}{\bf 1}_{\mathcal{E}_{1}^{i}\cap\mathcal{D}_{\gamma}^{i,2}}\\
 	&  & \quad\leq\text{ }\frac{5\left(1+\delta\right)\sqrt{8/\pi} \Vert \boldsymbol{h} \Vert ^{3}}{\alpha_{z}^{\mathrm{lb}}\left\Vert \boldsymbol{z}\right\Vert }+\frac{7}{100}\left(1+\delta\right)\left\Vert \boldsymbol{h}\right\Vert ^{2},
\end{eqnarray*}
where the last inequality follows from (\ref{eq:I2}) and (\ref{eq:I3}). 

Recall that $\gamma_{4}=3\alpha_{h}$. Taken together all these bounds
lead to the upper bound
\begin{eqnarray*}
	\frac{1}{4m} \Vert \boldsymbol{v} \Vert ^2 & \leq & \left(1+\delta\right)\left\{ 4+\frac{5\sqrt{8/\pi}\left\Vert \boldsymbol{h}\right\Vert }{\alpha_{z}^{\mathrm{lb}}\left\Vert \boldsymbol{z}\right\Vert }+\frac{7}{100}\right\} 
	\left\Vert \boldsymbol{h}\right\Vert ^{2} \\
	& \leq & \left(1+\delta\right)\left\{ 4+\frac{5\sqrt{8/\pi}}{3\alpha_{h}}+\frac{7}{100}\right\} \left\Vert \boldsymbol{h}\right\Vert ^{2}
\end{eqnarray*}
whenever
%
%\begin{equation}
$	\frac{\left\Vert \boldsymbol{h}\right\Vert }{\left\Vert \boldsymbol{z}\right\Vert }
	\leq \min\left\{ \frac{\alpha_{z}^{\mathrm{lb}}}{3\alpha_{h}},
		 \frac{\alpha_{z}^{\mathrm{lb}}}{6},
		 \frac{\sqrt{98/3}\left(\alpha_{z}^{\mathrm{lb}}\right)^{2}}{2\alpha_{z}^{\mathrm{ub}}+\alpha_{z}^{\mathrm{lb}}},
		 \frac{1}{11}\right\} .
$
%\end{equation}
%
Substituting this into (\ref{eq:grad-ub}) completes the proof.

\section{Proofs for Section \ref{sec:proof-of-Minimax-Lower-Boundsi}}

\subsection{Proof of Lemma \ref{lemma:MinimaxConstruction}\label{sub:Proof-of-Lemma-Minimax-Construction}}

Firstly, we collect a few results on the magnitudes of $\boldsymbol{a}_{i}^{\top}\boldsymbol{x}$
($1\leq i\leq m$) that will be useful in constructing the hypotheses.
Observe that 
%for any small constant $0<\zeta<1$, there exists some constant $c_{\zeta}>0$ such that 
for any given $\boldsymbol{x}$ and any sufficiently large $m$, 
\begin{align*}
	& \mathbb{P}\left\{ \min_{1\leq i\leq m}\left|\boldsymbol{a}_{i}^{\top}\boldsymbol{x}\right|\geq\frac{1}{ m\log m}\left\Vert \boldsymbol{x}\right\Vert \right\}  
	 =  \left(\mathbb{P}\left\{ |\boldsymbol{a}_{i}^{\top}\boldsymbol{x}| \geq \frac{1}{m\log m} \Vert \boldsymbol{x} \Vert 
	\right\} \right)^{m} \\
	& \quad \geq\left(1-\frac{2}{\sqrt{2\pi}}\frac{1}{m\log m}\right)^{m}
	~\geq~ 1- o(1).
\end{align*}
Besides, since $\mathbb{E}\left[{\bf 1}_{\left\{ \left|\boldsymbol{a}_{i}^{\top}\boldsymbol{x}\right|\leq\frac{\left\Vert \boldsymbol{x}\right\Vert }{5\log m}\right\} }\right]\leq\frac{1}{\sqrt{2\pi}}\frac{2}{5\log m}\leq\frac{1}{5\log m}$,
applying Hoeffding's inequality yields 
\begin{eqnarray*}
 	&  & \mathbb{P}\bigg\{ \sum\nolimits_{i=1}^{m}{\bf 1}_{\left\{ \left|\boldsymbol{a}_{i}^{\top}\boldsymbol{x}\right|\leq\frac{\left\Vert \boldsymbol{x}\right\Vert }{5\log m}\right\} }> \frac{m} {4\log m} \bigg\} \\
 	&  & \quad = \mathbb{P}\bigg\{ \frac{1}{m}\sum\nolimits_{i=1}^{m}\left({\bf 1}_{\left\{ \left|\boldsymbol{a}_{i}^{\top}\boldsymbol{x}\right|\leq\frac{\left\Vert \boldsymbol{x}\right\Vert }{5\log m}\right\} }
		- \mathbb{E}\left[{\bf 1}_{\left\{ \left|\boldsymbol{a}_{i}^{\top}\boldsymbol{x}\right|\leq\frac{\left\Vert \boldsymbol{x}\right\Vert }{5\log m}\right\} }\right]\right)>{\frac{1}{20\log m}}\bigg\} \\
	& & \quad \leq \exp\left( - \Omega\Big( \frac{m} { \log^2 m }  \Big) \right).
\end{eqnarray*}
To summarize,  with probability $1  - o(1)$, one has
\begin{eqnarray}
	\min\nolimits_{1\leq i\leq m}\left|\boldsymbol{a}_{i}^{\top}\boldsymbol{x}\right| & \geq & \frac{1} {m\log m} \Vert \boldsymbol{x} \Vert; 
	\label{eq:LB-Gaussian}\\
	\sum\nolimits_{i=1}^{m}{\bf 1}_{\left\{ \left|\boldsymbol{a}_{i}^{\top}\boldsymbol{x}\right|
		\leq \frac{\left\Vert \boldsymbol{x}\right\Vert }{\log m}\right\} } & \leq & \frac{m}{4\log m} := k.
	\label{eq:UB-nontrivial-Gaussian}
\end{eqnarray}

In the sequel, we will first produce a set $\mathcal{M}_{1}$ of exponentially many
vectors surrounding $\boldsymbol{x}$ in such a way that every pair
is separated by about the same distance, and then verify that a non-trivial
fraction of $\mathcal{M}_{1}$ obeys (\ref{eq:ratio-bound}). 
%The
%argument assumes $\delta=1$, 
%but it immediately
%extends to any $0<\delta<1$ via proper scaling. 
Without loss of generality,
we assume that $\boldsymbol{x}$ takes the form $\boldsymbol{x}=\left[b,0,\cdots,0\right]^{\top}$
for some $b>0$.

The construction of $\mathcal{M}_{1}$ follows a standard random packing
argument. Let $\boldsymbol{w}=\left[w_{1},\cdots,w_{n}\right]^{\top}$ be a random vector
with 
\[
	w_{i}=x_{i}+\frac{1}{\sqrt{2n}}z_{i},\quad1\leq i\leq n,
\]
where $z_{i} \overset{\text{ind.}}{\sim} \mathcal{N}\left(0,1\right)$.
The collection $\mathcal{M}_{1}$ is then obtained by generating $M_{1}=\exp\left(\frac{n}{20}\right)$
independent copies $\boldsymbol{w}^{(l)}$ ($1\leq l<M_{1}$) of $\boldsymbol{w}$.
For any $\boldsymbol{w}^{(l)},\boldsymbol{w}^{(j)}\in\mathcal{M}_{1}$,
the concentration inequality \cite[Corollary 5.35]{Vershynin2012}
gives
\begin{align}
	\begin{array}{ll}
	\mathbb{P}\left\{ 0.5\sqrt{n}-1\leq\sqrt{n}\big\Vert \boldsymbol{w}^{(l)}-\boldsymbol{w}^{(j)}\big\Vert \leq1.5\sqrt{n}+1\right\} 
	\geq 1 - 2\exp\left(-n/8\right);	\nonumber\\
	\text{ }
	\mathbb{P}\left\{ 0.5\sqrt{n}-1\leq\sqrt{2n}\big\Vert \boldsymbol{w}^{(l)}-\boldsymbol{x}\big\Vert \leq1.5\sqrt{n}+1\right\} 
	\text{ }\geq 1 - 2\exp\left(-n/8\right). 	\nonumber
	\end{array}
\end{align}
Taking the union bound over all ${M_{1} \choose 2}$ pairs we obtain
\begin{eqnarray}
	\begin{array}{ll}
	\qquad 0.5-n^{-1/2} \leq  \text{ }\left\Vert \boldsymbol{w}^{(l)}-\boldsymbol{w}^{(j)}\right\Vert &\leq 1.5+n^{-1/2},\quad\forall l\neq j
	\label{eq:w-w-sep}	\\
	1/{\sqrt{8}}- (2n)^{-1/2} \leq \text{ }\text{ } \left\Vert \boldsymbol{w}^{(l)}-\boldsymbol{x}\right\Vert &\leq \sqrt{9/8}+ (2n)^{-1/2}, \quad1\leq l\leq M_{1}
	\end{array}
	\label{eq:w-x-sep}
\end{eqnarray}
with probability exceeding $1-2M_{1}^{2}\exp\left(-\frac{n}{8}\right)\geq1-2\exp\left(-\frac{n}{40}\right)$.

The next step is to show that many vectors in $\mathcal{M}_{1}$ satisfy
(\ref{eq:ratio-bound}). For any given $\boldsymbol{w}$ with $\boldsymbol{r}:=\boldsymbol{w}-\boldsymbol{x}$,
by letting $\boldsymbol{a}_{i,\perp}:=\left[a_{i,2},\cdots,a_{i,n}\right]^{\top}$,
$r_{\|}:=r_{1}$, and $\boldsymbol{r}_{\perp}:=\left[r_{2},\cdots,r_{n}\right]^{\top}$,
we derive
\begin{align}
	%\begin{array}{l}
	\frac{ |\boldsymbol{a}_{i}^{\top}\boldsymbol{r} |^{2}}{ |\boldsymbol{a}_{i}^{\top}\boldsymbol{x} |^{2}}  
	\leq  \frac{2|a_{i,1}r_{\|}|^{2}+2|\boldsymbol{a}_{i,\perp}^{\top}\boldsymbol{r}_{\perp}|^{2}}{\left|a_{i,1}\right|^{2}\left\Vert \boldsymbol{x}\right\Vert ^{2}}
	\leq \frac{2|r_{\|}|^{2}}{\left\Vert \boldsymbol{x}\right\Vert ^{2}}+\frac{2|\boldsymbol{a}_{i,\perp}^{\top}\boldsymbol{r}_{\perp}|^{2}}{\left|a_{i,1}\right|^{2}\left\Vert \boldsymbol{x}\right\Vert ^{2}}
	\leq \frac{2\|\boldsymbol{r}\|^{2}}{\left\Vert \boldsymbol{x}\right\Vert ^{2}}+\frac{2|\boldsymbol{a}_{i,\perp}^{\top}\boldsymbol{r}_{\perp}|^{2}}{\left|a_{i,1}\right|^{2}\left\Vert \boldsymbol{x}\right\Vert ^{2}}.
	\label{eq:key-quantity}
	%\end{array}
\end{align}
It then boils down to developing an upper bound on
$\frac{\left|\boldsymbol{a}_{i,\perp}^{\top}\boldsymbol{r}_{\perp}\right|^{2}}{\left|a_{i,1}\right|^{2}}$.
This ratio is convenient to work with since the numerator and
denominator are stochastically independent. To simplify presentation,
we reorder $\{\boldsymbol{a}_{i}\}$ in a way that
\[
	(m\log m)^{-1}\left\Vert \boldsymbol{x}\right\Vert \leq \left|\boldsymbol{a}_{1}^{\top}\boldsymbol{x}\right|\leq\left|\boldsymbol{a}_{2}^{\top}\boldsymbol{x}\right|\leq\cdots\leq\left|\boldsymbol{a}_{m}^{\top}\boldsymbol{x}\right|;
\]
this will not affect our subsequent analysis concerning $\boldsymbol{a}_{i,\perp}^{\top}\boldsymbol{r}_{\perp}$
since it is independent of $\boldsymbol{a}_{i}^{\top}\boldsymbol{x}$.

To proceed, we let $\boldsymbol{r}_{\perp}^{(l)}$ consist of all
but the first entry of $\boldsymbol{w}^{(l)}-\boldsymbol{x}$, and
introduce the indicator variables
\begin{equation}
	\xi_{i}^{l}:=\begin{cases}
	{\bf 1}_{\left\{ \left|\boldsymbol{a}_{i,\perp}^{\top}\boldsymbol{r}_{\perp}^{(l)}\right|\leq\frac{1}{m}\sqrt{\frac{n-1}{2n}}\right\} },\quad & 1\leq i\leq k,\\
	{\bf 1}_{\big\{ \left|\boldsymbol{a}_{i,\perp}^{\top}\boldsymbol{r}_{\perp}^{(l)}\right|\leq\sqrt{\frac{2\left(n-1\right)\log n}{n}}\big\} }, & i>k,
\end{cases}
	\label{eq:Defn-xi}
\end{equation}
where $k=\frac{m}{4\log m}$ as before. In words, we divide
$\boldsymbol{a}_{i,\perp}^{\top}\boldsymbol{r}_{\perp}^{(l)}$, $1\leq
i\leq m$ into two groups, with the first group enforcing far more
stringent control than the second group.  These indicator variables are useful since any $\boldsymbol{w}^{(l)}$ obeying
$\prod_{i=1}^m\xi_i^l=1$ will satisfy (\ref{eq:ratio-bound}) when $n$ is sufficiently large.   To see this, note that for the first group of indices,  $\xi_i^l=1$  requires
\begin{equation}
	\left|\boldsymbol{a}_{i,\perp}^{\top}\boldsymbol{r}_{\perp}^{(l)}\right|\leq\frac{1}{m}\sqrt{\frac{n-1}{2n}}\leq\frac{2}{m}\frac{\sqrt{n-1}}{\sqrt{n}-2}\big\Vert \boldsymbol{r}^{(l)}\big\Vert \leq\frac{3}{m}\big\Vert \boldsymbol{r}^{(l)}\big\Vert ,
	\quad1\leq i\leq k,
	\label{eq:r_perp_UB_small_k}
\end{equation}
where the second inequality follows from (\ref{eq:w-x-sep}). This  taken collectively with (\ref{eq:LB-Gaussian}) and (\ref{eq:key-quantity}) yields
\begin{eqnarray*}
	\frac{\left|\boldsymbol{a}_{i}^{\top}\boldsymbol{r}^{(l)}\right|^{2}}{\left|\boldsymbol{a}_{i}^{\top}\boldsymbol{x}\right|^{2}} 
	& \leq & \frac{2 \|\boldsymbol{r}^{(l)}\| ^{2}}{\left\Vert \boldsymbol{x}\right\Vert ^{2}} + \frac{\frac{9}{m^{2}}\left\Vert \boldsymbol{r}^{(l)}\right\Vert ^{2}}{\frac{1}{m^2\log^2m}\Vert \boldsymbol{x} \Vert ^2}
	\leq \frac{ (2+ 9 \log^2m)\left\Vert \boldsymbol{r}^{(l)}\right\Vert ^{2} }{ \Vert \boldsymbol{x} \Vert ^2}  ,\quad1\leq i\leq k.
\end{eqnarray*}
Regarding the second group of indices, $\xi_i^l=1$ gives
\begin{equation}
	\left|\boldsymbol{a}_{i,\perp}^{\top}\boldsymbol{r}_{\perp}^{(l)}\right|\leq\sqrt{\frac{2\left(n-1\right)\log n}{n}}
	\leq \sqrt{17\log n}\big\Vert \boldsymbol{r}^{(l)}\big\Vert ,\quad i=k+1,\cdots,m,
	\label{eq:r_perp_UB_large_k}
\end{equation}
where the last inequality again follows from (\ref{eq:w-x-sep}).
Plugging (\ref{eq:r_perp_UB_large_k}) and (\ref{eq:UB-nontrivial-Gaussian})
into (\ref{eq:key-quantity})  gives 
\begin{eqnarray*}
	\frac{\left|\boldsymbol{a}_{i}^{\top}\boldsymbol{r}^{(l)}\right|^{2}}{\left|\boldsymbol{a}_{i}^{\top}\boldsymbol{x}\right|^{2}} & \leq & \frac{2 \| \boldsymbol{r}^{(l)} \| ^{2}}{\left\Vert \boldsymbol{x}\right\Vert ^{2}} 
	+ \frac{17\left\Vert \boldsymbol{r}^{(l)}\right\Vert ^{2}\log n}{\left\Vert \boldsymbol{x}\right\Vert ^{2}/\log^{2}m}
	\leq \frac{(2+17\log^3m)\left\Vert \boldsymbol{r}^{(l)}\right\Vert ^{2}}{\left\Vert \boldsymbol{x}\right\Vert ^{2}} ,\quad i\geq k+1.
\end{eqnarray*}
Consequently, (\ref{eq:ratio-bound}) is satified for all  $1\leq i\leq m$.  It then suffices to guarantee the existence of  exponentially many
vectors obeying $\prod_{i=1}^m\xi_i^l=1$.

Note that the first group of indicator variables are quite stringent, namely, for each $i$ only a fraction $O(1/m)$ of the equations could satisfy $\xi_i^l=1$.  Fortunately,   $M_1$ is exponentially large, and hence even ${M_1}/{m^k}$ is exponentially large. 
Put formally, we claim that the
first group satisfies
\begin{eqnarray}
	\sum_{l=1}^{M_{1}}\prod_{i=1}^{k}\xi_{i}^{l} & \geq & \frac{1}{2}\frac{M_{1}}{\left(2\pi\right)^{k/2}(1+4\sqrt{k/n})^{k/2}}\left(\frac{1}{\sqrt{2\pi}m}\right)^{k}  := \widetilde{M}_1
	\label{eq:claim-count}
\end{eqnarray}
with probability exceeding $1-\exp\left(-\Omega\left(k\right)\right)-\exp (- \widetilde{M}_1 / 4 )$.
With this claim in place (which will be proved later), one has 
\[
	\sum_{l=1}^{M_{1}}\prod_{i=1}^{k}\xi_{i}^{l}\geq\frac{1}{2}M_{1}\frac{1}{\left(e^{2}m\right)^{k}}
	= \frac{1}{2}\exp\left(\left(\frac{1}{20}-\frac{k\left(2+\log m\right)}{n}\right)n\right)
	\geq \frac{1}{2}\exp\left(\frac{1}{25}n\right)
\]
when $n$ and $m/n$ are sufficiently large. 
% In fact, any $\boldsymbol{w}^{(l)}$
%obeying $\prod_{i=1}^{k}\xi_{i}^{l}=1$ must have very small $\boldsymbol{a}_{i,\perp}^{\top}\boldsymbol{r}_{\perp}^{(l)}$ (on the order of $\frac{1}{m}\|\boldsymbol{r}^{l}\|$)---not much larger than $\min_{i} |\boldsymbol{a}_i^{\top}\boldsymbol{x}|$. 
% @@@
In light of this, we will let $\mathcal{M}_{2}$ be a collection comprising
all $\boldsymbol{w}^{(l)}$ obeying $\prod_{i=1}^{k}\xi_{i}^{l}=1$,
which has size $M_{2}\geq\frac{1}{2}\exp\left(\frac{1}{25}n\right)$
based on the preceding argument.  For notational simplicity, it will be
assumed that the vectors in $\mathcal{M}_{2}$ are exactly $\boldsymbol{w}^{(j)}$
($1\leq j\leq M_{2}$).

We now move on to the second group by examining how many vectors
$\boldsymbol{w}^{(j)}$ in $\mathcal{M}_{2}$ further satisfy
$\prod_{i=k+1}^{m}\xi_{i}^{j}=1$. Notably, the above construction of
$\mathcal{M}_2$ relies only on $\{\boldsymbol{a}_i\}_{1 \le i \le k}$
and is independent of the remaining vectors
$\{\boldsymbol{a}_i\}_{i> k}$. In what follows the argument proceeds
conditional on $\mathcal{M}_2$ and
$\{\boldsymbol{a}_i\}_{1 \le i \le k}$. 
Applying the union bound
%and recalling that $\boldsymbol{a}_{i,\perp}^{\top}\boldsymbol{r}_{\perp}^{(l)}\sim\mathcal{N}\left(0,\sqrt{\frac{n-1}{2n}}\right)$
gives 
\begin{align*}
 	& \mathbb{E}\left[\sum\nolimits_{j=1}^{M_{2}}\left(1-\prod\nolimits_{i=k+1}^{m}\xi_{i}^{j}\right)\right]
	= \sum\nolimits_{j=1}^{M_{2}}\mathbb{P}\bigg\{ \exists i\text{ }(k<i\leq m):\text{ }\left|\boldsymbol{a}_{i,\perp}^{\top}\boldsymbol{r}_{\perp}^{(l)}\right|>\sqrt{\frac{2\left(n-1\right)\log n}{n}} \bigg\} \\
 	& \quad\leq \text{ }\sum_{j=1}^{M_{2}}\sum_{i=k+1}^{m}\mathbb{P}\left\{ \left|\boldsymbol{a}_{i,\perp}^{\top}\boldsymbol{r}_{\perp}^{(l)}\right|>\sqrt{\frac{2\left(n-1\right)\log n}{n}}\right\} \text{}\leq\text{ }M_{2}m\frac{1}{n^{2}}.
\end{align*}
This combined with Markov's inequality gives
\[
	\sum\nolimits_{j=1}^{M_{2}}\left(1-\prod\nolimits_{i=k+1}^{m}\xi_{i}^{j}\right)\leq\frac{m\log m}{ n^{2}}\cdot M_{2}
\]
with probability $1- o(1)$. Putting the above inequalities
together suggests that with probability $1-o(1)$,
there exist at least
\[
	\left(1-\frac{m\log m}{n^{2}}\right)M_{2}\geq\frac{1}{2}\left(1-\frac{m\log m}{n^{2}}\right)\exp\left(\frac{1}{25}n\right)
	\geq \exp\left(\frac{n}{30}\right)
\]
vectors in $\mathcal{M}_{2}$ satisfying $\prod_{l=k+1}^m\xi_i^l=1$. 
%@@@
We then choose $\mathcal{M}$ to be the set consisting of all these vectors, which forms a valid collection satisfying
the properties of Lemma \ref{lemma:MinimaxConstruction}.

%We justify that $\mathcal{M}$ is indeed a valid collection satisfying
%the properties of Lemma \ref{lemma:MinimaxConstruction}. 
%@@@

Finally, the only remaining step is to establish the claim (\ref{eq:claim-count}).
To start with, consider an $n\times k$ matrix $\boldsymbol{B}:=\left[\boldsymbol{b}_{1},\cdots,\boldsymbol{b}_{k}\right]$
of i.i.d.~standard normal entries, and let $\boldsymbol{u}\sim\mathcal{N}\left({\bf 0},\frac{1}{n}\boldsymbol{I}_{n}\right)$.
Conditional on the $\{\boldsymbol{b}_{i}$'s, 
\[
	\boldsymbol{b}_{\boldsymbol{u}}=\left[\begin{array}{c}
		b_{1,\boldsymbol{u}}\\
		\vdots\\
		b_{k,\boldsymbol{u}}
		\end{array}\right]
	:= \left[ \begin{array}{c}
		\boldsymbol{b}_{1}^{\top}\boldsymbol{u}\\
		\vdots\\
		\boldsymbol{b}_{k}^{\top}\boldsymbol{u}
		\end{array}\right]
	\sim \mathcal{N}\left({\bf 0},\text{ }\frac{1}{n}\boldsymbol{B}^{\top}\boldsymbol{B}\right).
\]
For sufficiently large $m$, one has $k=\frac{m}{4\log m}\leq\frac{1}{4}n$.
Using \cite[Corollary 5.35]{Vershynin2012} we get 
\begin{equation}
	\Big\Vert \frac{1}{n}\boldsymbol{B}^{\top}\boldsymbol{B}-\boldsymbol{I} \Big\Vert \leq 4\sqrt{k/n}
	\label{eq:B-I-norm}
\end{equation}
with probability $1- \exp\left(-\Omega(k)\right)$. Thus, for any
constant $0<\epsilon<\frac{1}{2}$, conditional on
$\{\boldsymbol{b}_{i}\}$ and (\ref{eq:B-I-norm}) we obtain
\begin{align}
 	& \mathbb{P}\left\{ \bigcap_{i=1}^{k}\left\{ |\boldsymbol{b}_{i}^{\top}\boldsymbol{u}| \leq\frac{1}{m}\right\} \right\} 	\text{} \\
	& \quad \geq\text{ } \left(2\pi\right)^{-\frac{k}{2}}\text{det}^{-\frac{1}{2}}\Big(\frac{1}{n}\boldsymbol{B}^{\top}\boldsymbol{B}\Big){\displaystyle \int}_{\boldsymbol{b}_{\boldsymbol{u}} \in \Upsilon}
		\exp\Big( - \frac{1}{2}\boldsymbol{b}_{\boldsymbol{u}}^{\top} \Big(\frac{1}{n}\boldsymbol{B}^{\top}\boldsymbol{B} \Big)^{-1}\boldsymbol{b}_{\boldsymbol{u}} \Big)  \mathrm{d}\boldsymbol{b}_{\boldsymbol{u}} \nonumber \\
 	& \quad\geq\text{ }\left(2\pi\right)^{-\frac{k}{2}}\left(1+4\sqrt{k/n}\right)^{-\frac{k}{2}}{\displaystyle \int}_{\boldsymbol{b}_{\boldsymbol{u}}\in\Upsilon}
		\exp\Big(-\frac{1}{2}\left(1-4\sqrt{k/n}\right)^{-1}\sum_{i=1}^{k}b_{i,u}^{2} \Big)\mathrm{d}\boldsymbol{b}_{\boldsymbol{u}}
	\label{eq:LB-E}\\
	& \quad\geq\text{ }\left(2\pi\right)^{-\frac{k}{2}}\big(1+4\sqrt{k/n}\big)^{-\frac{k}{2}} \big( \sqrt{2\pi}m \big)^{-k},
	\label{eq:LB-E2}
\end{align}
where $\Upsilon:=\{ \tilde{\boldsymbol{b}}\mid|\tilde{b}_{i}|\leq m^{-1},1\leq i\leq k \}$ and (\ref{eq:LB-E}) is a direct consequence from (\ref{eq:B-I-norm}).

When it comes to our quantity of interest, the above lower bound
(\ref{eq:LB-E2}) indicates that on an event (defined via $\{\boldsymbol{a}_i\}$) of
probability approaching 1, we have 
\begin{eqnarray}
	\mathbb{E}\Big[\sum\nolimits_{l=1}^{M_{1}}\prod\nolimits_{i=1}^{k}\xi_{i}^{l} \Big] & \geq & \text{ }M_{1}\left(2\pi\right)^{-\frac{k}{2}}\left(1+4\sqrt{k/n}\right)^{-\frac{k}{2}}
	\big( \sqrt{2\pi}m \big)^{-k}.
\end{eqnarray}
Since conditional on $\{\boldsymbol{a}_i\}$,
$\prod\nolimits_{i=1}^{k}\xi_{i}^{l}$ are independent across $l$,
applying the Chernoff-type bound \cite[Theorem
4.5]{mitzenmacher2005probability} gives 
\[
	\sum\nolimits_{l=1}^{M_{1}}\prod\nolimits_{i=1}^{k}\xi_{i}^{l}\geq\frac{M_{1}}{2}\left(2\pi\right)^{-\frac{k}{2}}\left(1+4\sqrt{k/n}\right)^{-\frac{k}{2}}
	\big( \sqrt{2\pi}m \big)^{-k}
\]
with probability exceeding $1-\exp\Big( -\frac{1}{8}\frac{M_{1}}{\left(2\pi\right)^{k/2} (1+4\sqrt{k/n} )^{k/2}}\left(\frac{1}{\sqrt{2\pi}m}\right)^{k} \Big)$. This 
concludes the proof. 

\subsection{Proof of Lemma \ref{lemma-KL-UB}\label{sub:Proof-of-Lemma-KL-UB}}

%\ejc{Introduce $\chi^2$ divergence first with basic inequality. No
%  need to introduce it in Section 7 since it is not used there. Then
%  show calculations for for $\chi^2$ divergence.}\yxc{Revised as suggested.}

Before proceeding, we introduce the $\chi^{2}$-divergence between  two
probability measures $P$ and $Q$ as
\begin{equation}
	\chi^{2}\left(P\|Q\right):={\displaystyle \int}\bigg(\frac{\mathrm{d}P}{\mathrm{d}Q}\bigg)^{2}\mathrm{d}Q-1.
\end{equation}
It is well known (e.g. \cite[Lemma 2.7]{tsybakov2009introduction}) that
\begin{equation}
	\mathsf{KL}\left(P\|Q\right) \leq \log (1+ \chi^{2}\left(P\|Q\right)), \label{eq:KL-chi-UB}
\end{equation}
and hence it suffices to develop an upper bound on the $\chi^2$ divergence.
%which we shall rely on to upper bound the KL divergence under study.
%which is useful for the calculation below.

Under independence, for any $\boldsymbol{w}_{0},\boldsymbol{w}_{1}\in\mathbb{R}^{n}$,
the decoupling identity of the $\chi^{2}$ divergence \cite[Page 96]{tsybakov2009introduction} gives
\begin{eqnarray}
	\chi^{2}\left(\mathbb{P}\left(\boldsymbol{y}\mid\boldsymbol{w}_{1}\right)\hspace{0.3em}\|\hspace{0.3em}\mathbb{P}\left(\boldsymbol{y}\mid\boldsymbol{w}_{0}\right)\right) 
	& = & \prod\nolimits_{i=1}^{m}\left(1+\chi^{2}\left(\mathbb{P}\left(y_{i}\mid\boldsymbol{w}_{1}\right)\hspace{0.3em}\|\hspace{0.3em}\mathbb{P}\left(y_{i}\mid\boldsymbol{w}_{0}\right)\right)\right)-1\nonumber \\
 	& = & \exp \left(\sum\nolimits_{i=1}^{m}\frac{\left(|\boldsymbol{a}_{i}^{\top}\boldsymbol{w}_{1}|^{2} - |\boldsymbol{a}_{i}^{\top}\boldsymbol{w}_{0} |^{2}\right)^{2}}{ |\boldsymbol{a}_{i}^{\top}\boldsymbol{w}_{0} |^{2}}\right) - 1.
	\label{eq:second}
\end{eqnarray}
The preceding identity (\ref{eq:second}) arises from the following computation:
by definition of $\chi^2(\cdot \| \cdot)$,
\begin{align*}
	 & \chi^{2}\left(\mathsf{Poisson}\left(\lambda_{1}\right)\hspace{0.3em}\|\hspace{0.3em}\mathsf{Poisson}\left(\lambda_{0}\right)\right)
		= \bigg\{ \sum\nolimits_{k=0}^{\infty}\frac{\left(\lambda_{1}^{k}\exp\left(-\lambda_{1}\right)\right)^{2}}{\lambda_{0}^{k}\exp\left(-\lambda_{0}\right)k!} \bigg\} -1\\
	 & =\exp\Big(\lambda_{0}-2\lambda_{1}+\frac{\lambda_{1}^{2}}{\lambda_{0}}\Big)\Big\{ \sum\nolimits_{k=0}^{\infty}\frac{\left( \lambda_{1}^{2} /\lambda_{0} \right)^{k}}{k!}\exp\Big(-\frac{\lambda_{1}^{2}}{\lambda_{0}}\Big)\Big\} -1
	= \exp\Big(\frac{\left(\lambda_{1}-\lambda_{0}\right)^{2}}{\lambda_{0}}\Big)-1.
\end{align*}
Set $\boldsymbol{r}:=\boldsymbol{w}_{1}-\boldsymbol{w}_{0}$. To summarize,  
\begin{align}
  \mathsf{KL}\left(\mathbb{P}\left(\boldsymbol{y}\mid\boldsymbol{w}_{1}\right)\hspace{0.3em}\|\hspace{0.3em}\mathbb{P}\left(\boldsymbol{y}\mid\boldsymbol{w}_{0}\right)\right) & \le \text{ } 
  \sum_{i=1}^{m}\frac{\left(|\boldsymbol{a}_{i}^{\top}\boldsymbol{w}_{1}|^2
      - |\boldsymbol{a}_{i}^{\top}\boldsymbol{w}_{0} |^2 \right)^{2}}{
    |\boldsymbol{a}_{i}^{\top}\boldsymbol{w}_{0} |^2} \text{ } \\ & 
  \leq \text{ } \sum_{i=1}^{m}\frac{ \left|\boldsymbol{a}_{i}^{\top}\boldsymbol{r}\right|^{2} \left(2\left|\boldsymbol{a}_{i}^{\top}\boldsymbol{w}_{0}\right|+\left|\boldsymbol{a}_{i}^{\top}\boldsymbol{r}\right|\right)^{2}}{|\boldsymbol{a}_{i}^{\top}\boldsymbol{w}_{0} |^2 }\nonumber \\
  &  = \text{ } \sum_{i=1}^{m}
  |\boldsymbol{a}_{i}^{\top}\boldsymbol{r} |^2 \left(\frac{8
      |\boldsymbol{a}_{i}^{\top}\boldsymbol{w}_{0} |^2 + 2
      |\boldsymbol{a}_{i}^{\top}\boldsymbol{r} |^2 }{
      |\boldsymbol{a}_{i}^{\top}\boldsymbol{w}_{0} |^2 }\right).
	\label{eq:KL-UB}
\end{align}

\section{Initialization via truncated spectral Method\label{proof-truncated-spectral}}

This section demonstrates that the truncated spectral
method works when $m\asymp n$, as stated in the
proposition below.

\begin{proposition}
  Fix $\delta>0$
  and $\boldsymbol{x}\in\mathbb{R}^{n}$. Consider the model where $y_i= |\boldsymbol{a}_i^{\top}\boldsymbol{x}|^2 +\eta_i$ and $\boldsymbol{a}_i \overset{\mathrm{ind.}}{\sim} \mathcal{N}({\bf 0}, \boldsymbol{I})$. 
% and let $\boldsymbol{z}^{(0)}$ be the solution returned by the truncated
% spectral method. 
Suppose that 
%$\left\Vert \boldsymbol{\eta}\right\Vert _{\infty}\leq\varepsilon\left\Vert \boldsymbol{x}\right\Vert ^2$
\begin{equation}
  \label{eq:eta_cond}
  |\eta_i| \leq \varepsilon \max\{ \| \boldsymbol{x} \|^2,~ |\boldsymbol{a}_i^{\top}\boldsymbol{x}|^2 \}, \quad 1\leq i\leq m
\end{equation}
for some sufficiently small constant $\varepsilon>0$. With probability
exceeding $1-\exp\left(-\Omega\left(m\right)\right)$, the solution $\boldsymbol{z}^{(0)}$ returned by the truncated
spectral method obeys
\begin{equation}
	\mathrm{dist} (\boldsymbol{z}^{(0)},\boldsymbol{x}) \leq \delta \| \boldsymbol{x} \|,  
	\label{eq:quality-SM}
\end{equation}
provided that $m>c_{0}n$ for some constant $c_{0}>0$.
\end{proposition}

\begin{proof}
By homogeneity, it suffices to consider the case where $\Vert \boldsymbol{x} \Vert =1$.
Recall from \cite[Lemma 3.1]{candes2012phaselift} that 
$\frac{1}{m}\sum_{i=1}^{m} |\boldsymbol{a}_i^{\top}\boldsymbol{x}|^2 \in [1\pm\varepsilon]\Vert \boldsymbol{x}\Vert ^2$ holds with probability $1-\exp(-\Omega(m))$. 
%In what follows, we will mainly prove the theorem for the case where $\Vert \boldsymbol{\eta} \Vert _{\infty}\leq \varepsilon \|\boldsymbol{x}\|^2$.  
Under the hypothesis (\ref{eq:eta_cond}), 
%$\Vert \boldsymbol{\eta} \Vert _{\infty}\leq \varepsilon \|\boldsymbol{x}\|^2$, 
one has 
\[
  \frac{1}{m}\left\Vert \boldsymbol{\eta}\right\Vert _{1}\leq   \frac{1}{m} \sum\nolimits_{i=1}^m \varepsilon (  \|\boldsymbol{x}\|^{2} + |\boldsymbol{a}_i^{\top}\boldsymbol{x}|^2 ) \leq \varepsilon \|\boldsymbol{x}\|^2 + \varepsilon(1+\epsilon) \|\boldsymbol{x}\|^2 \leq 3\varepsilon \|\boldsymbol{x}\|^2,
\]
which yields
\begin{align}
	\lambda_0^2 := \frac{1}{m} \sum\nolimits_{l=1}^{m} y_{l}  \text{ }=\text{ }  
	\frac{1}{m} \sum\nolimits_{l=1}^{m} |\boldsymbol{a}_{l}^{\top}\boldsymbol{x}|^{2} + \frac{1}{m}\sum_{l=1}^{m}{\eta}_l
	\text{ }\in\text{ }  [1\pm 4\varepsilon] \Vert \boldsymbol{x} \Vert ^2
\label{eq:lambda0}
\end{align}
with probability $1-\exp(-\Omega(m))$. 

Consequently, when $|\eta_i| \leq \varepsilon \|\boldsymbol{x}\|^2$, one has
%This in turn implies that
%
\begin{align}
\begin{array}{ll}
	%{\bf 1}_{\left\{ |\boldsymbol{a}_{i}^{\top}\boldsymbol{x}|^{2}\leq\left(1-2\varepsilon\right)\alpha_{y}^{2}-\varepsilon\right\} }
	{\bf 1}_{\left\{ |(\boldsymbol{a}_{i}^{\top}\boldsymbol{x})^{2}+\eta_{i}|\leq\alpha_{y}^{2}(\frac{1}{m}\sum_{l}y_{l}) \right\} }
	\leq {\bf 1}_{\left\{ |\boldsymbol{a}_{i}^{\top}\boldsymbol{x}|^{2} \leq\alpha_{y}^{2}(\frac{1}{m}\sum_{l}y_{l}) + |\eta_i|\right\} }
	\leq {\bf 1}_{\left\{ |\boldsymbol{a}_{i}^{\top}\boldsymbol{x}|^{2}\leq(1+4\varepsilon)\alpha_{y}^{2}+\varepsilon\right\}; } \\
	{\bf 1}_{\left\{ |(\boldsymbol{a}_{i}^{\top}\boldsymbol{x})^{2}+\eta_{i}|\leq\alpha_{y}^{2}(\frac{1}{m}\sum_{l}y_{l}) \right\} }
	\geq {\bf 1}_{\left\{ |\boldsymbol{a}_{i}^{\top}\boldsymbol{x}|^{2} \leq\alpha_{y}^{2}(\frac{1}{m}\sum_{l}y_{l}) - |\eta_i|\right\} }
	\geq {\bf 1}_{\left\{ |\boldsymbol{a}_{i}^{\top}\boldsymbol{x}|^{2}\leq(1-4\varepsilon)\alpha_{y}^{2} - \varepsilon\right\}. }\end{array}
\end{align}
Besides, in the case where $|\eta_i| \leq \varepsilon |\boldsymbol{a}_i^{\top}\boldsymbol{x}|^2$,  
\begin{align}
\begin{array}{ll}
	%{\bf 1}_{\left\{ |\boldsymbol{a}_{i}^{\top}\boldsymbol{x}|^{2}\leq\left(1-2\varepsilon\right)\alpha_{y}^{2}-\varepsilon\right\} }
	{\bf 1}_{\left\{ |(\boldsymbol{a}_{i}^{\top}\boldsymbol{x})^{2}+\eta_{i}|\leq\alpha_{y}^{2}(\frac{1}{m}\sum_{l}y_{l}) \right\} }
	\leq {\bf 1}_{\left\{ (1-\varepsilon) |\boldsymbol{a}_{i}^{\top}\boldsymbol{x}|^{2} \leq\alpha_{y}^{2}(\frac{1}{m}\sum_{l}y_{l}) \right\} }
	\leq {\bf 1}_{\left\{ |\boldsymbol{a}_{i}^{\top}\boldsymbol{x}|^{2}\leq \frac{1+4\varepsilon}{1- \varepsilon}\alpha_{y}^{2} \right\} }; \\
	{\bf 1}_{\left\{ |(\boldsymbol{a}_{i}^{\top}\boldsymbol{x})^{2}+\eta_{i}|\leq\alpha_{y}^{2}(\frac{1}{m}\sum_{l}y_{l}) \right\} }
	\geq {\bf 1}_{\left\{ (1+\varepsilon)|\boldsymbol{a}_{i}^{\top}\boldsymbol{x}|^{2} \leq\alpha_{y}^{2}(\frac{1}{m}\sum_{l}y_{l}) \right\} }
	\geq {\bf 1}_{\left\{ |\boldsymbol{a}_{i}^{\top}\boldsymbol{x}|^{2}\leq \frac{1-4\varepsilon}{1+ \varepsilon}\alpha_{y}^{2} \right\} }.\end{array}
\end{align}
Taken collectively, these inequalities imply that
\begin{equation}
	\underset{:=\boldsymbol{Y}_{2}}
		{\underbrace{\frac{1}{m}\sum_{i=1}^{m}\boldsymbol{a}_{i}\boldsymbol{a}_{i}^{\top}\left(\boldsymbol{a}_{i}^{\top}\boldsymbol{x}\right)^{2}{\bf 1}_{\left\{ |\boldsymbol{a}_{i}^{\top}\boldsymbol{x}|\leq \gamma_2 \right\} }}}
	\preceq \text{ }\boldsymbol{Y}\text{ }
	\preceq \text{ } \underset{:=\boldsymbol{Y}_{1}}
		{\underbrace{\frac{1}{m}\sum_{i=1}^{m}\boldsymbol{a}_{i}\boldsymbol{a}_{i}^{\top}\left(\boldsymbol{a}_{i}^{\top}\boldsymbol{x}\right)^{2}{\bf 1}_{\left\{ |\boldsymbol{a}_{i}^{\top}\boldsymbol{x}|\leq \gamma_1 \right\} }}}.
	\label{eq:Y-sandwich}
\end{equation}
where $\gamma_1:= \max \{ \sqrt{(1+4\varepsilon)\alpha_{y}^{2}+\varepsilon}, \sqrt{ \frac{1+4\varepsilon}{1-\varepsilon} }\alpha_{y} \}$ and $\gamma_2:= \min \{ \sqrt{(1-4\varepsilon)\alpha_{y}^{2}-\varepsilon}, \sqrt{ \frac{1-4\varepsilon}{1+\varepsilon} }\alpha_{y} \}$.  
Letting $\xi\sim\mathcal{N}\left(0,1\right)$, one can compute 
\begin{equation}
	\,\mathbb{E}\left[\boldsymbol{Y}_{1}\right]=\beta_{1}\boldsymbol{x}\boldsymbol{x}^{\top} + \beta_{2}\boldsymbol{I},
	\quad\text{and}\quad
	\mathbb{E}\left[\boldsymbol{Y}_{2}\right]=\beta_{3}\boldsymbol{x}\boldsymbol{x}^{\top} + \beta_{4}\boldsymbol{I},
\end{equation}
where 
$\beta_1 := \mathbb{E}\big[\xi^{4}{\bf 1}_{\{|\xi|\leq \gamma_1 \}}\big]
	  - \mathbb{E}\big[\xi^{2}{\bf 1}_{\{|\xi|\leq \gamma_1 \}}\big]$,
$\beta_2 := \mathbb{E}\big[\xi^{2}{\bf 1}_{\{|\xi|\leq \gamma_1 \}}\big]$,
$\beta_3 := \mathbb{E}\big[\xi^{4}{\bf 1}_{\{|\xi|\leq \gamma_2 \}}\big]
	  - \mathbb{E}\big[\xi^{2}{\bf 1}_{\{|\xi|\leq \gamma_2 \}}\big]$
and 
$\beta_4 := \mathbb{E}\big[\xi^{2}{\bf 1}_{\{|\xi|\leq \gamma_2 \}}\big]$.
Recognizing that 
$\boldsymbol{a}_{i}\boldsymbol{a}_{i}^{\top}\left(\boldsymbol{a}_{i}^{\top}\boldsymbol{x}\right)^{2}{\bf 1}_{\left\{ |\boldsymbol{a}_{i}^{\top}\boldsymbol{x})|\leq c\right\} }$
can be rewritten as $\boldsymbol{b}_{i}\boldsymbol{b}_{i}^{\top}$
for some sub-Gaussian vector 
$\boldsymbol{b}_{i}:=\boldsymbol{a}_{i}\left(\boldsymbol{a}_{i}^{\top}\boldsymbol{x}\right){\bf 1}_{\left\{ |\boldsymbol{a}_{i}^{\top}\boldsymbol{x})|\leq c\right\} }$,
we apply standard results on random matrices with non-isotropic sub-Gaussian
rows \cite[Equation (5.26)]{Vershynin2012} to deduce 
\begin{equation}
	\left\Vert \boldsymbol{Y}_{1}-\mathbb{E}\left[\boldsymbol{Y}_{1}\right]\right\Vert \leq\delta,\quad\left\Vert \boldsymbol{Y}_{2}-\mathbb{E}\left[\boldsymbol{Y}_{2}\right]\right\Vert 
	\leq \delta
	\label{eq:deviation-spectral-2}
\end{equation}
with probability $1-\exp\left(-\Omega\left(m\right)\right)$, provided
that $m/n$ exceeds some large constant. Besides, when $\varepsilon$
is sufficiently small, one further has $\|\mathbb{E}\left[\boldsymbol{Y}_{1}\right]-\mathbb{E}\left[\boldsymbol{Y}_{2}\right]\|\leq\delta$.
These taken collectively with (\ref{eq:Y-sandwich}) give
\begin{equation}
	\|\boldsymbol{Y}-\beta_{1}\boldsymbol{x}\boldsymbol{x}^{\top}-\beta_2 \boldsymbol{I}\| \leq 3 \delta.
	%\quad\text{and}\quad
	%\boldsymbol{x}^{\top}\boldsymbol{Y}\boldsymbol{x} \geq \beta_1 + \beta_2 - 2 \delta.
	\label{eq:dev-spectral}
\end{equation}

Fix some small $\tilde{\delta}>0$. With (\ref{eq:dev-spectral}) in place, applying the Davis-Kahan $\sin\Theta$ theorem \cite{davis1970rotation} and taking $\delta,\varepsilon$ to be
sufficiently small, we obtain
\[
	\mathrm{dist}(  \tilde{\boldsymbol{z}},\boldsymbol{x} ) \leq \tilde{\delta}, 
\]
where $\tilde{\boldsymbol{z}}$ is the leading eigenvector of  $\boldsymbol{Y}$. Since  $\boldsymbol{z}^{(0)} := \lambda_0\tilde{\boldsymbol{z}}$, one derives
\begin{align}
	\mathrm{dist}(  \boldsymbol{z}^{(0)},\boldsymbol{x} ) & \leq  \mathrm{dist}(  	\lambda_0\tilde{\boldsymbol{z}},  \tilde{\boldsymbol{z}} ) + \mathrm{dist}(  \tilde{\boldsymbol{z}},\boldsymbol{x} ) \nonumber\\ 
	& \leq   |\lambda_0 -1| + \tilde{\delta}  ~\leq~  \max\{ \sqrt{1+2\varepsilon}-1,~ 1 - \sqrt{1-2\varepsilon}\} + \tilde{\delta} 
\end{align}
as long as $m/n$ is sufficiently large, where the last inequality follows from (\ref{eq:lambda0}). Picking $\tilde{\delta}$ and $\varepsilon$ to be sufficiently small, we establish the claim. 

\end{proof}

We now justify that the Poisson model (\ref{eq:Poisson}) satisfies the condition (\ref{eq:eta_cond}). 
%$\|\boldsymbol{\eta}\|\leq \varepsilon \| \boldsymbol{x} \|^2$ whenever $\|\boldsymbol{x}\| \geq \log^{1.5}m$.  
Suppose that $\mu_{i}=(\boldsymbol{a}_{i}^{\top}\boldsymbol{x})^{2}$
and hence $y_{i}\sim\mathsf{Poisson}(\mu_{i})$. It follows from the
Chernoff bound that for all $t\geq 0$, 
\[
	\mathbb{P}\left(y_{i}-\mu_{i}\geq\tau\right)\leq\frac{\mathbb{E}\left[e^{ty_{i}}\right]}{\exp\left(t(\mu_{i}+\tau)\right)}=\frac{\exp\left(\mu_{i}\left(e^{t}-1\right)\right)}{\exp\left(t(\mu_{i}+\tau)\right)}=\exp\left(\mu_{i}\left(e^{t}-t-1\right)-t\tau\right),
%\quad\forall t\geq0.
\]
Taking $\tau=2\tilde{\varepsilon}\mu_{i}$ and $t=\tilde{\varepsilon}$
for any $0\leq\tilde{\varepsilon}\leq1$ gives
\begin{eqnarray*}
	\mathbb{P}\left(y_{i}-\mu_{i}\geq2\tilde{\varepsilon}\mu_{i}\right) & \leq & \exp\left(\mu_{i}\left(e^{t}-t-1-2\tilde{\varepsilon}t\right)\right)\text{ }\overset{(\text{i})}{\leq}\text{ }\exp\left(\mu_{i}\left(t^{2}-2\tilde{\varepsilon}t\right)\right)=\exp\left(-\mu_{i}\tilde{\varepsilon}^{2}\right),
\end{eqnarray*}
where (i) follows since $e^{t}\leq1+t+t^{2}$ ($0\leq t\leq1$). In addition, for any $\tilde{\varepsilon}> 1$, taking $t=1$ we get
\begin{eqnarray*}
	\mathbb{P}\left(y_{i}-\mu_{i}\geq2\tilde{\varepsilon}\mu_{i}\right) & \leq & \exp\left(\mu_{i}\left(e^{t}-t-1-2\tilde{\varepsilon}t\right)\right)\text{ } \leq ~\exp\left(-0.5\mu_{i}\tilde{\varepsilon}\right). 
\end{eqnarray*}

Suppose that $\|\boldsymbol{x}\| \gtrsim\log m$. When $\mu_i \gtrsim \| \boldsymbol{x} \|^2 $, setting $\tilde{\varepsilon}=0.5 \varepsilon < 1$ yields
\begin{eqnarray*}
	\mathbb{P}\left(y_{i}-\mu_{i}\geq {\varepsilon}\mu_{i}\right) {\leq}~ \exp\left(- \mu_{i}{\varepsilon}^{2} / 4 \right) ~\leq~\exp\left(- \Omega( {\varepsilon}^{2} \| \boldsymbol{x} \|^2   ) \right)~ =~ \exp\left(- \Omega( {\varepsilon}^{2} \log^2 m  ) \right).
\end{eqnarray*}
When $\mu_i \lesssim \| \boldsymbol{x} \|^2 $,  letting
$\kappa_{i} = \mu_{i}/\|\boldsymbol{x}\|^{2}$ and setting $\tilde{\varepsilon}=\varepsilon/2\kappa_{i} \gtrsim \varepsilon$,
we obtain
\begin{align*}
	& \mathbb{P}\left(y_{i}-\mu_{i}\geq\varepsilon\|\boldsymbol{x}\|^{2}\right)
	 =  \mathbb{P}\left(y_{i}-\mu_{i}\geq2\tilde{\varepsilon}\mu_{i}\right) 
	\leq \exp\left(- \min\{ \tilde{\varepsilon}, 0.5\} \cdot \tilde{\varepsilon} \mu_i \right) \\
	& \quad = \exp\Big( - 0.5\min\{ \tilde{\varepsilon}, 0.5\} \varepsilon \| \boldsymbol{x} \|^2 \Big) 
	%= \exp\Big(-\frac{\varepsilon^{2}\|\boldsymbol{x}\|^{2}}{4\kappa_{i}}\Big) \\
	 =  \exp\big(-\Omega\left(\varepsilon^2 \log^2 m\right)\big).
\end{align*}
%
%In addition, standard results on Gaussian measures indicate that
%$\max_{1\leq i\leq m}\kappa_{i}\lesssim\log m$. 
%
% As a consequence,
% if $\|\boldsymbol{x}\|^{2}\gtrsim\log^{3}m$, then $\frac{\|\boldsymbol{x}\|^{2}}{\kappa_{i}}\gtrsim\log^{2}m$
% ($1\leq i\leq m$), which further gives
%
% \begin{eqnarray*}
%	\mathbb{P}\left(\eta_{i}\geq\varepsilon\|\boldsymbol{x}\|^{2}\right) = \mathbb{P}\left( y_{i}-\mu_{i}\geq\varepsilon\|\boldsymbol{x}\|^{2}\right) 
%	& \leq & \exp\big(-\Omega\left(\varepsilon^{2}\log^{2}m\right)\big).
%	%\label{eq:eta_UB}
% \end{eqnarray*}
%
In view of the union bound, 
\begin{eqnarray}
	\mathbb{P} \left(\exists i: \eta_{i} \geq \varepsilon \max\{ \|\boldsymbol{x}\|^{2},~ |\boldsymbol{a}_i^{\top}\boldsymbol{x}|^2\} \right) ~ \leq ~ m\exp\big(-\Omega\left(\varepsilon^{2}\log^{2}m\right)\big) ~\rightarrow~ 0.
	\label{eq:eta_UB1}
\end{eqnarray}
%
%in view of the union bound. 
% Take the union bound to deduce 
%
%\[
% \eta_{i}=y_{i}-\mu_{i}\leq\varepsilon\|\boldsymbol{x}\|^{2},\quad\forall1\leq i\leq m.
% \]
%

Similarly, applying the same argument on $-y_{i}$ we get $\eta_{i}\geq-\varepsilon \max\{ \|\boldsymbol{x}\|^2, |\boldsymbol{a}_i^{\top}\boldsymbol{x}|^2 \}$
for all $i$, which together with (\ref{eq:eta_UB1}) establishes  the condition (\ref{eq:eta_cond}) with high probability. In conclusion, the claim (\ref{eq:quality-SM})
applies to the Poisson model.

%\yxc{I omit the remaining steps here to reduce the bits. Not sure whether we should fill}

\section{Local error contraction with backtracking line search\label{sec:Backtracking-line-search}}

In this section, we verify the effectiveness of a backtracking line
search strategy by showing local error contraction. To keep it
concise, we only sketch the proof for the noiseless case, but the
proof extends to the noisy case without much difficulty. Also we do
not strive to obtain an optimized constant.  For concreteness, we
prove the following proposition.
\begin{proposition}
  The claim in Proposition \ref{prop-Contraction}
  continues to hold if $\alpha_{h}\geq 6$, $\alpha_{z}^{\mathrm{ub}}\geq5$,
  $\alpha_{z}^{\mathrm{lb}}\leq0.1$, $\alpha_{p}\geq 5$, and
  \begin{equation}
	\Vert \boldsymbol{h}\Vert /\Vert \boldsymbol{z}\Vert \leq \epsilon_{\mathrm{tr}}
  \end{equation}
  for some constant $\epsilon_{\mathrm{tr}}>0$ independent of $n$ and $m$.
\end{proposition}

Note that if $\alpha_{h}\geq 6$,
$\alpha_{z}^{\mathrm{ub}}\geq 5$ and $\alpha_{z}^{\mathrm{lb}}\leq 0.1$,
then the boundary step size $\mu_{0}$ given in Proposition \ref{prop-Contraction}
satisfies 
\[
	\frac{0.994-\zeta_{1}-\zeta_{2}-\sqrt{2/(9\pi)}\alpha_{h}^{-1}}{2\left(1.02+0.665\alpha_{h}^{-1}\right)}\geq 0.384.
\]
Thus, it suffices to show that the step size obtained by a backtracking
line search lies within (0,0.384). For notational convenience, we will
set
\[
	\boldsymbol{p} := m^{-1}\nabla\ell_{\mathrm{tr}}\left(\boldsymbol{z}\right) \quad\text{and}\quad
	\mathcal{E}_{3}^{i} := \left\{ \left|\boldsymbol{a}_{i}^{\top}\boldsymbol{z}\right|\geq\alpha_{z}^{\mathrm{lb}}\left\Vert \boldsymbol{z}\right\Vert \text{ and }\left|\boldsymbol{a}_{i}^{\top}\boldsymbol{p}\right|\leq\alpha_{p}\left\Vert \boldsymbol{p}\right\Vert \right\} 
\]
throughout the rest of the proof. We also impose the assumption
\begin{equation}
\left\Vert \boldsymbol{p}\right\Vert /\left\Vert \boldsymbol{z}\right\Vert \leq\epsilon\label{eq:p-z-ratio}
\end{equation}
for some sufficiently small constant $\epsilon>0$, so that $\left|\boldsymbol{a}_{i}^{\top}\boldsymbol{p}\right|/\left|\boldsymbol{a}_{i}^{\top}\boldsymbol{z}\right|$
is small for all non-truncated terms. It is self-evident from (\ref{eq:regularity})
that in the regime under study, one has
\begin{equation}
	\left\Vert \boldsymbol{p}\right\Vert \geq 2\left\{ 1.99-2\left(\zeta_{1}+\zeta_{2}\right)-\sqrt{8/\pi}(3\alpha_{h})^{-1}-o\left(1\right)\right\} \left\Vert \boldsymbol{h}\right\Vert 
	\geq 3.64 \left\Vert \boldsymbol{h}\right\Vert .
	\label{eq:p-h-ratio}
\end{equation}

To start with, consider three scalars $h$, $b$, and $\delta$. Setting
$b_{\delta}:=\frac{\left(b+\delta\right)^{2}-b^{2}}{b^{2}}$, we get
\begin{align}
	 & \left(b+h\right)^{2}\log\frac{\left(b+\delta\right)^{2}}{b^{2}}-\left(b+\delta\right)^{2}+b^{2}=(b+h)^{2}\log\left(1+b_{\delta}\right)-b^{2}b_{\delta}\nonumber \\
 	& \quad\overset{(\text{i})}{\leq}\left(b+h\right)^{2}\left\{ b_{\delta}-0.4875b_{\delta}^{2}\right\} -b^{2}b_{\delta}\text{ }=\text{ }(\left(b+h\right)^{2}-b^{2})b_{\delta}-0.4875\left(b+h\right)^{2}b_{\delta}^{2}\nonumber \\
 	& \quad = h\delta\left(2+h/b\right)\left(2+\delta/b\right)-0.4875\left(1+h/b\right)^{2}\left|\delta\left(2+\delta/b\right)\right|^{2}\nonumber \\
 	& \quad=4h\delta+\frac{2h^{2}\delta}{b}+\frac{2h\delta^{2}}{b}+\frac{h^{2}\delta^{2}}{b^{2}}-0.4875\delta^{2}\left(1+\frac{h}{b}\right)^{2}\left(2+\frac{\delta}{b}\right)^{2},
	\label{eq:scalar-ub}
\end{align}
where (i) follows from the inequality $\log\left(1+x\right)\leq x-0.4875x^{2}$
for sufficiently small $x$. To further simplify the bound, observe that
\[
	\delta^{2}\left(1+\frac{h}{b}\right)^{2}\left(2+\frac{\delta}{b}\right)^{2}\geq4\delta^{2}\left(1+\frac{h}{b}\right)^{2}+\delta^{2}\left(1+\frac{h}{b}\right)^{2}\frac{4\delta}{b} 
\]
\[ 
	\text{and}\quad
	\frac{2h\delta^{2}}{b}+\frac{h^{2}\delta^{2}}{b^{2}}=\left(\left(1+\frac{h}{b}\right)^{2}-1\right)\delta^{2}.
\]
Plugging these two identities into (\ref{eq:scalar-ub}) yields
\begin{eqnarray*}
	(\ref{eq:scalar-ub}) & \leq & 4h\delta+\frac{2h^{2}\delta}{b}-\left(0.95\left(1+\frac{h}{b}\right)^{2}+1\right)\delta^{2}-0.4875\delta^{2}\left(1+\frac{h}{b}\right)^{2}\frac{4\delta}{b}\\
 	& \leq & 4h\delta-1.95\delta^{2}+\frac{2h^{2}\left|\delta\right|}{\left|b\right|} + \frac{1.9 |h|}{\left|b\right|}\delta^{2}+\frac{1.95\left|\delta^{3}\right|}{\left|b\right|}\left(1+\frac{h}{b}\right)^{2}.
\end{eqnarray*}
Replacing respectively $b$, $\delta$, and $h$ with $\boldsymbol{a}_{i}^{\top}\boldsymbol{z}$,
$\tau\boldsymbol{a}_{i}^{\top}\boldsymbol{p}$, and $-\boldsymbol{a}_{i}^{\top}\boldsymbol{h}$,
one sees that the log-likelihood $\ell_{i}\left(\boldsymbol{z}\right)=y_{i}\log(|\boldsymbol{a}_{i}^{\top}\boldsymbol{z}|^{2})-|\boldsymbol{a}_{i}^{\top}\boldsymbol{z}|^{2}$
obeys
\begin{align*}
 	& \ell_{i}  \left(\boldsymbol{z}+\tau\boldsymbol{p}\right)-\ell_{i}\left(\boldsymbol{z}\right)=y_{i}\log\frac{\left|\boldsymbol{a}_{i}^{\top}\left(\boldsymbol{z}+\tau\boldsymbol{p}\right)\right|^{2}}{\left|\boldsymbol{a}_{i}^{\top}\boldsymbol{z}\right|^{2}}-\left|\boldsymbol{a}_{i}^{\top}\left(\boldsymbol{z}+\tau\boldsymbol{p}\right)\right|^{2}+\left|\boldsymbol{a}_{i}^{\top}\boldsymbol{z}\right|^{2}\\
 	& \quad\leq  \underset{:=I_{1,i}}{\underbrace{-4\tau\left(\boldsymbol{a}_{i}^{\top}\boldsymbol{h}\right)\left(\boldsymbol{a}_{i}^{\top}\boldsymbol{p}\right)}}  -  \underset{:=I_{2,i}}{\underbrace{1.95\tau^{2}\left(\boldsymbol{a}_{i}^{\top}\boldsymbol{p}\right)^{2}}}+\underset{:=I_{3,i}}{\underbrace{\frac{2\tau\left(\boldsymbol{a}_{i}^{\top}\boldsymbol{h}\right)^{2}\left|\boldsymbol{a}_{i}^{\top}\boldsymbol{p}\right|}{\left|\boldsymbol{a}_{i}^{\top}\boldsymbol{z}\right|}}}+\underset{:=I_{4,i}}{\underbrace{\frac{1.9\tau^{2}\left|\boldsymbol{a}_{i}^{\top}\boldsymbol{h}\right|}{\left|\boldsymbol{a}_{i}^{\top}\boldsymbol{z}\right|}\left(\boldsymbol{a}_{i}^{\top}\boldsymbol{p}\right)^{2}}}\\
 	& \quad\quad\quad\quad\quad + \underset{:=I_{5,i}}{\underbrace{\frac{1.95\tau^{3}\left|\boldsymbol{a}_{i}^{\top}\boldsymbol{p}\right|^{3}}{\left|\boldsymbol{a}_{i}^{\top}\boldsymbol{z}\right|}\left(1-\frac{\boldsymbol{a}_{i}^{\top}\boldsymbol{h}}{\boldsymbol{a}_{i}^{\top}\boldsymbol{z}}\right)^{2}}}.
\end{align*}

The next step is then to bound each of these terms separately. Most
of the following bounds are straightforward consequences from \cite[Lemma 3.1]{candes2012phaselift}
combined with the truncation rule. For the first term, applying the
AM-GM inequality we get
\begin{align*}
	\frac{1}{m}\sum_{i=1}^{m}I_{1,i}{\bf 1}_{\mathcal{E}_{3}^{i}} &\leq\frac{4\tau}{3.64m}\sum_{i=1}^{m}\left\{ \frac{3.64^2}{2}\left(\boldsymbol{a}_{i}^{\top}\boldsymbol{h}\right)^{2}+\frac{1}{2}\left(\boldsymbol{a}_{i}^{\top}\boldsymbol{p}\right)^{2}\right\} \\
	&\leq \frac{4\tau\left(1+\delta\right)}{3.64}\left\{ \frac{3.64^2}{2}\left\Vert \boldsymbol{h}\right\Vert ^{2}+\frac{1}{2}\left\Vert \boldsymbol{p}\right\Vert ^{2}\right\} .
\end{align*}
Secondly, it follows from Lemma \ref{Lemma:I1_I2} that
\begin{eqnarray*}
	\frac{1}{m}\sum_{i=1}^{m}I_{2,i}{\bf 1}_{\mathcal{E}_{3}^{i}} 
	& = & -1.95\tau^{2}\frac{1}{m}\sum_{i=1}^{m}\left(\boldsymbol{a}_{i}^{\top}\boldsymbol{p}\right)^{2}{\bf 1}_{\mathcal{E}_{3}^{i}}
	\leq -1.95\left(1-\tilde{\zeta}_{1}-\tilde{\zeta}_{2}\right)\tau^{2}\left\Vert \boldsymbol{p}\right\Vert ^{2},
\end{eqnarray*}
where $\tilde{\zeta}_{1}:= \max \{ \mathbb{E}\big[\xi^{2}{\bf 1}_{\left\{ \left|\xi\right|\leq\sqrt{1.01}\alpha_{z}^{\mathrm{lb}}\right\} }\big], \mathbb{E}\big[{\bf 1}_{\left\{ \left|\xi\right|\leq\sqrt{1.01}\alpha_{z}^{\mathrm{lb}}\right\} }\big] \}$
and $\tilde{\zeta}_{2}:=\mathbb{E}\big[\xi^{2}{\bf 1}_{\left\{ \left|\xi\right|>\sqrt{0.99}\alpha_{h}\right\} }\big]$.
The third term is controlled by
\[
	\frac{1}{m}\sum_{i=1}^{m}I_{3,i}{\bf 1}_{\mathcal{E}_{3}^{i}}\leq2\tau\frac{\alpha_{p}\left\Vert \boldsymbol{p}\right\Vert }{\alpha_{z}^{\mathrm{lb}}\left\Vert \boldsymbol{z}\right\Vert }\left\{ \frac{1}{m}\sum_{i=1}^{m}\left(\boldsymbol{a}_{i}^{\top}\boldsymbol{h}\right)^{2}\right\} 
	\lesssim  \tau\epsilon\left\Vert \boldsymbol{h}\right\Vert ^{2}.
\]
Fourthly, it arises from the AM-GM inequality that 
\begin{eqnarray*}
	\frac{1}{m}\sum_{i=1}^{m}I_{4,i}{\bf 1}_{\mathcal{E}_{3}^{i}} 
	& \leq & \frac{1.9\tau^{2}\alpha_{p}\left\Vert \boldsymbol{p}\right\Vert }{\alpha_{z}^{\mathrm{lb}}\left\Vert \boldsymbol{z}\right\Vert }\frac{1}{m}\sum_{i=1}^{m}\left|\boldsymbol{a}_{i}^{\top}\boldsymbol{h}\right|\left|\boldsymbol{a}_{i}^{\top}\boldsymbol{p}\right|
	\lesssim \epsilon\tau^{2}\frac{1}{m}\sum_{i=1}^{m}\left\{ 2\left|\boldsymbol{a}_{i}^{\top}\boldsymbol{h}\right|^{2}+\frac{1}{8}\left|\boldsymbol{a}_{i}^{\top}\boldsymbol{p}\right|^{2}\right\} \\
	& \lesssim &  \epsilon\tau^{2}\left\Vert \boldsymbol{p}\right\Vert ^{2}.
\end{eqnarray*}
Finally, the last term is bounded by
\begin{eqnarray*}
	\frac{1}{m}\sum_{i=1}^{m}I_{5,i}{\bf 1}_{\mathcal{E}_{3}^{i}}
	 & \leq & \frac{1}{m}\sum_{i=1}^{m}\frac{1.95\tau^{3}\left|\boldsymbol{a}_{i}^{\top}\boldsymbol{p}\right|^{3}}{\left|\boldsymbol{a}_{i}^{\top}\boldsymbol{z}\right|}\left(\frac{\boldsymbol{a}_{i}^{\top}\boldsymbol{x}}{\boldsymbol{a}_{i}^{\top}\boldsymbol{z}}\right)^{2}
	\leq \frac{1.95\tau^{3}\alpha_{p}^{3}\left\Vert \boldsymbol{p}\right\Vert ^{3}}{(\alpha_{z}^{\mathrm{lb}})^{3}\left\Vert \boldsymbol{z}\right\Vert ^{3}}\frac{1}{m}\sum_{i=1}^{m}\left(\boldsymbol{a}_{i}^{\top}\boldsymbol{x}\right)^{2} \\
	& \lesssim & \tau^{3}\epsilon\frac{\left\Vert \boldsymbol{x}\right\Vert ^{2}}{\left\Vert \boldsymbol{z}\right\Vert ^{2}}\left\Vert \boldsymbol{p}\right\Vert ^{2}.
\end{eqnarray*}
Under the hypothesis (\ref{eq:p-h-ratio}), we can further derive
$\frac{1}{m}\sum_{i=1}^{m}I_{1,i}{\bf 1}_{\mathcal{E}_{3}^{i}}\leq\tau\left(1.1+\delta\right)\left\Vert \boldsymbol{p}\right\Vert ^{2}$.
Putting all the above bounds together yields that the truncated objective
function is majorized by
\begin{align}
 	& \frac{1}{m}\sum_{i=1}^{m}\left\{ \ell_{i}\left(\boldsymbol{z}+\tau\boldsymbol{p}\right)-\ell_{i}\left(\boldsymbol{z}\right)\right\} {\bf 1}_{\mathcal{E}_{3}^{i}}\leq\frac{1}{m}\sum_{i=1}^{m}\left(I_{1,i}+I_{2,i}+I_{3,i}+I_{4,i}+I_{5,i}\right){\bf 1}_{\mathcal{E}_{3}^{i}}\nonumber \\
 	& \quad\leq\tau\left(1.1+\delta\right)\left\Vert \boldsymbol{p}\right\Vert ^{2}-1.95\left(1-\tilde{\zeta}_{1}-\tilde{\zeta}_{2}\right)\tau^{2}\left\Vert \boldsymbol{p}\right\Vert ^{2}+\tau\tilde{\epsilon}\left\Vert \boldsymbol{p}\right\Vert ^{2}\nonumber \\
 	& \quad=\left\{ \tau\left(1.1+\delta\right)-1.95\left(1-\tilde{\zeta}_{1}-\tilde{\zeta}_{2}\right)\tau^{2}+\tau\tilde{\epsilon}\right\} \left\Vert \boldsymbol{p}\right\Vert ^{2}
	\label{eq:majorization}
\end{align}
for some constant $\tilde{\epsilon}>0$ that is linear in $\epsilon$.

Note that the backtracking line search seeks a point satisfying $\frac{1}{m}\sum_{i=1}^{m}\left\{ \ell_{i}\left(\boldsymbol{z}+\tau\boldsymbol{p}\right)-\ell_{i}\left(\boldsymbol{z}\right)\right\} {\bf 1}_{\mathcal{E}_{3}^{i}}\geq\frac{1}{2}\tau\left\Vert \boldsymbol{p}\right\Vert ^{2}$.
Given the above majorization (\ref{eq:majorization}), this search
criterion is satisfied only if
\[
	\tau/2\leq\tau\left(1.1+\delta\right)-1.95(1-\tilde{\zeta}_{1}-\tilde{\zeta}_{2})\tau^{2}+\tau\tilde{\epsilon}
\]
or, equivalently,
\[
	\tau\leq\frac{0.6+\delta+\tilde{\epsilon}}{1.95(1-\tilde{\zeta}_{1}-\tilde{\zeta}_{2})}:=\tau_{\mathrm{ub}}.
\]
Taking $\delta$ and $\tilde{\epsilon}$ to be sufficiently small,
we see that $\tau\leq\tau_{\mathrm{ub}}\leq0.384$, provided that $\alpha_{z}^{\mathrm{lb}}\leq0.1$, $\alpha_z^{\mathrm{ub}} \geq 5$, $\alpha_h \geq 6$,
and $\alpha_{p}\geq 5$. 

Using very similar arguments, one can also show that $\frac{1}{m}\sum_{i=1}^{m}\left\{ \ell_{i}\left(\boldsymbol{z}+\tau\boldsymbol{p}\right)-\ell_{i}\left(\boldsymbol{z}\right)\right\} {\bf 1}_{\mathcal{E}_{3}^{i}}$
is minorized by a similar quadratic function, which combined with
the stopping criterion $\frac{1}{m}\sum_{i=1}^{m}\left\{ \ell_{i}\left(\boldsymbol{z}+\tau\boldsymbol{p}\right)-\ell_{i}\left(\boldsymbol{z}\right)\right\} {\bf 1}_{\mathcal{E}_{3}^{i}}\geq\frac{1}{2}\tau\left\Vert \boldsymbol{p}\right\Vert ^{2}$
suggests that $\tau$ is bounded away from 0. We omit this part for
conciseness.

\section*{Acknowledgements}

E.~C.~is partially supported by NSF under grant CCF-0963835 and by the
Math + X Award from the Simons Foundation.  Y.~C.~is supported by the
same NSF grant.  We thank Carlos Sing-Long and Weijie Su for their constructive
comments about an early version of the manuscript. 
  E.~C.~is grateful to Xiaodong Li and Mahdi Soltanolkotabi for many
  discussions about Wirtinger flows. We thank the anonymous reviewers for helpful comments.

\bibliographystyle{unsrt}
\bibliography{bibfileNonconvex}

\end{document}